\documentclass[%
reprint,
superscriptaddress,
nofootinbib,
 amsmath,amssymb,
 aps,
]{revtex4-2}

\usepackage[english]{babel}
\usepackage{xcolor}
\usepackage[linktocpage, colorlinks, citecolor=blue, linkcolor=blue, urlcolor=blue]{hyperref} 
\usepackage{enumerate}  
\usepackage{float}
\usepackage{calc}
\usepackage{dutchcal}
\usepackage{revsymb}
\usepackage{comment}
\usepackage{upgreek}
\usepackage{IEEEtrantools}
\usepackage{MyCommands}
\usepackage[normalem]{ulem}
\newtheorem{theorem}{Theorem}
\newtheorem{definition}{Definition}
\newtheorem{lemma}{Lemma}

\usepackage{amsfonts, amssymb, amsthm, mathtools, mathrsfs}
\usepackage{physics}

\usepackage{graphicx}
\usepackage{dcolumn}
\usepackage{bm}

\usepackage[nonumberlist,toc,nopostdot,style=tree]{glossaries}
\makenoidxglossaries
\newglossaryentry{memory}{
    name={memory},
    description={A component of a gravitational polarization that induces a permanent, non-oscillatory offset between two test-masses in the idealized response of a gravitational wave detector.}
}

\newglossaryentry{tensor memory}
{
    name=tensor memory,
    description={(tensor/vector/scalar) Memory associated to a tensor/vector/scalar gravitational polarization.},
    parent=memory
}

\newglossaryentry{null memory}
{
    name=null memory,
    description={Memory sourced by the energy-momentum of radiation.},
    parent=memory
}

\newglossaryentry{ordinary memory}
{
    name=ordinary memory,
    description={Memory sourced by massive bodies that are unbound in their initial or final states (or both)},
    parent=memory
}

\newglossaryentry{waves}
{
    name=waves,
    description={Field-perturbations with characteristic high frequency $f_H$.}
}

\newglossaryentry{gravitational waves}
{
    name=gravitational waves,
    description={Waves of the physical metric perturbation.},
    parent=waves
}

\newglossaryentry{gravitational field}
{
    name=gravitational field,
    description={A tensor field (other than the physical metric) that is non-minimally coupled to the physical metric. In a metric theory of gravity, a gravitational field does not couple directly to matter fields, as only the physical metric does.}
}

\newglossaryentry{radiation}
{
    name=radiation,
    description={Field-perturbation that escape to null infinity.}
}

\newglossaryentry{gravitational radiation}
{
    name=gravitational radiation,
    description={Radiation of the physical metric perturbation.},
    parent=radiation
}

\newglossaryentry{limit to null infinity}
{
    name=limit to null infinity,
    description={Leading-order behavior in the limit $r\rightarrow \infty$ at fixed $u=t-r$, in source centered coordinates $\{t,x,y,z\}$, with $r=\sqrt{x^2+y^2+z^2}$.}
}

\newglossaryentry{radiative modes}
{
    name=radiative modes,
    description={c.f.radiative degrees of freedom.}
}

\newglossaryentry{radiative degrees of freedom}
{
    name=radiative degrees of freedom,
    description={Gauge-invariant and propagating solutions to the wave equation in the limit to null infinity.}
}

\newglossaryentry{polarization}
{
    name=polarization,
    description={Property of a radiative mode that specifies the geometrical orientation of the oscillation, distinguishing it from other radiative modes.}
}

\newglossaryentry{tensor polarizations}
{
    name=tensor polarizations,
    description={(tensor/vector/scalar) Modes that transform like a tensor, a vector or a scalar under rotations about the direction of propagation.},
    parent=polarization
}

\newglossaryentry{gravitational polarizations}
{
    name=gravitational polarizations,
    description={Polarization modes of the physical metric perturbation far away from the source.},
    parent=polarization
}

\newglossaryentry{physical metric}
{
    name=physical metric,
    description={The metric that minimally couples to matter. In a metric theory of gravity, the physical metric is the only gravitational field that couples directly to matter.}
}

\renewcommand*{\glossarymark}[1]{}

\begin{document}

\title{Gravitational wave memory beyond general relativity}

\author{Lavinia Heisenberg}
\affiliation{Institut f\"{u}r Theoretische Physik, Philosophenweg 16, 69120 Heidelberg, Germany}
\affiliation{Institute for Theoretical Physics, ETH Z\"{u}rich, Wolfgang-Pauli-Strasse 27, 8093, Z\"{u}rich, Switzerland}
\author{Nicol\'{a}s Yunes}
\affiliation{Illinois Center for Advanced Studies of the Universe, Department of Physics University of Illinois at Urbana-Champaign, Champaign, IL 61801, USA}
\author{Jann Zosso}
\email{zosso.jann@bluewin.ch}
\affiliation{Institute for Theoretical Physics, ETH Z\"{u}rich, Wolfgang-Pauli-Strasse 27, 8093, Z\"{u}rich, Switzerland}
\affiliation{Illinois Center for Advanced Studies of the Universe, Department of Physics University of Illinois at Urbana-Champaign, Champaign, IL 61801, USA}
\date{\today}

\begin{abstract}
Gravitational wave memory is a nonoscillatory correction to the gravitational wave strain predicted by general relativity, which has yet to be detected. Within general relativity its dominant component known as the null memory, can be understood as arising from the backreaction of the energy carried by gravitational waves, and therefore it corresponds to a direct manifestation of the nonlinearity of the theory. In this paper, we investigate the null-memory prediction in a broad class of modified gravity theories, with the aim of exploring potential lessons to be learned from future measurements of the memory effect. Based on Isaacson's approach to the leading-order field equations, we in particular compute the null memory for the most general scalar-vector-tensor theory with second-order equations of motion and vanishing field potentials. We find that the functional form of the null memory is only modified through the potential presence of additional radiative null energy sources in the theory. We subsequently generalize this result by proving a theorem that states that the simple structure of the tensor null-memory equation remains unaltered in any metric theory whose massless gravitational fields satisfy decoupled wave equations to first order in perturbation theory, which encompasses a large class of viable extensions to general relativity. 
\end{abstract}

\maketitle


\section{\label{sec:Intro}Introduction}

The first direct measurement of gravitational waves \cite{LIGOScientific:2016aoc} opened a new window to probe the fundamental physics of gravitation. Together with subsequent observations~\cite{LIGOScientific:2018mvr,LIGOScientific:2021usb}, this has given hope to the idea that one day soon we may find an eventual glimpse beyond the currently unquestioned theoretical bedrock, general relativity (GR), which has so far passed all tests with flying colors \cite{Will:2014kxa,Yunes:2013dva,Yunes:2016jcc,LIGOScientific:2016lio,LIGOScientific:2018dkp,Nair:2019iur,LIGOScientific:2019fpa,LIGOScientific:2020tif,Perkins:2021mhb,LIGOScientific:2021sio}. A particularly interesting prediction of GR is the gravitational wave \gls{memory} effect \cite{Zeldovich:1974gvh,Christodoulou:1991cr,Blanchet:1992br,Thorne:1992sdb,PhysRevD.44.R2945} (see also \cite{Favata:2008yd,Favata:2010zu,Bieri:2013ada,Garfinkle:2022dnm}), which manifests itself as a permanent displacement of test masses within an idealized gravitational wave detector. Gravitational wave memory is an interesting effect because its dominant contribution is sourced by the outgoing gravitational radiation itself, and is therefore a direct manifestation of the nonlinearity of GR. Moreover, for compact binary coalescences, the measurable component of the memory effect is most sensitive to the merger dynamics of the system \cite{Favata:2009ii}. The recent discovery of gravitational waves therefore begs the question of whether we can use future observations of the memory effect to test the ``ability of gravity to gravitate'' and, moreover, probe the high-curvature regime of the merger of compact objects. On top of that, the manifestation of memory is directly connected with supertranslations of the Bondi-Metzner-Sachs (BMS) group of asymptotically flat spacetimes \cite{Bondi:1962px,Sachs:1962wk,FrauendienerJ,Ashtekar:2014zsa} and can also be related to Weinberg's soft-graviton theorem \cite{Weinberg_PhysRev.140.B516,Strominger:2014pwa}. Gravitational wave memory therefore also represents a window for experimental tests of the fundamental symmetries of spacetime.

Observational probes of general relativity through gravitational wave memory require both the detection of the memory effect and the modeling of beyond-GR memory effects that can be searched for in the data. On the detection front, while the memory effect has not yet been observed, the prospects of its detection are positive, both with future ground-based experiments \cite{Lasky:2016knh,Yang:2018ceq,Hubner:2019sly,Boersma:2020gxx,Ebersold:2020zah,Islam:2021old}, as well as the eagerly awaited Laser Interferometer Space Antenna (LISA) mission \cite{Favata:2009ii,Islo:2019qht,Gasparotto:2023fcg}. On the modeling front, only recently have there been several concrete efforts to compute the memory effect beyond GR. Within a post-Newtonian (PN) expansion (i.e.~an expansion of the field equations in weak fields and slow motions~\cite{Blanchet:2013haa}) of Brans-Dicke (BD) theory~\cite{Brans:1961sx,Dicke:1961gz,poisson2014gravity}, a new memory contribution originating from a dipole-dipole coupling was found~\cite{lang_compact_2014,lang_compact_2015,tahura_gravitational-wave_2021}. Another way to compute memory hinges on the structure of asymptotically flat spacetimes, in particular on the relation of memory to BMS balance laws at future null infinity~\cite{FrauendienerJ,Ashtekar:2014zsa}. BMS balance laws were recently derived in BD theory \cite{hou_gravitational_2021,tahura_brans-dicke_2021,hou_conserved_2021,hou_gravitational_2021_2}, by showing that the theory retains the same asymptotic group structure as in GR, despite its altered peeling properties. The resulting memory component was shown to match the earlier PN calculation \cite{tahura_gravitational-wave_2021}. Moreover, different aspects of the scalar memory within BD theory were investigated in~\cite{du_gravitational_2016,koyama_testing_2020}. More recently, the memory effect was also computed for dynamical Chern-Simons (dCS) gravity~\cite{Jackiw:2003pm,Alexander:2009tp} by establishing the corresponding BMS balance laws \cite{hou_gravitational_2022,Hou:2021bxz}.

In this paper, we begin a broad investigation of the gravitational wave memory in theories beyond GR, which will be crucial for the creation of future, memory-based probes of gravity. For now, we will restrict ourselves to the study of the so-called \textit{\gls{null memory}}~\cite{Bieri:2013ada}, which can be thought of as being sourced by the energy flux of the outgoing massless radiation fields themselves~\cite{Thorne:1992sdb}. This is motivated by the expectation that for binary coalescences the \textit{\gls{ordinary memory}}, associated with unbound massive objects of the system, including remnant kicks, will generally be subdominant, a statement that is confirmed both in GR \cite{Christodoulou:1991cr,Favata:2008ti} and in BD theory \cite{tahura_gravitational-wave_2021}. Moreover, to put this study on a firm footing, we choose to only consider so-called \textit{dynamical metric theories of gravity} (see Definition~\ref{DefMetricTheory}), i.e.~local and covariant theories with a symmetric and locally flat, \gls{physical metric} tensor, that is the only field that couples minimally to matter, while any additional dynamical \gls{gravitational field} exclusively couples to the metric. Moreover, we choose to only consider metric theories with massless gravitational fields and therefore for simplicity only consider operators with at least two powers of derivatives, which in particular imposes vanishing field potentials.

These requirements select a specific class of theories that are local-Lorentz invariant and preserve the geodesic deviation equation of GR, which is at the heart of current gravitational wave observations. However, it should be mentioned, that it is sometimes computationally advantageous to rewrite a metric theory through field-dependent transformations of the metric, which might result in an apparent violation of the minimal coupling requirement with respect to the redefined metric. A restriction to metric theories of gravity therefore more precisely corresponds to only attributing physical significance to gauge-invariant observables associated with the physical metric. Such an idea is, in fact, strongly suggested by the Einstein equivalence principle (see e.g. \cite{poisson2014gravity,Will:2018bme}).

Dynamical metric theories of gravity can be classified by the number and type of gravitational fields they contain, as well as by their governing evolution equations. The desire for the inclusion of additional fields derives from a theorem by Lovelock~\cite{Lovelock1969ArRMA}, which, under certain theoretically and observationally desirable assumptions, establishes that GR is the only theory that propagates two massless, spin-2 degrees of freedom. More precisely, Lovelock's theorem asserts that in four dimensions, the Einstein equations are the unique, second-order field equations of a covariant and local theory with a single metric tensor. Similarly, at the level of amplitudes, GR can be proven to be the unique gauge-invariant theory of interacting, massless, spin-2 particles with second-order equations of motion \cite{Krasnov:2014eza,Rodina:2016jyz}, confirming earlier work based on the requirement of Lorentz invariant S-matrices \cite{Weinberg:1965rz}. 

These theorems imply that a broad class of still-viable departures from Einstein gravity necessitate the introduction of new gravitational fields \cite{Heisenberg:2018vsk}. Indeed, additional gravitational degrees of freedom are, for instance, a generic prediction of string theory compactifications (see e.g. \cite{Zwiebach:1985uq,Gross:1986mw,Gross:1986iv,Moura:2006pz,Cano:2021rey}) and Kaluza-Klein reductions \cite{Appelquist:1987nr,Dereli:1990he}, in particular of higher-dimensional Lovelock gravity \cite{Charmousis:2014mia}. 
Such additional gravitational fields naturally introduce additional radiative modes to the theory, which can be excited through various mechanisms in gravitational wave sources, such as in compact binary coalescences, as investigated in detail for many theories (see e.g.~\cite{Yagi:2011xp,Barausse:2012da,Yagi:2015oca,Benkel:2016kcq,Benkel:2016rlz,Ramazanoglu:2017xbl,Silva:2017uqg,Okounkova:2017yby,Ramazanoglu:2018tig,Ramazanoglu:2019gbz,Berti:2018cxi,Silva:2020omi,Doneva:2021tvn,Annulli:2021lmn,Elley:2022ept,Doneva:2022ewd})

On general grounds, it can be expected that additional propagating degrees of freedom in the theory of gravity will greatly affect the null memory, as, in principle, null memory is sensitive to any radiative losses. How exactly does the presence of additional null sources modify null memory, and can there be other types of modifications? These are the questions we focus on and answer in this paper. We begin in Sec.~\ref{sec:Isaacson} by developing a novel way to understand the null memory from the so-called \textit{Isaacson approach}, in which the metric tensor is decomposed into a high- and a low-frequency part. Within this framework, the null memory naturally arises as a low-frequency background perturbation induced by the coarse-grained energy density carried by any high-frequency perturbations present in the theory.

We then proceed to apply this approach explicitly to the most general massless scalar-vector-tensor (SVT) gravity theory with second-order equations of motion \cite{Heisenberg:2018acv}. Section~\ref{sec:SVTFirst} first defines the theory and presents its propagating degrees of freedom and gravitational polarizations, while Sec.~\ref{sec:SVTMem} calculates the associated null memory. The beyond-GR null-memory result of this calculation [presented in Eqs.~\eqref{eq:StressEnergyFourth} and \eqref{NonLinDispMemoryModesSVT}] encompasses many popular modified-gravity theories and matches with the memory extracted from the asymptotic symmetry approach in BD theory \cite{hou_gravitational_2021,tahura_brans-dicke_2021}, as we explicitly demonstrate in Appendix~\ref{MatchToAsymptoticsBD}. Moreover, we further exemplify the generalized Horndeski result by offering the explicit correspondences to scalar-Gauss-Bonnet (sGB) gravity~\cite{Zwiebach:1985uq,Gross:1986iv,Moura:2006pz,Pani:2009wy} and double-dual Riemann gravity~\cite{deRham:2011by,Charmousis:2011ea,Charmousis:2011bf}. The beyond-GR null-memory results of Sec.~\ref{sec:SVTMem} are also presented both in terms of the transverse-traceless part of the components of the metric perturbation, as well as in terms of the elements of the spin-weighted spherical harmonic expansion of the metric perturbation. A proof of the key identity (see Eq.~\eqref{eq:RelBlanchet}), which allows us to establish the equivalence between these representations, is provided in Appendix~\ref{DerivationEq}.

We continue in Sec.~\ref{sec:GeneralFormula} by generalizing these results further through the development of a new \textit{theorem}, which proves that the structure of the null-memory equation remains in fact unchanged for theories which to \textit{first order} on a flat background satisfy decoupled massless wave equations. This assumption is met by arguably almost any massless and local Lorentz invariant metric theory of gravity, whose effective field theory (EFT) structure remains stable under the inclusion of deviations from GR. In other words, the theorem proves that the beyond-GR null memory simply receives additional contributions proportional to the energy flux of any new, propagating gravitational degree of freedom, similar to contributions from ordinary massless matter fields. 

Interestingly, the results described above suggest that one may be able to use future gravitational wave memory observations not only as a probe of the highly nonlinear regime of gravity, but also as a largely model-agnostic test for additional gravitational degrees of freedom in nature. We discuss these future research directions in the Conclusion presented in Sec.~\ref{sec:Discussion}. Henceforth, we use a $(-,+,+,+)$ metric signature and set $c=1$ throughout. At the end of the paper, we also provide a glossary that clearly defines the terminology we use in this work. 



\section{\label{sec:Isaacson}Gravitational memory in the Isaacson picture}

The gravitational null (or nonlinear) memory in GR can be understood as originating from the Einstein equations within the Landau-Lifshitz approach \cite{PhysRevD.44.R2945,Favata:2008yd}. We will argue in this section, however, that the equations at the root of the memory arise rather naturally from the alternative Isaacson viewpoint, which identifies the null memory as a \textit{low-frequency perturbation} of the background spacetime, that is parametrically distinct from the high-frequency waves that generate the memory through backreaction of the stress energy carried by the waves. Moreover, this shift of perspective provides stringent arguments for the necessity of a spacetime averaging over the radiative memory-source tensor.

For convenience, we first offer a quick review of the arguments originally brought forth by Isaacson \cite{Isaacson_PhysRev.166.1263,Isaacson_PhysRev.166.1272} (see also \cite{misner_gravitation_1973,Flanagan:2005yc,maggiore2008gravitational}) within GR in Sec.~\ref{sec:IsaacsonGR}. This will also allow us to set the notation. Toward the end of Sec.~\ref{sec:IsaacsonGR} we then explore how the memory equation naturally arises from the coarse-grained low-frequency Isaacson equation. These arguments are then generalized in Sec.~\ref{sec:IsaacsonBeyondGR} to more general dynamical metric theories of gravity with an arbitrary number of gravitational fields on top of the metric. 

\subsection{Isaacson picture and null memory in GR}\label{sec:IsaacsonGR}

Consider GR, hence a theory on a four-dimensional manifold $\mathscr M$ with Lorentzian metric $g_{\mu\nu}$ and an associated Levi-Civita connection, governed by the Einstein equations. Omitting any matter contributions for simplicity\footnote{In this work, we will only be interested in regions outside any potential matter sources of the system.} the vacuum Einstein equations can be written as $R_{\mu\nu}=0$. The first part of this subsection will closely follow the treatment in \cite{misner_gravitation_1973,maggiore2008gravitational}.

Within the Isaacson picture \cite{Isaacson_PhysRev.166.1263, Isaacson_PhysRev.166.1272}, the notion of gravitational waves propagating on an arbitrary background spacetime is given physical meaning by a clear separation of scales. More concretely, one requires the separation
\begin{equation}\label{eq:SmallParam1GR}
    f_{L}\ll f_{H}\,
\end{equation} 
between a slowly varying background of frequencies lower than $f_L$ and high-frequency perturbations, i.e. gravitational waves, of characteristic frequency $f_H$. We therefore decompose the metric as
\begin{equation}\label{eq:IsaacsonSplitGR}
   g_{\mu\nu}=\bar{g}^L_{\mu\nu}+h^H_{\mu\nu}\,,
\end{equation}
where a super- or subscript $L$ and $H$ indicate the dependence of a quantity on the low or high frequencies respectively. Because of the clear separation of physical scales, this decomposition is actually unique and does not depend on the chosen coordinates. Moreover, the metric perturbations are assumed to be of small amplitude 
\begin{equation}\label{eq:SmallParam2GR}
    |h^H_{\mu\nu}| = {\cal{O}}(\alpha)\,,
\end{equation}
where $\alpha\ll 1$ compared to the background $\bar{g}^L_{\mu\nu}= \mathcal{O}(1)$.\footnote{Locally, we can always choose a coordinate system in which the diagonal elements of $\bar{g}^L_{\mu\nu}$ are of $\mathcal{O}(1)$.} 
In conclusion, there are two small parameters at hand, namely the amplitude of the perturbations $\alpha$, as well as the ratio of frequencies $f_L/f_H$. Moreover, since the scale on which the low- and high-frequency components of the metric [Eq.~\eqref{eq:IsaacsonSplitGR}] vary is determined by $f_L$ and $f_H$ respectively, one schematically has that $\partial\bar{g}^L_{\mu\nu}\leq\mathcal{O}(f_L)$ and $\partial h^H_{\mu\nu}= \mathcal{O}(\alpha f_H)$ \cite{misner_gravitation_1973}.

At this point, a comment on the choice of separation of scales Eq.~\eqref{eq:SmallParam1GR} in frequency space is in order. In the literature (for instance in the original work by Isaacson \cite{Isaacson_PhysRev.166.1263,Isaacson_PhysRev.166.1272}) the distinction between a ``slowly varying'' background and a wavelike perturbation is usually made by imposing a clear separation in terms of scales of \textit{spacial} variations $L_L$ and $L_H$, instead of temporal variations characterized through scales of frequency $f_L$ and $f_H$. Hence, instead of Eq.~\eqref{eq:SmallParam1GR} one demands  $L_H\ll L_B$, also known as a \textit{short-wave expansion}, where the scale $L_H$ is associated with the characteristic wavelength of the wavelike perturbation. Demanding $L_H\ll L_B$ is in principle distinct from Eq.~\eqref{eq:SmallParam1GR}. This is because, while $L_H$ and $f_{H}$ are naturally related through the dispersion relation of the high-frequency wave, this is \textit{a priori} not the case for the variations $L_L$ and $f_L$ of the background, where the notion of slowly varying in time or in space are in principle unrelated. However, the two choices are interchangeable in the sense that the conclusions drawn below would remain the same if the condition $L_H\ll L_B$ was assumed instead of Eq.~\eqref{eq:SmallParam1GR}. The only difference would be that the distinction between the slowly varying background and wave perturbation would be drawn at a different level (see \cite{maggiore2008gravitational} for more details).\footnote{Yet, from the point of view of current gravitational wave detectors on Earth, it is actually the condition Eq.~\eqref{eq:SmallParam1GR} that allows for a clear distinction between gravitational waves and the background \cite{maggiore2008gravitational}.}

In light of Eq.~\eqref{eq:IsaacsonSplitGR}, we can now expand the vacuum Einstein equations. Writing the Ricci tensor as a sum of its background component and an infinite series of operators depending on increasing powers in $h^H_{\mu\nu}$, one finds,
\begin{align}\label{eq:EOMGR1}
\begin{split}
    0=&\, \phantom{}_{\s{(0)}}{R}_{\mu\nu}[\bar{g}^L] +\phantom{}_{\s{(1)}}R_{\mu\nu}[\bar{g}^L,h^H]+\phantom{}_{\s{(2)}}R_{\mu\nu}[\bar{g}^L,h^H]
    \\
    &+\sum_{i=3}^\infty\phantom{}_{\s{(i)}}R_{\mu\nu}[\bar{g}^L,h^H]\,.
\end{split}
\end{align}
Let us explain the new notation we have introduced. The quantity $\phantom{}_{\s{(N)}}O[A,B]$ denotes the expanded operator $O$ at $N$th order in the perturbation field $B$, as computed with the background $A$. For example,
\begin{align}
    _{\s{(1)}}R_{\mu\nu\rho\sigma}[\bar{g},H]=&-\frac{1}{2}\Big(\bar{\nabla}_\sigma\bar\nabla_\mu H_{\nu\rho}+\bar\nabla_\rho\bar\nabla_\nu H_{\mu\sigma}\nonumber\\
    &-\bar\nabla_\sigma\bar\nabla_\nu H_{\mu\rho}-\bar\nabla_\rho\bar\nabla_\mu H_{\nu\sigma}\label{eq:RiemmanFirstOrder}\\
    &+\phantom{}_{\s{(0)}}R_{\mu\gamma\rho\sigma}[\bar{g}]H\ud{\gamma}{\nu}-\phantom{}_{\s{(0)}}R_{\nu\gamma\rho\sigma}[\bar{g}]H\ud{\gamma}{\mu}\Big)\,,\nonumber
\end{align}
for some tensor perturbation $H_{\mu\nu}$ and background metric $\bar{g}_{\mu\nu}$, such that also
\begin{equation}\label{eq:RicciFirstOrder}
\begin{split}
    \phantom{}_{\s{(1)}}R_{\nu\sigma}[\bar{g},H]=&-\frac{1}{2}\bar{g}^{\mu\rho}\Big(\bar\nabla_\sigma\bar\nabla_\mu H_{\nu\rho}+\bar\nabla_\rho\bar\nabla_\nu H_{\mu\sigma}\\
    &-\bar\nabla_\sigma\bar\nabla_\nu H_{\mu\rho}-\bar\nabla_\rho\bar\nabla_\mu H_{\nu\sigma}\Big)\,.
\end{split}
\end{equation}

Let us now estimate the size of the operators that appear in Eq.~\eqref{eq:RicciFirstOrder}. The second operator in the sum in Eq.~\eqref{eq:EOMGR1}, namely, $\phantom{}_{\s{(1)}}R_{\mu\nu}[\bar{g}^L,h^H]$ is clearly maximally of $\mathcal{O}\left(\alpha f_H^2\right)$. Indeed, due to the separation of scales given by Eq.~\eqref{eq:SmallParam1GR}, any contribution for which a derivative acts on a background metric instead of the perturbation field will be parametrically suppressed. Similarly, the second perturbation of the Ricci tensor $\phantom{}_{\s{(2)}}R_{\mu\nu}$  (see e.g. \cite{misner_gravitation_1973,maggiore2008gravitational} for the exact form) also only involves two derivative operators, as any term in Eq.~\eqref{eq:EOMGR1}, and therefore $\phantom{}_{\s{(2)}}R_{\mu\nu}[\bar{g}^L,h^H]=\mathcal{O}\left(\alpha^2 f_H^2\right)$. Except for the background term $\phantom{}_{\s{(0)}}R_{\mu\nu}$, whose order we will address later in the discussion, every term $\phantom{}_{\s{(M)}}R_{\mu\nu}$ in Eq.~\eqref{eq:EOMGR1}, maximally of $\mathcal{O}\left(\alpha^M f_H^2 \right)$, is therefore suppressed compared to the leading order terms in $\phantom{}_{\s{(N)}}R_{\mu\nu}$ of $\mathcal{O}\left(\alpha^N f_H^2 \right)$, whenever $M>N$, and we can rewrite Eq.~\eqref{eq:EOMGR1} as
\begin{align}
\begin{split}\label{eq:EOMGR2}
    0=&\,\underset{?}{\underbrace{\phantom{}_{\s{(0)}}R_{\mu\nu}[\bar{g}^L]}}+\underset{\mathcal{O}\left(\alpha f_H^2\right)}{\underbrace{\phantom{}_{\s{(1)}}R_{\mu\nu}[\bar{g}^L,h^H]}}+\underset{\mathcal{O}\left(\alpha^2f_H^2\right)}{\underbrace{\phantom{}_{\s{(2)}}R_{\mu\nu}[\bar{g}^L,h^H]}}\\
    &+\mathcal{O}\left(\alpha^3f_H^2\right)\,.
\end{split}
\end{align}

To continue, we want to solve Eq.~\eqref{eq:EOMGR2} to leading-order in our bivariate expansion, by also imposing the decomposition of this equation into low- and high-frequency parts, thus only equating terms with comparable frequency behavior. This can be viewed as performing a multiple-scale analysis of this problem. Such a decomposition leads to the following key observations. First, clearly the background term $\phantom{}_{\s{(0)}}R_{\mu\nu}$, which is independent of the high-frequency field, only contains low-frequency modes. Similarly, any terms coming from $\phantom{}_{\s{(1)}}R_{\mu\nu}$, which are linear in high-frequency perturbation fields, will only contribute to the high-frequency equations. On the other hand and quite crucially, the quantities $\phantom{}_{\s{(2)}}R_{\mu\nu}$ at second order in perturbation fields will contain contributions both at the level of $f_H$, as well as at the background scales $f_L$, since two different high-wave-vector modes can combine to form a low-frequency contribution. 

A natural and practical way to single out the low-frequency part of an expression is to perform an average $\langle...\rangle$ over a spacetime region\footnote{In this work, we will primarily be interested in the asymptotic region far from any matter source. In that case, a mere spatial average over several wavelengths or a temporal average over several periods would actually suffice (see \cite{maggiore2008gravitational} for more details), but for the sake of generality, we choose here to work with a spacetime average.}  with averaging kernel of characteristic scale $f_L\ll f_{\rm av}\ll f_H$. For instance, the low-frequency part of $\phantom{}_{\s{(2)}}R_{\mu\nu}$, which we will denote as $\left[\phantom{}_{\s{(2)}}R_{\mu\nu}\right]^L$, satisfies $\left[\phantom{}_{\s{(2)}}R_{\mu\nu}\right]^L=\big\langle\phantom{}_{\s{(2)}}R_{\mu\nu}\big\rangle$. The high-frequency part can then be determined through $\left[\phantom{}_{\s{(2)}}R_{\mu\nu}\right]^H= \phantom{}_{\s{(2)}}R_{\mu\nu}-\big\langle \phantom{}_{\s{(2)}}R_{\mu\nu}\big\rangle$. One can also think about the averaging procedure as a type of coarse graining, which integrates out over the small scales we are not interested in \cite{misner_gravitation_1973,maggiore2008gravitational}. In the following, such an averaging will also turn out to be crucial to promote the energy-momentum flux to a gauge-invariant, and therefore, physically meaningful, quantity. Several averaging schemes could be applied, but their details are not essential, as long as the following properties hold \cite{Isaacson_PhysRev.166.1272,misner_gravitation_1973,Flanagan:2005yc,maggiore2008gravitational,Zalaletdinov:2004wd}:\footnote{Strictly speaking, requiring these properties introduces an error that would be relevant at higher orders in perturbation theory, e.g.~boundary terms arising from integration over a box of finite size do not vanish completely. We will, however, not be concerned by any of these higher order corrections in this paper.} (I) the average of an odd number of short-wavelength quantities vanishes; (II) total derivatives of tensors average out to zero; (III) as a corollary of the above, integration by parts of covariant derivatives are allowed. See also~\cite{Stein:2010pn} for a discussion of these three properties. 

Considering all of the above while solving Eq.~\eqref{eq:EOMGR2} to leading order, we arrive at the following set of equations
\begin{align}
    \phantom{}_{\s{(0)}}R_{\mu\nu}[\bar{g}^L]&=-\big\langle \phantom{}_{\s{(2)}}R_{\mu\nu}[\bar{g}^L,h^H]\big\rangle\,,   \label{eq:EOMISGR}\\
    \phantom{}_{\s{(1)}}R_{\mu\nu}[\bar{g}^L,h^H]&=0\,,\label{eq:EOMIISGR}
\end{align}

The first Equation~\eqref{eq:EOMISGR} is the leading-order, low-frequency equation, which we can interpret as telling us that the background curvature is modified by the backreaction of the coarse-grained operator $\big\langle \phantom{}_{\s{(2)}}R_{\mu\nu}\big\rangle$ of high-frequency gravitational waves. This is because, in the absence of other matter sources, the order of $\phantom{}_{\s{(0)}}R_{\mu\nu}$ is simply determined by the leading-order gravitational wave operator that contributes to the low-frequency equations. Quite naturally, the right-hand side of Eq.~\eqref{eq:EOMISGR} is therefore interpreted as the energy-momentum (pseudo)tensor of gravitational waves \cite{Isaacson_PhysRev.166.1263,Isaacson_PhysRev.166.1272,misner_gravitation_1973,maggiore2008gravitational}
\begin{equation}
    \big\langle \phantom{}_{\s{(2)}}R_{\mu\nu}[\bar{g}^L,h^H]\big\rangle\propto \phantom{}_{\s{(2)}}t^{\st{GR}}_{\mu\nu}\,.
\end{equation}
On the other hand, Eq.~\eqref{eq:EOMIISGR} is the leading-order, high-frequency equation, and it simply corresponds to a propagation equation for the leading-order gravitational waves $h^H$.

This concludes the review of the basic textbook introduction to the Isaacson picture, which we can now use to construct a well-defined equation that will naturally give rise to null memory in GR. Let us then make the additional assumption that the low-frequency background metric $\bar{g}^L_{\mu\nu}$ can be further split into a time-independent piece $\bar{\eta}_{\mu\nu}$ of amplitude $\mathcal{O}(1)$ and a low-frequency correction $\delta h^L_{\mu\nu}$. The time-independent metric $\bar{\eta}_{\mu\nu}$ is assumed to solve the vacuum Einstein equations, but it need not be the Minkowski metric. In fact, we chose here to work with a time-independent background solution $\bar{\eta}_{\mu\nu}$ only for the sake of simplicity, but a generalization to a background that varies slowly at ${\cal{O}}(\tilde{f}_L\lesssim f_L)$ would not change the results obtained below. The background metric is then split via
\begin{equation}\label{eq:IsaacsonSplitBackgroundGR}
   \bar{g}^L_{\mu\nu}=\bar{\eta}_{\mu\nu}+\delta h^L_{\mu\nu}\,,
\end{equation}
where we assume that
\begin{equation}\label{eq:SmallParam3GR}
    |\delta h^L_{\mu\nu}|=\mathcal{O}(\beta)\,,
\end{equation}
and $\beta \ll 1$.
Up to order $\mathcal{O}(\alpha^2\beta f_H^2)$ and $\mathcal{O}(\alpha\beta f_H^2)$, Eqs.~\eqref{eq:EOMISGR} and \eqref{eq:EOMIISGR}, respectively become
\begin{align}
    \phantom{}_{\s{(0)}}R_{\mu\nu}[\bar{\eta}]&=0\,, \label{eq:EOMBackgroundS2GR}\\
    \phantom{}_{\s{(1)}}R_{\mu\nu}[\bar{\eta},\delta h^L]&=-\big\langle \phantom{}_{\s{(2)}}R_{\mu\nu}[\bar{\eta},h^H]\big\rangle\,,   \label{eq:EOMIS2GR}\\
    \phantom{}_{\s{(1)}}R_{\mu\nu}[\bar{\eta},h^H]&=0\,.\label{eq:EOMIIS2GR}
\end{align}

One can view $\bar{\eta}_{\mu\nu}$ and $\delta h^L_{\mu\nu}$ as the homogeneous and particular solutions of Eq.~\eqref{eq:IsaacsonSplitBackgroundGR} respectively, where $\delta h^L_{\mu\nu}$ is the low-frequency component determined by the backreaction of the energy-momentum carried by gravitational waves. In general, the background geometry $\bar{\eta}_{\mu\nu}$ can be viewed as being sourced by some matter field content outside the region of interest.

Even though we have introduced here a third small parameter $\beta$, in addition to $f_L/f_H$ and $\alpha$, these parameters cannot all be independent of each other. Since $\delta h^L_{\mu\nu}$ is determined through the coarse-grained backreaction of the high-frequency gravitational wave perturbations, the scale of $\beta$ is determined through Eq.~\eqref{eq:EOMIS2GR} to be
\begin{equation}\label{eq:BetaParam}
    \beta\sim \alpha^2\frac{f_H^2}{f_L^2}\,,
\end{equation}
because $\phantom{}_{\s{(1)}}R_{\mu\nu}[\bar{\eta},\delta h^L] = {\cal{O}}(\beta f_L^2)$ and $\phantom{}_{\s{(2)}}R_{\mu\nu}[\bar{\eta},h^H] = {\cal{O}}(\alpha^2 f_H^2)$, hence, the requirement that $\beta \ll 1$ imposes a hierarchy between the two expansion parameters $f_L/f_H$ and $\alpha$, namely\footnote{Note that the original work by Isaacson \cite{Isaacson_PhysRev.166.1263} explicitly only considers the situation in which $\alpha\sim f_L/f_H$.}
\begin{equation}
\label{eq:hierarchy-inequ}
    \alpha\ll \frac{f_L}{f_H}\,.
\end{equation}

The scenario captured by the assumption in Eq.~\eqref{eq:IsaacsonSplitBackgroundGR} also encompasses the asymptotic region of an asymptotically flat spacetime around an isolated source. More precisely, in the \gls{limit to null infinity}, the metric satisfies $g_{\mu\nu}=\eta_{\mu\nu}+\mathcal{O}(1/r)$ in a set of coordinates $\{t,x,y,z)\}$\footnote{More formally, we choose to work in the asymptotic rest frame of the source \cite{misner_gravitation_1973,Thorne:1980ru}. Since we will neglect any linear or ordinary memory contributions, and henceforth also any remnant kicks, we are not concerned with more rigorous definitions of BMS rest frames in this paper.} and where $r\equiv \sqrt{x^2+y^2+z^2}$ is the radial source-centered coordinate. In the following, we will frequently employ the loose terminology ``limit to (future) null infinity'' to describe the asymptotic behavior of a quantity as $r\rightarrow\infty$ at a fixed asymptotic retarded time $u\equiv t-r$ and therefore take the limit up to the first nontrivial term in the $r$ expansion.\footnote{This abuse of terminology mainly arises because of the precise formulation of asymptotic flatness defined though the existence of a conformal completion ($\hat{\mathscr M}=\mathscr M \cup \scri,\hat{g}_{\mu\nu}$) where $\hat{g}_{\mu\nu}=\Omega^2g_{\mu\nu}$, with boundary $\scri$ of topology $S^2\times R$ at null infinity, where $\Omega=0$ and $\hat{\nabla}_\mu\Omega\neq 0$ (see e.g. \cite{Geroch:1977jn,Ashtekar:1981bq,Ashtekar:2014zsa,WaldBook,DAmbrosio:2022clk}). Thus, the conformal completion is a manifold with actual boundary at null infinity, whose non-trivial limit corresponds to leading-order terms in the expansion in $r$ on the physical spacetime.}

With this in mind, we can now make a few observations about the metric in the limit to null infinity. The background metric $\bar{\eta}_{\mu\nu}$ defined in Eq.~\eqref{eq:IsaacsonSplitBackgroundGR} asymptotes to the Minkowksi metric $\eta_{\mu\nu}$. Moreover, as we will explicitly see in Sec.~\ref{sec:SVTMem}, the low-frequency perturbation to the background metric $\delta h^L_{\mu\nu}$ naturally describes the gravitational null memory in GR. This is because in a suitable gauge, Eq.~\eqref{eq:EOMIS2GR} reduces to
\begin{equation}
    \Box \delta h^L_{\mu\nu}\propto \phantom{}_{\s{(2)}}t^{\st{GR}}_{\mu\nu}[\eta,h^H]\,,
\end{equation}
where
\begin{equation}
    \phantom{}_{\s{(2)}}t^{\st{GR}}_{\mu\nu}[\eta,h^H]\propto \big\langle\partial_\mu h^H_{\alpha\beta}\partial_\nu h_H^{\alpha\beta}\big\rangle\,.
\end{equation}
This is the equation whose solution $\delta h^L_{\mu\nu}$ yields a gravitational memory contribution that is sourced by the coarse-grained energy-momentum carried by the high-frequency gravitational wave $h^H_{\mu\nu}$ \cite{PhysRevD.44.R2945,Favata:2008yd}.

Since $\eta_{\mu \nu}$ is time independent and $\bar{g}^L_{\mu\nu}$ varies on the scale $f_L$, the latter corresponds to the characteristic frequency of the gravitational null memory in GR. Interestingly, Eq.~\eqref{eq:BetaParam} implies that the null memory is enhanced by an additional factor of $f_H^2/f_L^2$ compared to terms at $\mathcal{O}\left(\alpha^2\right)$. This parallels the observation in \cite{Favata:2008yd}, that, while the hereditary time integral of oscillatory corrections scales with the orbital timescale, the memory scales with the radiation-reaction timescale instead. However, note that an estimate of $\beta$ above does not represent a faithful estimate of the amplitude of the memory, mostly because one has to take into account that the memory correction also falls off as $1/r$, which boosts its amplitude up to roughly $10\%$ of the (nonmemory) oscillatory signal (see also \cite{Thorne:1992sdb}).

Moreover, we want to point out that the inequality in Eq.~\eqref{eq:hierarchy-inequ} is indeed satisfied in astrophysical systems of interest to ground- or space-based gravitational wave detectors. The amplitude of gravitational waves is typically of ${\cal{O}}(10^{-22})$ and ${\cal{O}}(10^{-19})$ for ground- and space-based detectors, respectively. The high-frequency $f_H$ can be approximated with the frequency at merger, which for a $10^2 M_\odot$ and $10^5 M_\odot$ (total) mass binary is approximately $f_H \approx 10^2$ Hz and $f_H \approx 10^{-1}$Hz for ground- and space-based detectors, respectively. The characteristic frequency of the memory can be estimated from the inverse of the rise time of the memory at merger, which is approximately $10^2 M$, where $M$ is the total mass of the binary (see e.g.~Fig.~1 in~\cite{Favata:2009ii}). The low frequency $f_L$ for a $10^2 M_\odot$ and $10^5 M_\odot$ (total) mass binary is then $f_L \approx 10$ Hz and $f_L \approx 10^{-2}$ Hz for ground- and space-based detectors, respectively.

We want to conclude this subsection by defining some terminology that will be important to distinguish two similar, but in truth very different, notions used in this paper (see also \cite{Isaacson_PhysRev.166.1263}). The first one is the notion
of high-frequency perturbations as defined in this section, which are introduced in contrast to the slowly varying field content of the background. We, henceforth choose to use the terminology \textit{\gls{waves}} for such perturbation fields. Hence, \gls{gravitational waves} denote the high-frequency perturbations of the physical metric. On the other hand, we want to reserve the term \textit{\gls{radiation}} to denote fields that escape to null infinity, to which we can associate a power that is irreversibly carried away from a localized source. In particular, the polarization modes of the physical metric in the limit to null infinity will be called \gls{gravitational radiation}. Interestingly, such a distinction is important in order to properly address the manifestation of gravitational memory. This is because, in light of the discussion above, null memory is a low-frequency perturbation, and hence, according to the preceding definition, it is \textit{not} a ``gravitational wave.'' On the other hand, memory does escape to null infinity as a component of the polarization modes of the physical metric (as a direct consequence of the definition of memory, which will become clear in Sec.~\ref{sec:SVTMem}), and therefore it is part of the gravitational radiation emitted by the system.

\subsection{Isaacson picture and null memory beyond GR}\label{sec:IsaacsonBeyondGR}

We now want to generalize the above arguments to more generic dynamical metric theories of gravity with spacetime $(\mathscr M,g)$, where the metric $g_{\mu\nu}$ couples minimally to matter and nonminimally to a number of additional dynamical gravitational fields, which we will collectively refer to as $\Psi$. Moreover, as already mentioned in the Introduction, for simplicity we will further restrict ourselves to theories, for which each term in the action carries at least two derivative operators. 
While the main steps will remain similar to the above treatment within GR, there are a few important differences that we will highlight below.

A generic dynamical metric theory is governed by a set of vacuum field equations, which we schematically denote as $\mathcal{G}_{\mu\nu}=0$ for the metric equations and $\mathcal{J}=0$ for all other field equations.
We will primarily be interested in the equations for the metric, hence the generalized vacuum Einstein equations $\mathcal{G}_{\mu\nu}=0$. Therefore, henceforth, we mainly only explicitly write down this equation in the derivation that follows, but in principle, the considerations below also hold for any of the additional field equations. 

Let us again postulate a separation of frequency scales 
\begin{equation}\label{eq:SmallParam1}
    f_{L}/f_{H}\ll 1\,,
\end{equation} 
and assume that it is possible to this time not only decompose the metric, but also any other gravitational field into a low-frequency background and high-frequency perturbation component
\begin{equation}\label{eq:IsaacsonSplit}
   g_{\mu\nu}=\bar{g}^L_{\mu\nu}+h^H_{\mu\nu}\,,\quad \Psi=\bar{\Psi}^L+\Psi^H\,.
\end{equation}
For simplicity, we will assume that all high-frequency perturbations can be captured by the same small expansion parameter $\alpha$, such that
\begin{equation}\label{eq:SmallParam2}
    |h^H_{\mu\nu}|\,,\;\;|\Psi^H| =\cal{O}( \alpha)\,,
\end{equation}
where $\alpha \ll 1$, although in practice of course the amplitudes of each field wave might be different. In general, one can assume that the gravitational wave tensor perturbations will dominate the expansion of the field equations. This is so if one assumes that deviations from GR must remain small. Here we will assume that this is the case, and therefore, we ignore the situation in which there are waves of the additional fields $\Psi$ but no gravitational waves present in the spacetime.\footnote{Note that if the waves of a particular additional gravitational field are parametrically smaller in amplitude than the gravitational waves, then it might be that the effect of that wave only comes in at higher orders and would therefore drop out of our leading-order analysis.}

Let us start by again expanding the field equations in the perturbation fields
\begin{equation}\label{eq:EOMI}
\begin{split}
    0=\phantom{}_{\s{(0)}}\mathcal{G}_{\mu\nu}+\phantom{}_{\s{(1)}}\mathcal{G}_{\mu\nu}&+\phantom{}_{\s{(2)}}\mathcal{G}_{\mu\nu}+\sum_{i=3}^\infty\phantom{}_{\s{(i)}}\mathcal{G}_{\mu\nu}\,.
\end{split}
\end{equation}
where, compared to Eq.~\eqref{eq:EOMGR1}, the respective arguments are given by replacing $\bar{g}^L\rightarrow\{\bar{g}^L,\Psi^L\}$ and $h^H\rightarrow\{h^H,\Psi^H\}$.
On very general grounds, we can still expect that the background terms $\phantom{}_{\s{(0)}}\mathcal{G}_{\mu\nu}$ will only contain low-frequency modes, 
while the terms in $\phantom{}_{\s{(1)}}\mathcal{G}_{\mu\nu}$, which are linear in perturbation fields, will only contribute to the high-frequency equations.
In contrast to GR, however, the terms in this expansion might involve operators with more than two derivatives. This implies that for a generic metric theory of gravity, the order counting needs to be modified. In particular, it is crucial to ensure that in the expansion in Eq.~\eqref{eq:EOMI}, the higher-order terms in $\alpha$ remain negligible, such that we can still solve the equations of motion order by order. As we will now argue, under certain assumptions, this requirement can be met, such that the beyond-GR complications do not affect the main character of the results in Eqs.~\eqref{eq:EOMISGR} and \eqref{eq:EOMIISGR}.

In the following, we will separately consider theories with and without any higher-order derivative interactions. 
Let us first consider the latter, and work with the metric $g_{\mu\nu}$ and any additional fields $\Psi$ in dimensionless natural units. Since by assumption any term involves exactly two derivative operators, at each order $N$ in $\alpha$, there might be terms of $\mathcal{O}\left(\alpha^N f_H^2\right)$, $\mathcal{O}\left(\alpha^N f_Hf_L\right)$ or $\mathcal{O}\left(\alpha^N f_L^2\right)$, just as in GR. Because of the assumption in Eq.~\eqref{eq:SmallParam1}, it is clear that for each $N$, the term of $\mathcal{O}\left(\alpha^N f_H^2\right)$ dominates. Thus, the order counting remains the same as in GR [cf. Eq.~\eqref{eq:EOMGR2}] and the arguments in Sec.~\ref{sec:IsaacsonGR} go through.

Let us now consider beyond-GR theories with field equations that have interactions with more than two derivatives in the action. By dimensional analysis, each such term must be multiplied by a corresponding power of a dimensionfull coupling, which we will collectively denote as $\epsilon$. Theories involving such higher-order derivative interactions are generally further subclassified in two different types: (i) theories with higher-derivative interactions that still admit equations of motion at second order in derivative operators (hence only up to two derivative operators per field); and (ii) theories whose equations of motion are higher order in derivatives.

A restriction to theories with second-order field equations is usually motivated by the Ostrogradsky theorem \cite{Ostrogradsky:1850fid}. This theorem states that, quite generally, theories with field equations involving more than two time derivatives per field possess ghost instabilities. Such instabilities are rooted in a Hamiltonian that is unbound from below. Yet, solutions to theories with higher-order field equations can nevertheless be stable within the so-called \textit{small-coupling approximation} (see e.g. \cite{Yunes:2013dva}). More precisely, the small-coupling approximation hinges on an effective field theory (EFT) point of view (see e.g. \cite{Weinberg:2008hq,Donoghue:2012zc,Porto:2016pyg,Endlich:2017tqa}), in which the higher-order derivative terms can be viewed as (quantum) corrections, which may leave traces in observables within the regime of validity of the EFT. This is indeed the case, for example, in dCS gravity~\cite{Jackiw:2003pm,Alexander:2009tp}, which propagates a ghost degree of freedom~\cite{Motohashi:2011ds} that can be eliminated if one imposes the small-coupling approximation~\cite{Yunes:2013dva}.

If one assumes that a quantization of the theory is inevitable, however, then the first type of theories described above also needs to be understood as an EFT, for which quantum corrections are required to remain under control (see, for instance, \cite{Pirtskhalava:2015nla,Heisenberg:2020cyi} for a study of quantum stability of subclasses of covariantized Galileon theories). Moreover, an EFT point of view for any theory involving higher-order derivative interactions is further motivated by the observation that such theories are generally believed to only be well-posed in the \textit{weakly-coupled regime} (see e.g. \cite{Kovacs:2020pns,Ripley:2022cdh}).  Horndeski theories can only be shown to be well posed in this regime \cite{Kovacs:2020ywu}, while the well posedness of other higher-order derivative theories still needs to be proven. By ``well posedness'' of a theory, here we mean  more precisely that the hyperbolic partial differential (field) equations have a well-posed initial value problem and can be meaningfully evolved in numerical simulations, which is guaranteed if the system is strongly hyperbolic \cite{Ripley:2022cdh}. Such a restriction to the weakly coupled regime effectively translates into the requirement that the couplings of terms involving more than two derivative operators remain small compared to the highest-frequency scale $f_H$.

The arguments above therefore naturally suggest that one should only apply such theories in regimes in which the coupling $\epsilon$ of any term in the action that involves more than two derivatives is treated as small when compared to the energy scale of the perturbations. Such a small-coupling requirement then enables us to include such higher-derivative theories in our analysis. More precisely, we demand that for any scale of energy $f_H$ below the scale of validity of the EFT, we have
\begin{equation}\label{eq:SmallParam4}
    \epsilon f_H \lesssim 1\,.
\end{equation}
This directly implies that the leading-order terms in the expansion in Eq.~\eqref{eq:EOMI} are again dominated by $\alpha$, in the sense that any operator $\phantom{}_{\s{(M)}}\mathcal{G}_{\mu\nu}$, maximally of $\mathcal{O}\left(\alpha^M f_H^2 (f_H\epsilon)^i\right)$ for some $i\geq 0$, remains subdominant compared to the leading-order terms in $\phantom{}_{\s{(N)}}\mathcal{G}_{\mu\nu}$ of $\mathcal{O}\left(\alpha^N f_H^2 (f_H\epsilon)^j\right)$, whenever $M>N$.

The above is all we need to establish the analog to Eqs.~\eqref{eq:EOMISGR} and \eqref{eq:EOMIISGR} for the case of generic metric theories. The arguments above imply that we may solve the generalized field equations in Eq.~\eqref{eq:EOMI} to leading order in $\alpha$, while ensuring that any higher-order correction remains subdominant. Moreover, because of the split between low- and high-frequency equations, which we insist is crucial in this case, we can consider the leading-order contributions of the low- and high-frequency equations separately. Just as in GR, these equations will be given by the low-frequency terms from $\big\langle \phantom{}_{\s{(2)}}\mathcal{G}_{\mu\nu}\big\rangle$ and $\phantom{}_{\s{(0)}}\mathcal{G}_{\mu\nu}$ [determined through backreaction of the coarse-grained contribution $\big\langle \phantom{}_{\s{(2)}}\mathcal{G}_{\mu\nu}\big\rangle$, at most of $\mathcal{O}\left(\alpha f_H^2\right)$] and the high-frequency contribution $\phantom{}_{\s{(1)}}\mathcal{G}_{\mu\nu}$ at most of $\mathcal{O}\left(\alpha^2 f_H^2\right)$ respectively, while all other contributions will be of higher order. Thus,
\begin{align}
    &\phantom{}_{\s{(0)}}\mathcal{G}_{\mu\nu}[\{\bar g^L,\bar\Psi^L\}]=-\big\langle \phantom{}_{\s{(2)}}\mathcal{G}_{\mu\nu}[\{\bar g^L,\bar\Psi^L\},\{h^H,\Psi^H\}]\big\rangle\,,   \label{eq:EOMIS}\\
    &\phantom{}_{\s{(1)}}\mathcal{G}_{\mu\nu}[\{\bar g^L,\bar\Psi^L\},\{h^H,\Psi^H\}]=0\,,\label{eq:EOMIIS}
\end{align}
while now
\begin{equation}\label{eq:RHSToEM}
    \big\langle \phantom{}_{\s{(2)}}\mathcal{G}_{\mu\nu}\big\rangle\propto \phantom{}_{\s{(2)}}t_{\mu\nu}\,,
\end{equation}
captures the effective energy-momentum contribution of all of the high-frequency perturbations. For later reference, we explicitly also state here the leading-order low- and high-frequency equations for the additional field equations $\mathcal{J}=0$, which similarly read
\begin{align}
    &\phantom{}_{\s{(0)}}\mathcal{J}[\{\bar g^L,\bar\Psi^L\}]=-\big\langle \phantom{}_{\s{(2)}}\mathcal{J}[\{\bar g^L,\bar\Psi^L\},\{h^H,\Psi^H\}]\big\rangle\,,\label{eq:EOMISPsi}\\
    &\phantom{}_{\s{(1)}}\mathcal{J}[\{\bar g^L,\bar\Psi^L\},\{h^H,\Psi^H\}]=0\,.\label{eq:EOMIISPsi}
\end{align}

Following the GR derivation, let us further decompose the background metric and background field into time-independent pieces $\bar{\eta}_{\mu \nu}$ and $\bar{\Psi}_0$, and time-dependent pieces (with characteristic frequency $f_L$) $\delta h_{\mu \nu}^L$ and $\delta \Psi^L$, namely  
\begin{equation}\label{eq:IsaacsonSplitBackground}
   \bar{g}^L_{\mu\nu}=\bar{\eta}_{\mu\nu}+\delta h^L_{\mu\nu}\,, \quad \bar{\Psi}^L=\bar{\Psi}_0+\delta \Psi^L
\end{equation}
where the set $\{\bar{\eta}_{\mu\nu},\bar{\Psi}_0\}\equiv \bar{\eta}_0$ solves the background equations of motion $\phantom{}_{\s{(0)}}\mathcal{G}_{\mu\nu}[\bar{\eta}_0]=0$, and where
\begin{equation}\label{eq:SmallParam3}
    |\delta h^L_{\mu\nu}|\,,\;\;|\delta \Psi^L| = {\cal{O}}(\beta)\,,
\end{equation}
where $\beta \ll 1$. Equations~\eqref{eq:EOMIS}, \eqref{eq:EOMIIS}, \eqref{eq:EOMISPsi} and \eqref{eq:EOMIISPsi} to leading order [i.e.~to $\mathcal{O}\left(\alpha^2f_H^2\right)$ and $\mathcal{O}\left(\alpha f_H^2\right)$, respectively] become
\begin{align}
    _{\s{(1)}}\mathcal{G}_{\mu\nu}[\bar{\eta}_0,\{\delta h^L,\delta\Psi^L\}]&=-\big\langle \phantom{}_{\s{(2)}}\mathcal{G}_{\mu\nu}[\bar{\eta}_0,\{h^H,\Psi^H\}]\big\rangle\,,   \label{eq:EOMIS2}\\
    _{\s{(1)}}\mathcal{J}[\bar{\eta}_0,\{\delta h^L,\delta\Psi^L\}]&=-\big\langle \phantom{}_{\s{(2)}}\mathcal{J}[\bar{\eta}_0,\{h^H,\Psi^H\}]\big\rangle\,,   \label{eq:EOMIS2Psi}\\
    \phantom{}_{\s{(1)}}\mathcal{G}_{\mu\nu}[\bar{\eta}_0,\{h^H,\Psi^H\}]&=0\,,\label{eq:EOMIIS2}\\
    \phantom{}_{\s{(1)}}\mathcal{J}[\bar{\eta}_0,\{h^H,\Psi^H\}]&=0\,.\label{eq:EOMIIS2Psi}
\end{align}
Observe that the relation in Eq.~\eqref{eq:BetaParam} $\beta\sim \alpha^2f_H^2/f_L^2$, still holds [at least for the tensor perturbations for which we know \textit{a priori} that in the presence of gravitational waves the right-hand side of Eq.~\eqref{eq:EOMIS2} does not vanish].

Let us close this section by commenting on the number of derivatives involved in each term of the equations above. Note that because of the small-coupling assumption in Eq.~\eqref{eq:SmallParam4}, the number of derivative operators involved in each leading-order term is no longer restricted. However, for theories satisfying second-order equations of motion [hence, also including theories of type (i) above[] and in the asymptotic region of an asymptotically flat spacetime (which we will be interested in, in this work), any operator at first order in perturbation fields $\phantom{}_{\s{(1)}}\mathcal{G}_{\mu\nu}[\bar{\eta}_0,\{h^H,\Psi^H\}]$ and $\phantom{}_{\s{(1)}}\mathcal{J}[\bar{\eta}_0,\{h^H,\Psi^H\}]$ involves exactly two derivatives. This is because the maximum number of derivatives per field is two, while any term involving a derivative acting on the Minkowski background will vanish. Moreover, anticipating the results in Sec.~\ref{sec:GeneralFormula} below, we will show that in the limit to null infinity, this implies that also the low-frequency contribution $\big\langle \phantom{}_{\s{(2)}}\mathcal{G}_{\mu\nu}[\bar{\eta}_0,\{h^H,\Psi^H\}]\big\rangle$ will only involve two derivatives. This statement will actually also hold for many theories of type (ii). However, for some of them, a stronger assumption than Eq.~\eqref{eq:SmallParam4} is needed, namely that
\begin{equation}\label{eq:SmallParam5}
    \epsilon f_H \ll 1\,,
\end{equation}
which assures that any term in $\phantom{}_{\s{(1)}}\mathcal{G}_{\mu\nu}[\bar{\eta}_0,\{h^H,\Psi^H\}]$ involving more than two derivatives is of higher order. This statement will be further discussed in Sec.~\ref{sec:ScopeOfClaim}.



\section{\label{sec:SVTFirst} The Basics of massless SVT Gravity}

As a first application of the approach to the Isaacson picture discussed in the previous section, we will use that framework to calculate the backreaction of the energy-momentum carried by wave perturbations in the limit to null infinity for a concrete set of metric theories beyond GR. This will not only confirm that the induced, low-frequency metric perturbation indeed corresponds to the null memory, but it will also provide a formula for the null memory in theories for which the memory has not yet been explored. For now, we choose to restrict our study to the leading-order effects in the asymptotic region of asymptotically flat spacetimes, mainly because this is the simplest situation from which observationally relevant consequences can be drawn.

Let us therefore focus on spacetimes that are asymptotically flat and expand the metric and any additional gravitational fields around a Minkowski background 
\begin{equation}\label{eq:AsymFlat}
    g_{\mu\nu}=\eta_{\mu\nu}+\mathcal{O}(1/r)\,,\quad\Psi=\Psi_0+\mathcal{O}(1/r)\,,
\end{equation}
Moreover, we want to consider a particular class of metric theories of gravity, namely the most general, massless and gauge-invariant\footnote{Gauge invariance, signaling a redundancy in the description of the vector through a tensor-field $A_\mu$, is a direct consequence of its masslessness and local-Lorentz invariance.} SVT theories with second-order equations of motion that include a single $U(1)$ gauge field $A_\mu$, with field strength $F_{\mu\nu}$, and a single massless scalar field $\Phi$, with vanishing potential. The theory thus propagates five degrees of freedom. In terms of the notation of Sec.~\ref{sec:IsaacsonBeyondGR}, the theory contains a physical metric $g_{\mu\nu}$ together with two additional gravitational fields $\Psi=\{A_\mu,\Phi\}$. The restriction to massless fields results directly from our assumption of trivial potentials. This is not a severe restriction because massive fields are, by definition, not expected to source any null memory, which we focus on in this paper. This section introduces the theory, as well as the standard first-order wave solutions in asymptotically flat spacetime. The actual computation of the gravitational null memory will then be carried out in the next section.

\subsection{Action and definitions}

The action of SVT gravity can be written as \cite{Heisenberg:2018acv}\footnote{This theory was derived through a decoupling limit of Generalized Proca theories \cite{Heisenberg:2014rta} written in gauge-invariant form, by introducing Stueckelberg fields.}
\begin{equation}\label{ActionHorndeski}
    S^{\st{SVT}}=\frac{1}{2\kappa_0}\int \dd^4 x\sqrt{-g}\left(\sum_{i=2}^5L_i\right)\,,
\end{equation}
where\footnote{Note that up to integrations by part, a term with $G_3(\Phi,X)=\Phi$ is equivalent to the kinetic term of the scalar, such that we specifically exclude such a term from $G_3$.}
\begin{align}
        L_2=&\,G_2(\Phi,X,Y,F,\tilde{F})\,,\\
        L_3=&-G_3(\Phi,X)\Box\Phi+\Big[\hat{G}_3(\Phi,X)\,g_{\alpha\beta}\notag\\
        &+\doublehat{G}_3(\Phi,X)\nabla_\alpha\Phi\nabla_\beta\Phi\Big]\tilde{F}^{\mu\alpha}\tilde{F}^{\nu\beta}\Phi_{\mu\nu}\,,\\
        L_4=&\,G_4(\Phi,X)\,R+G_{4X}\left[(\Box\Phi)^2-\Phi^{\mu\nu}\Phi_{\mu\nu}\right]\notag\\
        &+\hat{G}_4(\Phi,X)L^{\mu\nu\alpha\beta}F_{\mu\nu}F_{\alpha\beta}\notag\\
        &+\left[\doublehat{G}_4(\Phi)+\frac{1}{2}\hat{G}_{4X}\right]\tilde{F}^{\mu\alpha}\tilde{F}^{\nu\beta}\Phi_{\mu\nu}\Phi_{\alpha\beta}\,,\\
        L_5=&\,G_5(\Phi,X)\,G^{\mu\nu}\Phi_{\mu\nu}-\frac{G_{5X}}{6}\Big[(\Box\Phi)^3\notag\\
        &-3\,\Box\Phi\,\Phi^{\mu\nu}\Phi_{\mu\nu}+2\,\Phi_{\mu\nu}\Phi^{\nu\lambda}\Phi\du{\lambda}{\mu}\Big]\,,
\end{align}
with $\kappa_0\equiv 8\pi G_0$, where $G_0$ is the dimensionfull, bare gravitational constant and where $G_i$, $\hat{G}_i$ and $\doublehat{G}_i$ are arbitrary functions of $\Phi$, $X\equiv - ({1}/{2}) \nabla_\mu\Phi\nabla^\mu\Phi$, $Y\equiv\nabla_\mu\Phi\nabla_\nu\Phi F^{\mu\alpha}F\ud{\nu}{\alpha}$, $F\equiv - ({1}/{4}) F^{\mu\nu}F_{\mu\nu}$ and $\tilde{F}\equiv F^{\mu\nu}\tilde{F}_{\mu\nu}$, with the Hodge dual $\tilde{F}_{\mu\nu}\equiv\frac{1}{2}\epsilon_{\mu\nu\alpha\beta}F^{\alpha\beta}$. The quantity $L^{\mu\nu\alpha\beta}$ is the double-dual Riemann tensor, and it is given by
\begin{equation}\label{eq:ddR}
\begin{split}
    L^{\mu\nu\alpha\beta}\equiv&\, R^{\mu\nu\alpha\beta}+\Big(R^{\mu\beta}g^{\nu\alpha}
    +R^{\nu\alpha}g^{\mu\beta}-R^{\mu\alpha}g^{\nu\beta}\\
    &-R^{\nu\beta}g^{\mu\alpha}\Big)+\frac{1}{2}R\Big(g^{\mu\alpha}g^{\nu\beta}-g^{\mu\beta}g^{\nu\alpha}\Big)\,.
\end{split}
\end{equation}
We also define $\Phi_{\mu\nu}\equiv\nabla_{\mu}\nabla_{\nu}\Phi$ and $G_{iZ}\equiv \partial G_i/\partial Z$ for any function, and we impose a vanishing potential via $G_{2\Phi^n}(\Phi,0,0,0,0)= 0$ for any integer $n\geq 1$. Moreover, as already mentioned, we will for simplicity omit any explicit matter contributions, although it is important to keep in mind that only the metric couples minimally to matter while the additional scalar and vector remain purely in the gravity sector. The corresponding equations of motion associated with the action presented above can, for instance, be found in the Appendix of \cite{Heisenberg:2018mxx,Kobayashi:2011nu}.

The above action actually defines a large \textit{class} of theories. This class generalizes the well-known scalar Horndeski or covariant Galileon gravity class \cite{Horndeski:1974wa,Nicolis:2008in,Kobayashi:2019hrl}, reducing to it when $\nabla_\mu A_\nu=0$ and reducing to vector Horndeski gravity \cite{Horndeski:1976gi,Barrow:2012ay} for a vanishing scalar field. The theory members of this class are determined by the choices of $G_i$, $\hat{G}_i$, and $\doublehat{G}_i$ functionals. For example, $G_2 =({2\omega}/{\Phi}) X$, $G_4 =\Phi$, and all other $G_i=0$ correspond to BD gravity~\cite{Brans:1961sx,Dicke:1961gz}, while other choices lead to other theories, like sGB gravity~\cite{Zwiebach:1985uq,Gross:1986iv}, and double-dual Riemann gravity~\cite{Charmousis:2011ea,Charmousis:2011bf} (see Sec.~\ref{subsec:examples} for more details). 

\subsection{Leading-order waves}\label{sec:FOW}

We will now consider the leading-order waves of the theory presented in the previous section and solve for the corresponding propagating or \gls{radiative degrees of freedom} of the gravitational fields. Moreover, we will also comment on the direct detectability of the radiation through gravitational wave observations.

\subsubsection{Radiative degrees of freedom}\label{sec:RadDOFs}

We start by assuming that we can split all fields between slowly varying- and high-frequency components, as in Eq.~\eqref{eq:IsaacsonSplit} of Sec.~\ref{sec:IsaacsonBeyondGR}. Since we are focusing on asymptotically flat spacetimes that contain an isolated matter source, the limit to null infinity naturally selects the time-independent background solutions [Eq.~\eqref{eq:IsaacsonSplitBackground}] by identifying
\begin{equation}
    \bar{\eta}_{\mu\nu}=\eta_{\mu\nu}\,,\quad \bar{\Psi}_0=\Psi_0\,,
\end{equation}
where $\eta_{\mu\nu}$ and $\Psi_0$ are the Minkowski metric and asymptotic field values, defined through Eq.~\eqref{eq:AsymFlat}. In the concrete case of SVT gravity, we define the zeroth-order background fields ${\Psi}_{0}=\{a_{0\mu},\varphi_{0}\}$, which solve the background-field equations if we impose
\begin{equation}\label{eq: Background Solutions}
    a_{0\nu}=0\,,\; \nabla_{\mu}\varphi_0=0\,,\; 
    ^{\s{(0)}}G_{2}=0\,,
\end{equation}
where we define $^{\s{(0)}}G_i\equiv G_i(\varphi_0,0,..,0)$. The background equations of motion would actually be solved even for a nontrivial but constant background vector $a_{0\nu}$. However, in order to preserve local Lorentz invariance, we impose a vanishing asymptotic value for the vector field.

Moreover, we now assume that the isolated system produces gravitational radiation, whose physical modes asymptote to future null infinity with a $1/r$ falloff. We also impose a ``no-incoming radiation'' boundary condition at past null infinity. We then typically describe the oscillatory radiation modes in the radiation zone through perturbations of characteristic amplitude $\cal{O}(\alpha)$. These perturbations can therefore be identified with the high-frequency perturbations $h^H_{\mu\nu}$ and $\Psi^H$ of Sec.~\ref{sec:IsaacsonBeyondGR}, so that we have
\begin{equation}\label{eq:OriginalVariables}
     h^H_{\mu\nu}=h_{\mu\nu}\,,\quad \Psi^H=\{a_{\mu},\varphi\}\,.
\end{equation}
Note that for these fields, the terms \textit{radiation} and \textit{wave} can thus be used interchangeably. We further want to ensure a nonvanishing kinetic term for at least the tensor perturbations by imposing $^{\s{(0)}}G_{4}\neq 0$. For a certain gravitational wave source within a given theory, the scalar and vector waves might or might not be excited, depending on the concrete situation. We assume, however, that a tensor gravitational wave is present.

To address the leading-order wave propagation in Eqs.~\eqref{eq:EOMIIS2} and \eqref{eq:EOMIIS2Psi}, it is very useful to first expand the action in Eq.~\eqref{ActionHorndeski} to second order in perturbations,\footnote{At this order, the perturbations of $\cal{O}(\alpha^2)$ would contribute an additional linear term in the action, which is, however, irrelevant for the equations of motion, and we can safely neglect it at this stage.} which facilitates the determination of the physical dynamical degrees of freedom in the theory in terms of gauge-invariant modes. The second-order action in SVT theory contains a kinetic term that couples the metric and scalar perturbations $h_{\mu\nu}$ and $\varphi$. This term can, however, be removed through the field redefinition
\begin{equation}\label{RelationPhToh}
    \hat{h}_{\mu\nu}\equiv h_{\mu\nu}+\eta_{\mu\nu}\sigma\,\varphi\;,\quad \sigma\equiv \frac{^{\s{(0)}}G_{4\Phi}}{^{\s{(0)}}G_{4}}\,.
\end{equation}
Moreover, the scalar and vector perturbation can be rescaled so that their kinetic terms in the second-order action are canonically normalized. The necessary rescaling is
\begin{equation}\label{RescaledPhi}
    \hat{\varphi}\equiv \rho\,\varphi\,,\quad \hat{a}_\mu\equiv \zeta\,a_\mu\,,
\end{equation}  
where
\begin{equation}
    \rho\equiv\sqrt{3\,\sigma^2+\frac{(^{\scriptscriptstyle{(0)}}G_{2X}-2\,^{\scriptscriptstyle{(0)}}G_{3\Phi})}{^{\scriptscriptstyle{(0)}}G_4}}\,,
\end{equation}
and
\begin{equation}
    \zeta\equiv \sqrt{\frac{^{\scriptscriptstyle{(0)}}G_{2F}}{\phantom{}^{\scriptscriptstyle{(0)}}G_4}}\,.
\end{equation}
We require here that the coefficients $\sigma$, $\rho$ and $\zeta$ are real, which is also imposed by the positivity of the energy carried by the perturbations as we will see explicitly below.
In terms of the new variables in Eqs.~\eqref{RelationPhToh} and \eqref{RescaledPhi}, the second-order action of SVT theory then simply reads
\begin{equation}\label{ActionSVT2nd}
\begin{split}
    _{\s{(2)}}S^{\st{SVT}}=\frac{-1}{2\kappa_\text{eff}}\int\dd^4x\bigg[&\hat{h}^{\mu\nu}\mathcal{E}^{\alpha\beta}_{\mu\nu}\hat{h}_{\alpha\beta}\\
    &+\frac{1}{4}\hat{f}_{\mu\nu}\hat{f}^{\mu\nu}+\frac{1}{2}\partial_\mu\hat{\varphi}\partial^\mu\hat{\varphi}\bigg]\,,
\end{split}
\end{equation}
where we define the field strength of the leading-order vector perturbation $\hat{f}_{\mu\nu}\equiv\partial_\mu \hat{a}_{\nu}-\partial_\nu \hat{a}_\mu$, $\kappa_\text{eff}\equiv 8\pi G_\text{eff}$ with the effective gravitational constant $G_\text{eff}\equiv G_0/\phantom{}^{\scriptscriptstyle{(0)}}G_4$, as well as the trace $\hat{h}^t\equiv\eta^{\mu\nu}\hat{h}_{\mu\nu}$. Moreover, $\mathcal{E}_{\mu \nu}^{\alpha \beta}$ stands for the flat-space, Lichnerowicz operator
\begin{equation}
\begin{split}
    \mathcal{E}^{\alpha\beta}_{\mu\nu}\hat{h}_{\alpha\beta}=-\frac{1}{4}\Big[&\Box \hat{h}_{\mu\nu}-2\partial_\alpha\partial_{(\mu}\hat{h}\du{\nu)}{\alpha}+\partial_\mu\partial_\nu \hat{h}^t\\
    &-\eta_{\mu\nu}\left(\Box\hat{h}^t-\partial_\alpha\partial_\beta \hat{h}^{\alpha\beta}\right)\Big]\,,
\end{split}
\end{equation}
which allows for compact notation when writing down the Fierz-Pauli Lagrangian. Indeed, the first term in the action of Eq.~ \eqref{ActionSVT2nd} is equivalent to the usual Fierz-Pauli combination
\begin{equation}
\begin{split}
    \hat{h}^{\mu\nu}\mathcal{E}^{\alpha\beta}_{\mu\nu}\hat{h}_{\alpha\beta}=\frac{1}{4}&\Big[\partial_\mu\hat{h}_{\alpha\beta}\partial^\mu \hat{h}^{\alpha\beta}-\partial_\mu \hat{h}^t\partial^\mu \hat{h}^t\\
    &+2\partial_\mu \hat{h}^{\mu\nu}\partial_\nu\hat{h}^t -2\partial_\mu \hat{h}^{\mu\nu}\partial_\alpha\hat{h}\ud{\alpha}{\nu}\Big]\,,
\end{split}
\end{equation}
upon integration by parts.
Hence, in terms of new variables, the second-order action in Eq.~\eqref{ActionSVT2nd} is nothing but the linearized Einstein-Hilbert action with a sum of additional canonical fields. 

It is well known that perturbed solutions of metric theories of gravity are subject to gauge redundancies, which can be understood as arising from the invariance of the action under coordinate transformations, or equivalently, the invariance of the spacetime under diffeomorphisms.\footnote{More precisely, the spacetime solution $(\mathscr{M},g,\Psi)$ with a metric $g$ and additional tensor fields $\Psi$ defined on the manifold is physically equivalent to the solution $(\mathscr{N},\phi_*g,\phi_*\Psi)$, where $\phi:\mathscr{M}\rightarrow \mathscr{N}$ is a diffeomorphism and $\phi_*$ the associated pushforward \cite{WaldBook}.} Concretely, perturbations $\delta h_{\mu\nu}$ of a given background solution $\bar\eta_{\mu\nu}$ in a given coordinate system $x^\mu$ (as well as the perturbations of all other fields) are only physical, up to adding fake perturbations $\delta h^f_{\mu\nu}$ to the background that can be removed by a small coordinate transformation. Thinking actively, a perturbation is fake, if there exists an infinitesimal coordinate transformation that moves the points from $x^\mu\rightarrow x'^\mu=x^\mu+\xi^\mu$, with $|\xi^\mu|\ll 1$, such that\footnote{In terms of diffeomorphisms, the active transformation is given by the pushforward $(\phi_*\bar\eta)_{\mu\nu}$, whose components are equivalent to a coordinate transformation $\bar\eta'_{\mu\nu}$ defined through the pullback $x'^\mu=(\phi^*x)^\mu$ (see e.g. \cite{Bardeen:1980kt,WaldBook,carroll2019spacetime})}
\begin{equation}
    \bar\eta_{\mu\nu}(x')+h^f_{\mu\nu}(x')=\bar\eta'_{\mu\nu}(x')\,,
\end{equation}
where we compare the expressions at the same point $x'$. Therefore, the fake perturbations correspond to the Lie derivative of the background metric $h^f_{\mu\nu}(x)=\eta'_{\mu\nu}(x')-\eta'_{\mu\nu}(x')+\mathcal{O}(\xi^2)=-\mathcal L_\xi\bar\eta_{\mu\nu}+\mathcal{O}(\xi^2)$, such that any perturbation $\delta h_{\mu\nu}$ is only physical up the gauge transformation
\begin{equation}
    \delta h_{\mu\nu}\rightarrow \delta h_{\mu\nu}-\mathcal L_\xi \bar\eta_{\mu\nu}
\end{equation}
The same holds for all other field perturbations $\delta \Psi$ on a given background solution $\bar\Psi_0$
\begin{align}
    \delta \Psi &\rightarrow \delta \Psi-\mathcal L_\xi \bar\Psi_0\,.
\end{align}

For the SVT gravity at hand, only the metric perturbation transforms under this gauge symmetry through $\mathcal L_\xi \eta_{\mu\nu}=\left(\xi^\alpha\partial_\alpha \eta_{\mu\nu}+2\,\eta_{\alpha(\nu}\partial_{\mu)} \xi^\alpha\right)=2\,\eta_{\alpha(\nu}\partial_{\mu)} \xi^\alpha$, since for the background solutions in Eq.~\eqref{eq: Background Solutions} we have that $\mathcal L_\xi \varphi=\xi^\alpha\partial_\alpha \varphi_0=0$ and $\mathcal L_\xi a^\mu_0=\xi^\alpha\partial_\alpha a_0^\mu-a_0^\alpha\partial_\alpha \xi^\mu=0$. Moreover, since we made a split between the low-frequency and high-frequency parts of the fields, we need to make sure that the small coordinate transformations also affect only in this case the high-frequency part of the metric perturbation $h_{\mu\nu}$, which therefore transforms as
\begin{equation}\label{eq:CoordGaugeTransf}
  h_{\mu\nu}\rightarrow h_{\mu\nu}-2\,\eta_{\alpha(\nu}\partial_{\mu)} \xi_H^\alpha\,.
\end{equation}
Note that this gauge freedom is entirely inherited by the redefined perturbation variable $\hat h_{\mu\nu}$.
On the other hand, the vector perturbation $a^\mu$ inherits a gauge freedom from the $U(1)$ gauge transformations of the massless vector $A^\mu$, given by
\begin{equation}\label{eq:U(1)GaugeTransf}
  a_{\mu}\rightarrow a_{\mu}+\partial_\mu\Lambda^H\,.
\end{equation}

By performing suitable coordinate [Eq.~\eqref{eq:CoordGaugeTransf}] and $U(1)$ gauge transformations [Eq.~\eqref{eq:U(1)GaugeTransf}], we can impose at the level of the equations of motion the following gauge conditions
\begin{equation}\label{eq:FirstOrderGaugeFixing}
    \partial_\mu \hat{h}^{\mu\nu}=0\,,\quad\hat{h}^t=0\quad \text{and}\quad \partial_\mu \hat{a}^{\mu}=0\,.
\end{equation}
In this gauge, it is no surprise that the leading-order wave propagation described by Eqs.~\eqref{eq:EOMIIS2} and \eqref{eq:EOMIIS2Psi} lead to decoupled wave equations for all the hatted perturbations
\begin{equation}\label{eq:FirstOrderProp}
    \Box \hat{h}_{\mu\nu}=0\,,\quad \Box \hat{a}_\mu=0\,,\quad \Box \hat{\varphi}=0\,.
\end{equation}

Let us now single out the \gls{radiative modes} that dominate in the limit to null infinity. To do so, we perform the usual $3+1$ decomposition with a spatial orthonormal basis, given by a unit longitudinal or radial direction $n^i$ and two transverse vectors $u^i$ and $v^i$, such that $\delta_{ij}=n_in_j+u_iu_j+v_iv_j$. The plane wave solutions of Eq.~\eqref{eq:FirstOrderProp} can then be expanded in terms of 5 radiative degrees of freedom, 2 of which appear in the gauge-invariant, transverse-traceless (TT) tensor part ($h_+$ and $h_\times$), 2 are transverse vectors ($\hat{a}_u$ and $\hat{a}_v$), and 1 is a scalar $\hat\varphi$:
\begin{equation}\label{eq:PhysicalPol}
    \hat{h}^\text{TT}_{ij}=\hat h_+e^+_{ij}+\hat h_\times e^\times_{ij}\,,\quad \hat a^\text{T}_i=\hat a_u u_i+\hat a_v v_i\,,\quad \hat\varphi\,.
\end{equation}
where we defined the \gls{polarization} tensors $e^+_{ij}\equiv u_iu_j-v_iv_j$ and $e^\times_{ij}\equiv u_iv_j+v_iu_j$.

These radiative modes are indeed invariant under the residual gauge freedom left over after fixing Eq.~\eqref{eq:FirstOrderGaugeFixing}, which are transformations satisfying $\Box\xi_{H}^{\mu}=\partial_\mu\xi_H^\mu=0$ and $\Box\Lambda^H=0$. While other components of the metric and vector perturbations are not explicitly gauge invariant in this approach, this will not be relevant in our considerations, since, in particular, only the radiative modes will contribute to the gauge-invariant, and hence, physical, response of a detector in the radiation zone (see Sec.~\ref{GravPol} below). One could have chosen to work exclusively with manifestly gauge-invariant quantities (see e.g.~\cite{Flanagan:2005yc}) at the cost of having to deal with \textit{a priori} nonlocal fields. Once more, however, only certain combinations of them will be locally measurable in the physical detector response. Yet another approach would have been to not introduce any metric perturbations or vector potentials at all, and only work with manifestly gauge-invariant and local objects (see e.g.~ \cite{Koop_PhysRevD062002,Garfinkle:2022dnm}). In the end, all of these approaches are physically equivalent in the limit to null infinity.

Before proceeding, let us make several observations about the radiative modes presented above. First, observe that the \gls{tensor polarizations} defined above are identical to the tensor polarizations of the physical metric perturbations: $\hat h_+=h_+$ and $\hat h_\times =h_\times$, and hence, $\hat{h}^\text{TT}_{ij}=h^\text{TT}_{ij}$. Second, note that one can construct linear combinations of the radiative modes that have a certain (tensor $s=-2$ or vector $s=-1$) ``spin weight'' (see e.g.~\cite{DAmbrosio:2022clk}), namely
\begin{equation}\label{Scalarht}
     h=\hat{h}\equiv \hat{h}_{ij}\bar{m}_i\bar{m}_j=\frac{1}{2}\hat{h}_{ij}(e^+_{ij}-i\, e^\times_{ij})=\hat{h}_+-i\hat{h}_\times\,,
\end{equation}
and
\begin{equation}\label{Scalara}
     \hat{a}\equiv \hat{a}_{i}^\text{T}\sqrt{2}\bar{m}_i=\hat{a}_{u}-i\hat{a}_{v}\,,
\end{equation}
where $\bar m_i\equiv ({1}/{\sqrt{2}})(u_i-iv_i)$ is a transverse vector of spin weight $s=-1$ as determined by its behavior under rotations about the longitudinal direction. Third, observe that both the tensor and the vector polarization modes are \textit{transverse}, while the tensor modes are, moreover, \textit{traceless}. It is convenient to define a transverse projector
\begin{equation}
\perp_{ij}\equiv \delta_{ij}-n_in_j=u_iu_j+v_iv_j=m_i\bar m_j+\bar m_im_j\,,
\end{equation}
as well as a transverse-traceless projector
\begin{equation}\label{eq:Projectors}
    \perp_{ijab}\equiv \perp_{ia}\perp_{jb}-\frac{1}{2}\perp_{ij}\perp_{ab}\,.
\end{equation}
These projectors can be used to single out the transverse vector modes and transverse-traceless tensor modes in any given expression. For instance, we can write
\begin{equation}
    \hat{h}^\text{TT}_{ij}=\perp_{ijab}\hat{h}_{ab}\,,\quad \hat{a}^\text{T}_{i}=\perp_{ij}\hat{a}_{j}\,.
\end{equation}
In what follows, we will use the two spin-weighted scalars in Eqs.~\eqref{Scalarht} and \eqref{Scalara}, together with the scalar perturbation $\hat{\varphi}$, to describe the leading-order tensor, vector and scalar radiation respectively. In coordinates $\{u,r,\Omega=(\theta,\phi)\}$ and asymptotically close to null infinity, these radiative modes take the general form
\begin{align}
    &\left(\hat{h}\,,\;\hat{a}\,,\;\hat{\varphi} \right) \sim \frac{1}{r} \left[f_h(u,\theta,\phi), f_a(u,\theta,\phi), f_\varphi(u,\theta,\phi) \right]\,.\label{eq:OutgoingPlaneWave}
\end{align}
for some complex functions $f_{h,a}$ and real function $f_{\varphi}$.

\subsubsection{Gravitational polarizations}\label{GravPol}

At this point it is important to realize the difference between the notion of propagating degrees of freedom of the gravitational fields, also called modes, discussed above and \gls{gravitational polarizations}. The former are gauge-invariant solutions to the equations of motions for the leading-order perturbed fields in the theory. The latter are gauge-invariant radiative modes \textit{within the perturbations of the physical metric} that minimally couples to matter, and which, therefore, can be detected through typical gravitational wave observations. In order to determine the admitted gravitational polarizations of a theory, one must compute the electric part of the Riemann tensor of the physical metric that enters the geodesic deviation equation. 

In the nonrelativistic and low-frequency regime and in Fermi normal coordinates, the geodesic equation reads (see e.g. \cite{maggiore2008gravitational})
\begin{equation}\label{GeodesicDeviation}
    \ddot s_i=-R_{i0j0}s_j\,,
\end{equation}
where $s_i$ is the proper distance between nearby geodesics while overhead dots represent derivatives with respect to proper time along the geodesic. This equation is at the core of all gravitational wave experiments, and it is valid for any metric theory of gravity, since it is a direct consequence of minimal coupling to matter and does not depend on the field equations. 

Very generically, a metric theory of gravity admits up to six polarizations \cite{Eardley_PhysRevLett.30.884,Eardley_PhysRevD.8.3308}. This statement can be understood from the fact that, in the limit to null infinity of an asymptotically flat spacetime characterized by Eq.~\eqref{eq:AsymFlat}, the physical metric can be decomposed as 
\begin{equation}\label{eq:TotalMetricPert}
    g_{\mu\nu}=\eta_{\mu\nu}+H_{\mu\nu}\,,
\end{equation}
with $\eta_{\mu \nu}$ the Minkowski metric and an arbitrary perturbation $H_{\mu \nu}$ characterizing six distinct gauge-invariant radiation modes $H_\text{P}=\{H_+,H_\times,H_u,H_v,H_b,H_l\}$ of the physical metric, which fall off as $1/r$ near future null infinity [cf. Eq.~\eqref{eq:OutgoingPlaneWave}]. The electric part of the linear Riemann tensor $\phantom{}_{\s{(1)}}R_{\mu\nu\rho\sigma}[\eta,H]$, which itself is gauge-invariant, can then be written as a sum of these six gauge invariant polarizations $H_\text{P}$. More precisely, near future null infinity, one finds that (see e.g. \cite{poisson2014gravity,Flanagan:2005yc})
\begin{equation}\label{eq:LinRiemannPol}
    R_{i0j0}=-\frac{1}{2}\ddot{A}_{ij}+\mathcal{O}\left(\frac{1}{r^2}\right)\,,
\end{equation}
where
\begin{equation}
\begin{split}
    A_{ij}=&\,e^+_{ij}H_++e^\times_{ij}H_\times+2 n_{(i}u_{j)}H_u+2 n_{(i}v_{j)}H_v\\
    &+(u_iu_j+v_jv_i)H_b+n_in_jH_l\,.
\end{split}
\end{equation}
The structure of the electric part of the Riemann tensor tells us that, on top of the two familiar tensor modes $H_+$ and $H_\times$, the change in proper distance can also arise from two additional vector modes $H_u$ and $H_v$ or from a scalar longitudinal $H_l$ or scalar transverse (``breathing'') $H_b$ mode.

In order to determine the gravitational polarization content of the leading-order wave within the SVT gravity, we can simply evaluate the general linear expression in Eq.~\eqref{eq:RiemmanFirstOrder}, namely
\begin{equation}
    R_{i0j0}=-\frac{1}{2}\left(\partial_0\partial_0H_{ij}+\partial_i\partial_jH_{00}-\partial_0\partial_iH_{0j}-\partial_0\partial_jH_{0i}\right)\,,
\end{equation}
for $H_{\mu\nu}\rightarrow h_{\mu\nu}$. Decomposing the leading-order wave in terms of the gauge-invariant degrees of freedom of Eq.~\eqref{eq:PhysicalPol}, using Eq.~\eqref{RelationPhToh}, one can write
\begin{equation}
    h_{ij}=\hat{h}_{ij}^\text{TT}-\delta_{ij}\frac{\sigma}{\rho}\,\hat{\varphi}\quad \text{and}\quad h_{00}=\frac{\sigma}{\rho}\,\hat{\varphi}\,.
\end{equation}
After imposing the falloff of Eq.~\eqref{eq:OutgoingPlaneWave}, and noting that to $\mathcal{O}\left(r^{-1}\right)$ we can replace $\partial_i\rightarrow -n_i\partial_0$, the electric part of the linearized Riemann tensor reads
\begin{equation}\label{ElectricRiemannLin}
\begin{split}
    R_{i0j0}&=-\frac{1}{2}\left(\ddot{\hat{h}}^\text{TT}_{ij}-\left[\delta_{ij}-n_in_j\right]\frac{\sigma}{\rho}\,\ddot{\hat{\varphi}}\right)\\
    &=-\frac{1}{2}\left(e^+_{ij}\,\ddot{h}_++ e^\times_{ij}\,\ddot{h}_\times-\left[u_iu_j+v_iv_j\right]\sigma\,\ddot{\varphi}\right)\,.
\end{split}
\end{equation}
Thus, whenever $\sigma\neq 0$, the theory possesses an additional scalar breathing polarization within the detector response, as it is well known from scalar Horndeski theory \cite{Hou:2017bqj}.  On the other hand, since the massless vector does not couple nonminimally to the Ricci scalar, the gravitational vector degrees of freedom will never induce any additional gravitational polarizations. Therefore, this is a concrete example of a theory that possesses 5 propagating degrees of freedom, yet only a maximum of three polarizations of the physical metric survive.  

Observe that, in the language of Sec.~\ref{sec:Isaacson}, $H_{\mu\nu}$ here in principle represents the ``total'' perturbation, including both the leading-order high-frequency perturbations, as well as the low-frequency perturbations $h^H_{\mu\nu}$ and $\delta h^L_{\mu\nu}$. Therefore, $H_{\mu \nu}$ will generically also contain the memory component. In other words, while being a nonoscillatory ``Coulombi'' contribution (cf. \cite{Thorne:1992sdb}), memory can still be regarded as a part of the radiative modes of the physical metric, which are, in turn, part of the measurable components of the Riemann tensor in the limit to null infinity. This is nothing but a rephrasing of the statement that memory is part of the gravitational \textit{radiation}, as defined earlier. These memory modes in the perturbation of the physical metric will inherit the tensor structure from the leading-order radiation, and therefore, it is in general not possible for memory modes to excite a different gravitational polarization from the ones excited through the leading-order radiation (see also discussion below). This is because the excitation of additional polarizations can be associated with nonminimal couplings between the fields and the Ricci scalar in the full action of the theory.


\section{\label{sec:SVTMem} Null Memory in massless SVT Theories}


We will now proceed and compute the memory contribution as sourced by the leading-order radiation of the SVT theory introduced in the previous section. In general, the memory effect can be understood by integrating the linearized geodesic deviation equation [Eq.~\eqref{GeodesicDeviation}] twice using Eq.~\eqref{eq:LinRiemannPol}, and then solving for the separation vector $\Delta s_i\equiv \Delta s_i(\tau_f)-\Delta s_i(\tau_0)$ between an initial time $\tau_0$ and a final time $\tau_f$ before and after the passage of a gravitational wave\footnote{Note that we are neglecting here subdominant initial velocity contributions.}
\begin{equation}
    \Delta s_i\approx \frac{1}{2}\Delta A_{ij}s^j(\tau_0)\,.
\end{equation}
A permanent change in proper distance $\Delta s_i\neq 0$ characterizes gravitational memory. More precisely, any piece within the radiative gravitational polarizations of the physical metric in $ A_{ij}$, which induces such a permanent displacement will be called a memory component. Hence, on a very general basis, metric theories of gravity are expected to contain memory that can be associated with each of the six polarizations. This naturally leads to a distinction between scalar, vector and \gls{tensor memory}, where in this terminology, the terms ``scalar'', ``vector'' or ``tensor'' refer to the polarization type that induces a permanent displacement (not to be confused with the tensorial nature of a leading-order wave that acts as a source of memory, which is not restricted in any way).

Let us therefore consider the existence of leading-order waves with small amplitudes of $\mathcal{O}(\alpha)$, as described in the previous section, which satisfy the propagation equations [Eqs.~\eqref{eq:EOMIIS2} and \eqref{eq:EOMIIS2Psi}] and, in practice, are assumed to be known \textit{a priori}, and try to solve the leading-order low-frequency equations [Eqs.~\eqref{eq:EOMIS2} and \eqref{eq:EOMIS2Psi}] close to null infinity. These solutions will determine the low-frequency background components
\begin{equation}
    \delta h^L_{\mu\nu}=\delta h_{\mu\nu}\,,\quad \delta \Psi^L_{\mu\nu}=\{\delta a_\mu,\delta\varphi\}.
\end{equation}
As we will now see explicitly in the next subsection, these low-frequency perturbations will in general describe gravitational memory as defined above. In particular, the tensor memory described by $\delta h_{\mu\nu}$ as the solution to Eq.~\eqref{eq:EOMIS2} will correspond to the gravitational memory known from GR.

Can there also be similar nontrivial scalar or vector memory contributions? Or, in other words, does the low-frequency equation for the scalar and the vector field [Eq.~\eqref{eq:EOMIS2Psi}] in SVT theories also give rise to a corresponding scalar and vector memory component? For the vector perturbations, the answer can be given right away, since SVT gravity never excites any vector polarizations in the physical metric. By definition, this directly implies that SVT gravity will not contain any vector memory either. For the scalar equation, there is no analogous source entering the low-frequency equation [Eq.~\eqref{eq:EOMIS2Psi}], which simply reads $\Box\delta\varphi=0$. This is, however, not surprising, in the sense that within the SVT theory that we consider, no scalar source is radiated to null infinity. Therefore, the only nontrivial null-memory component is the tensor null memory, which is what we will focus on henceforth. Moreover, since the low-frequency perturbation $\delta h_{\mu\nu}$ is sourced by radiative modes that reach null infinity, it is more specifically called a \textit{tensor null-memory} component.

\subsection{Tensor null memory: Derivation of the memory evolution equation\label{sec:TensorMemorySVT}}

The right-hand side of Eq.~\eqref{eq:EOMIS2} is the memory-source term given by the averaged second-order metric field equations as a function of the leading-order waves and it can be readily computed and simplified tremendously. Using the gauge fixing conditions of Eq.~\eqref{eq:FirstOrderGaugeFixing}, as well as the leading-order equations of motion [Eq.~\eqref{eq:FirstOrderProp}], while also performing several integrations by parts,\footnote{Recall that integration by parts is permitted due to the spacetime averaging, which is crucial in this calculation.} the result can be reduced to
\begin{align}
   \Big\langle \phantom{}_{\s{(2)}}\mathcal{G}^{\st{SVT}}_{\mu\nu}\phantom{}\Big\rangle&\propto\Big\langle \partial_\mu\hat{h}_{\alpha\beta}\partial_\nu\hat{h}^{\alpha\beta}+2\,\partial_\mu\hat{a}_\alpha\partial_\nu\hat{a}^\alpha+2\,\partial_\mu\hat{\varphi}\partial_\nu\hat{\varphi} \Big\rangle \notag\\
   &\propto t^{\st{SVT}}_{\mu\nu}\,.
   \label{eq:G2-ave}
\end{align}
Observe that this expression only contains two derivative operators and does not contain any second-order perturbations and any mixed terms between first-order perturbations of the tensor, vector and scalar waves. Spacetime averaging removes any second-order perturbations, since they would depend on the high-frequency scale only, and thus, they average out. Averaging also allows for integrations by parts, which is crucial to ensure gauge invariance. The boundary terms generated upon integration by parts can be discarded because they introduce higher-order corrections only [see footnote above Eq.~\eqref{eq:EOMISGR}].

Let us comment on the result above. For notational compactness, Eq.~\eqref{eq:G2-ave} introduces the total effective energy-momentum tensor $t^{\st{SVT}}_{\mu\nu}\equiv  t_{\mu\nu}^{\st{GR}}+ t_{\mu\nu}^{\hat{a}}+ t_{\mu\nu}^{\hat{\varphi}}$.
Thus, as in GR and despite the many nontrivial operators in the full action of Eq.~\eqref{ActionHorndeski}, the result is simply proportional to the sum of known energy-momentum (pseudo)tensors of free bosonic fields in terms of leading-order perturbations
\begin{align}\label{eq:StressEnergyFirst}
    t_{\mu\nu}^{\st{GR}}&=\frac{1}{4\kappa_\text{eff}}\Big\langle \partial_\mu\hat{h}_{\alpha\beta}\partial_\nu\hat{h}^{\alpha\beta}\Big\rangle\,,\\
    t_{\mu\nu}^{\hat{a}}&=\frac{1}{2\kappa_\text{eff}}\Big\langle\hat{f}_{\mu\alpha}\hat{f}\du{\nu}{\alpha}\Big\rangle=\frac{1}{2\kappa_\text{eff}}\Big\langle\partial_{\mu}\hat{a}_{\alpha}\partial_{\nu}\hat{a}^{\alpha}\Big\rangle\,,\\
    t_{\mu\nu}^{\hat{\varphi}}&=\frac{1}{2\kappa_\text{eff}}\Big\langle\partial_\mu\hat{\varphi}\partial_\nu\hat{\varphi}\Big\rangle\,.
\end{align}
Because of the spacetime average and the wave equations [Eq.~\eqref{eq:FirstOrderProp}], this total energy-momentum tensor is conserved,\footnote{Conservation follows from the fact that the divergence commutes with the average and property (II) of the average procedure (see Sec.~\ref{sec:IsaacsonGR})} as well as traceless
\begin{align}\label{eq:StressEnergyTot}
   \partial^\mu t_{\mu\nu}^{\st{SVT}}=0\,,\quad\eta^{\mu\nu}  t_{\mu\nu}^{\st{SVT}}=0\,.
\end{align}
Furthermore, gauge invariance can easily be checked (see e.g. \cite{maggiore2008gravitational}) such that the total stress-energy (pseudo)tensor only depends on the modes in Eq.~\eqref{eq:PhysicalPol}, namely
\begin{align}\label{eq:StressEnergySecond}
    t_{\mu\nu}^{\st{SVT}}=\frac{1}{4\kappa_\text{eff}}\Big\langle \partial_\mu\hat{h}^\text{ TT}_{ij}\partial_\nu\hat{h}^{\text{TT}}_{ij}+2\,\partial_\mu\hat{a}^\text{T}_i\partial_\nu\hat{a}^{\text{T}}_i+2\,\partial_\mu\hat{\varphi}\partial_\nu\hat{\varphi}\Big\rangle\,.
\end{align}

The left-hand side of Eq.~\eqref{eq:EOMIS2} is the operator $_{\s{(1)}}\mathcal{G}_{\mu\nu}$ evaluated at the low-frequency perturbation, which will have the same structure as the operator governing the high-frequency propagation equation. Thus, within this term, the low-frequency perturbations can also be decoupled by performing a redefinition of fields analogous to what we used for the leading-order radiation [see Eq.~\eqref{RelationPhToh}], namely,
\begin{equation}\label{RelationPhTohMemory}
    \delta\hat{h}_{\mu\nu}\equiv \delta h_{\mu\nu}+\eta_{\mu\nu}\sigma\,\delta\varphi\,.
\end{equation}
Moreover, thanks to the properties of the source term in Eq.~\eqref{eq:StressEnergyTot}, infinitesimal coordinate transformations at the low-frequency level $\xi_\mu^L$ can be used to once again impose the gauge conditions\footnote{Note that these gauge conditions do not correspond to what is usually called the \textit{TT-gauge}. Indeed, the TT-gauge can in general only be imposed outside of the source (see e.g. \cite{maggiore2008gravitational}).}
\begin{equation}\label{eq:De-DonderSecond}
    \partial^\mu\delta\hat{h}_{\mu\nu}=0\,,\quad \eta^{\mu\nu}\delta\hat{h}_{\mu\nu}=0\,,
\end{equation}
such that the left-hand side of Eq.~\eqref{eq:EOMIS2} reduces to a wave equation $\phantom{}_{\s{(1)}}\mathcal{G}^{\st{SVT}}_{\mu\nu}\propto\Box \delta h_{\mu\nu}$. 

Putting everything together, the metric equation [Eq.~\eqref{eq:EOMIS2}] in the asymptotic region of an asymptotically flat spacetime described by SVT gravity simply reads
\begin{equation}\label{eq:SecondOrderEOMMem}
   \Box\delta\hat{h}_{\mu\nu}=-2\kappa_\text{eff}\,t^{\st{SVT}}_{\mu\nu}\,.
\end{equation}
Thus, as in GR, the physical energy-momentum carried by the leading-order waves induces a new, low-frequency perturbation $\delta\hat{h}_{\mu\nu}$ that satisfies a sourced wave equation. Equation~\eqref{eq:SecondOrderEOMMem} can be viewed as our first key result, as it essentially already implies that the new low-frequency component is identified with a null-memory contribution [see \cite{PhysRevD.44.R2945,Favata:2008yd,Garfinkle:2022dnm}] and anticipates that the memory equation in full SVT gravity is modified in a very minimal way, independent of the many details at the non-linear level. 

Let us stress, however, that the precise form of the leading-order TT radiation itself can still depend on the higher-order terms through modifications in the dynamics that create the gravitational waves. Such modifications of the amplitude and phase of the emitted gravitational waves therefore also alter the memory contribution in an indirect way. Interestingly, in the case of compact binary coalescence, modifications in the TT radiation already simply arise due to a change in the rate of emitted energy induced by the emission of additional scalar or vector radiation.

\subsection{Tensor null memory: Solution to the evolution equation\label{sec:DerTensorMemorySVT}}

We will now explicitly present the procedure to solve Eq.~\eqref{eq:SecondOrderEOMMem} for field points $(t,\vec{x})$ in the limit of outgoing null rays and obtain an explicit formula for the tensor null memory. 
Up to the expression for memory in terms of the TT part of the components of the physical metric perturbation in Eq.~\eqref{NonLinDispMemorySVT} below, we will follow the arguments laid out in \cite{PhysRevD.44.R2945} (see also \cite{Favata:2008yd}), filling in certain useful details. We then proceed by further simplifying the memory formula through an explicit rewriting of the result in terms of coefficients of an expansion in spin-weighted spherical harmonics. This is also the form in which memory is naturally extracted from asymptotic BMS balance laws. Explicitly showing the equivalence of the two different representations will therefore later allow us to confirm our results by a direct comparison to recent computations of BMS balance laws in the special case of BD theory (see Appendix~\ref{MatchToAsymptoticsBD}). 

The wave equation can be solved generically through the standard retarded Green's function
\begin{equation}\label{eq:NLinMem1}
    \delta\hat{h}_{\mu\nu}(x)=\frac{\kappa_\text{eff}}{2\pi}\int\dd^4x'\,t^{\st{SVT}}_{\mu\nu}(x')\,\frac{\delta(t-t'-|\vec{x}-\vec{x}'|)}{|\vec{x}-\vec{x}'|}\,.
\end{equation}
To transform this expression into a form that is useful to us, we need to perform several simplification steps. First, we will switch to spherical coordinates $\{u,r,\Omega=(\theta,\phi)\}$ for the arguments of the tensor components\footnote{As it is customary in a large part of the gravitational wave community, we will still use a Minkowski basis $\{t,x,y,z\}$ for the index structure of tensor components.} with $\vec{x}=r\vec{n}$, where $\vec{n}$ is the outgoing radial unit vector and $u= t-r$ is the asymptotic retarded time. This will ease an evaluation of the expression close to null infinity, when we take the limit $r\rightarrow\infty$ at fixed $u$. We do the same for the source, such that the integration measure becomes $\dd^4x'\rightarrow \dd u' r'^2\dd r'\dd^2\Omega'$. 

Moreover, like for the leading-order waves, the physical response to the tensor memory outside of the source\footnote{While the source of the null memory is itself constructed out of null waves within the radiation zone, such that both $r$ and $r'$ are large, any location at which the radiative source is nonzero within the past null cone of a point $(t,\vec{x})$ where we evaluate the null-memory component satisfies $r'\ll r$ (see e.g. \cite{Garfinkle:2022dnm}). This also justifies the approximation $t'\approx u+r'\vec{n}'\cdot \vec{n}$.} is captured by the propagating TT component of the metric perturbation $\delta \hat{h}_{ij}^\text{TT}=\delta h_{ij}^\text{TT}$ of the physical metric, as dictated by the general expression given in Eq.~\eqref{eq:LinRiemannPol}. This means that the measurable effect of the tensor null memory is given by a projection of the spatial components of Eq.~\eqref{eq:NLinMem1} onto its transverse-traceless part. 

The final step required to transform Eq.~\eqref{eq:NLinMem1} into a more suitable form for our calculations is to study the behavior of the integrand with respect to $r'$ in order to perform the radial integration. To do this, we note that, while we are interested in the behavior of $\delta\hat{h}_{\mu\nu}$ close to null infinity, the radiative source term itself also needs to be considered far from its own source (i.e.~in the radiation zone), where the outgoing leading-order waves are well defined. In particular, this implies that the leading-order waves satisfy Eq.~\eqref{eq:OutgoingPlaneWave} with respect to the source variable $r'$, such that to leading order in large $r'$, the radiative energy-momentum tensor takes the form (see e.g. \cite{maggiore2008gravitational})\footnote{To $\mathcal{O}\left(r'^{-1}\right)$, we have $t^{\st{SVT}}_{ij}= t_{00} n'_in'_j$, while $t_{i0}= -t_{00}n'_i$. This follows from the general structure $t^{\st{SVT}}_{\mu\nu}\sim\langle \partial_\nu\psi\partial_\mu\psi\rangle$, together with the falloff properties Eq.~\eqref{eq:OutgoingPlaneWave}, which imply that we can replace $\partial_i\rightarrow -n_i\partial_0$. Moreover, in units with $c=1$, we can equate an energy flux of radiation at speed $c$ with an energy density $1/r^2\,dE^{\s{\text{SVT}}}/dud\Omega=c \, t^{\s{\text{SVT}}}_{00}$.}

\begin{equation}\label{eq:EMTAs}
    t_{\mu\nu}(u',r',\Omega')=\frac{1}{r'^2}\frac{dE}{du'd\Omega'}\ell'_\mu \ell'_\nu\,.
\end{equation}
Here, the null vector $\ell_\mu$ is defined as $\ell_\mu\equiv -\nabla_\mu t+\nabla_\mu r$, with $\nabla_i r=n_i$ and the energy flux ${dE}/{du'd\Omega'}$ only depends on retarded time $u'$ and the direction $\Omega'$. This allows us to use the delta function to perform the radial integral and arrive at\footnote{Recall the formula $\delta g(x)=\sum_i\delta(x-x_i)/|g'(x_i)|$, where $g(x_i)=0$ is a root.}
\begin{equation}
\begin{split}
    \delta h_{ij}^\text{TT}(u,r,\Omega)=&\frac{\kappa_\text{eff}}{2\pi}\int_{-\infty}^u\dd u' \int\dd^2\Omega'\,\frac{dE^{\st{SVT}}}{du'd\Omega'}\\
    &\times\left[\frac{\perp_{ijab}(\Omega)n'_an'_b}{r(1-\vec{n}'\cdot\vec{n}(\Omega))+u-u'}\right]\,.
\end{split}
\end{equation}

The limit to null infinity can now be performed rather straightforwardly to yield
\begin{align}\label{NonLinDispMemorySVT}
    \delta h_{ij}^\text{TT}=\,\frac{\kappa_\text{eff}}{2\pi r} \int_{S^2}\dd^2\Omega'\,\frac{dE^{\scriptscriptstyle{\text{SVT}}}}{d\Omega'}\,\left[\frac{n'_in'_j}{1-\vec{n}'\cdot\vec{n}}\right]^\text{TT}\,,
\end{align}
where the superscript TT denotes a projection onto the TT component via $\perp_{ijab}(\Omega)$ defined in Eq.~\eqref{eq:Projectors} and where we define the energy per solid angle as
\begin{equation}\label{eq:EperSolidAngle}
    \frac{dE^{\s{\text{SVT}}}}{d\Omega'}\equiv\int_{-\infty}^u\dd u'\,\frac{dE^{\s{\text{SVT}}}}{du'd\Omega'}=r^2\int_{-\infty}^u\dd u'\,t_{00}^{\s{\text{SVT}}}(u',r,\Omega')\,.
\end{equation}
Note that Eq.~\eqref{NonLinDispMemorySVT} is a well-defined quantity even if $\vec{n}'=\vec{n}$, since in that case the numerator vanishes as well. 

The tensor null memory for SVT gravity can thus be simply evaluated by inserting the expression for the radiative energy density in Eq.~\eqref{eq:StressEnergySecond} 
into the expression for the time-integrated energy flux in Eq.~\eqref{eq:EperSolidAngle}. In terms of canonically normalized ``hatted'' variables, defined in Eqs.~\eqref{RelationPhToh} and \eqref{RescaledPhi}, and the expansion in terms of polarization tensors in Eq.~\eqref{eq:PhysicalPol}, the energy density reads
\begin{align}
    t_{00}^{\s{\text{SVT}}}(u',r,\Omega')
    =\frac{1}{2\kappa_\text{eff}}\,\Big\langle |\dot{\hat{h}}|^2+|\dot{\hat{a}}|^2+\dot{\hat{\varphi}}^2\Big\rangle\,.\label{eq:StressEnergyFourthHat}
\end{align}
However, the physical (in the sense of observationally relevant) modes are characterized in terms of the perturbations of the original fields that appear in the full action. Indeed, it is the scalar perturbation $\varphi$ that is associated with a potentially observable additional breathing mode, as shown in Eq.~\eqref{ElectricRiemannLin}. Thus, in terms of the physical modes, the radiative energy density becomes
\begin{align}
    t_{00}^{\scriptscriptstyle{\text{SVT}}}(u',r,\Omega')
    =\frac{1}{2\kappa_\text{eff}}\Bigg\langle |\dot{h}|^2+\zeta^2\,|\dot{a}|^2+\rho^2\,\dot{\varphi}^2\Bigg\rangle\,,\label{eq:StressEnergyFourth}
\end{align}
where we recall the definitions
\begin{equation}
    \rho\equiv\sqrt{3\,\sigma^2+\frac{(^{\s{(0)}}G_{2X}-2\,^{\s{(0)}}G_{3\Phi})}{^{\s{(0)}}G_4}}\,,
\end{equation}
and
\begin{equation}
    \sigma\equiv \frac{^{\s{(0)}}G_{4\Phi}}{^{\s{(0)}}G_{4}}\,\,,\quad\zeta\equiv \sqrt{\frac{^{\s{(0)}}G_{2F}}{\phantom{}^{\s{(0)}}G_4}}\,.
\end{equation}

Equations~\eqref{NonLinDispMemorySVT} and~\eqref{eq:StressEnergyFourth} define the tensor null memory of massless SVT gravity. These equations indeed represent a memory component, as one can simply infer from their functional form. That is, these expressions contain a time integral over a strictly positive function, which therefore induces a permanent displacement within the detector strain. Observe that this tensor null memory is sourced by \textit{all} radiative degrees of freedom in the SVT theory, independent of whether these radiative modes excite additional gravitational polarizations in the physical metric or not. More precisely, a nonzero value of $\sigma$, which implies that the breathing mode is excited (recall the discussion in Sec.~\ref{GravPol}), only influences the value of the scalar prefactor $\rho$, but it does not determine whether the emitted scalar radiation provides an additional tensor memory source in principle. Moreover, the energy density emitted in vector modes backreacts to produce tensor memory, even though they are in no way connected to any gravitational polarizations of the physical metric.

\subsection{Tensor null memory: Spin-weighted spherical harmonic decomposition\label{sec:SWSHTensorMemorySVT}}
The tensor null-memory expression in Eq.~\eqref{NonLinDispMemorySVT} above can still be further simplified for practical use. We do so by first defining the spin-weight $s=-2$ quantity [recall Eq.~\eqref{Scalarht}]
\begin{equation}
    \delta h\equiv \delta h_{ij}^\text{TT}\bar m_i\bar m_j\,,
\end{equation}
and then determining the value of the coefficients $\delta h^{lm}$ of its expansion in spin-weighted spherical harmonics (see Appendixes~\ref{Ap:SWSH} and \ref{Ap:SWSHExpansion}). Writing the memory in this alternative form will ultimately also allow a direct comparison to the memory arising from BMS balance laws or within a systematic PN expansion, as we will see below. 

The simplest way to obtain the spin-weighted spherical harmonic coefficients is to first expand the TT-projected term in the square brackets of Eq.~\eqref{NonLinDispMemorySVT} as a geometric series and then transform the result to a symmetric trace-free (STF) basis, which can subsequently be related to the spin-weighted spherical harmonic expansion. In Appendix~\ref{SWSH} we gathered a collection of formulas for different multipole expansions and the relations between them, and also introduce the notation we use for STF tensors (see Appendix~\ref{Ap:STF}). As we show explicitly in Appendix~\ref{DerivationEq} the transformation of the TT-projected square brackets in Eq.~\eqref{NonLinDispMemorySVT} in terms of STF tensors results in the identity [see also Eq.~(2.34) in \cite{BlanchetPaper}]
\begin{equation}\label{eq:RelBlanchet}
   \left[\frac{n'_in'_j}{1-\vec{n}'\cdot\vec{n}}\right]^\text{TT}=\,\perp_{ijab}\sum_{l=2}^\infty \frac{2(2l+1)!!}{(l+2)!}\,n_{L-2}n'_{\langle abL-2\rangle}\,,
\end{equation}
such that 
\begin{equation}\label{NonLinDispMemoryST}
\begin{split}
    \delta h_{ij}^\text{TT}=\frac{\kappa_\text{eff}}{2\pi r}&\perp_{ijab}\sum_{l=2}^\infty \frac{1}{l !}\,n_{L-2}\,\frac{2(2l+1)!!}{(l+1)(l+2)}\\
    &\times\int_{S^2}\dd^2\Omega'\,\frac{dE^{\scriptscriptstyle{\text{SVT}}}}{d\Omega'}(u)\,n'_{\langle abL-2\rangle}\,.
\end{split}
\end{equation}
By comparing to the general STF multipole expansion of a rank-2 TT tensor written out in Eq.~\eqref{eq:AExpansionULVL}, we immediately see that the memory only contributes via the electric-parity multipole, namely
\begin{align}
    \delta U_{L}&=\frac{\kappa_\text{eff}}{8\pi r}\frac{2(2l+1)!!}{(l+1)(l+2)}\int_{S^2}\dd^2\Omega'\,\frac{\dd E^{\scriptscriptstyle{\text{SVT}}}}{\dd \Omega'}(u)\,n'_{\langle L\rangle}\,,\\
    \delta V_{L}&=0\,,
\end{align}
where we have relabeled $ijL-2\rightarrow L$ through multi-index notation.

A change to the pure-spin TT harmonic basis using Eq.~\eqref{eq:AUVlmToULVL} as well as Eq.~\eqref{eq:AYlmToNL} then yields\footnote{The memory computed in a PN expansion assumes precisely this form, as explicitly shown in GR \cite{Blanchet:1992br} (see also \cite{Favata:2010zu}). However, we define the mass multipole without factoring out the $r^{-1}$ dependence. Note as well that we could have obtained Eq.~\eqref{ResU} more directly by using the identity in Eq.~\eqref{eq:AUlmCalc}.}
\begin{equation}\label{ResU}
    \delta U^{l m}=\frac{4\kappa_\text{eff}}{r}\sqrt{\frac{(l-2)!}{2(l+2)!}}\int_{S^2}\dd^2 \Omega' \,\frac{\dd E^{\s{\text{SVT}}}}{\dd \Omega'}(u)\,\bar{Y}^{l m}(\Omega')\,.
\end{equation}
This expression can finally be related to the spin-weighted spherical harmonic expansion through Eq.~\eqref{eq:AUVlmToHlm} to give
\begin{align}
\begin{split}\label{NonLinDispMemoryModesSVT}
  \delta h^{lm}_{\st{SVT}}=&\,\frac{2\kappa_\text{eff}}{r}\sqrt{\frac{(l-2)!}{(l+2)!}}\int_{-\infty}^u\dd u'\\
  &\times\int_{S^2}\dd^2\Omega'\,\bar{Y}^{lm}(\Omega')\,r^2\,t^{\st{SVT}}_{00}(u',r,\Omega')\,.
\end{split}
\end{align}
where $t^{\st{SVT}}_{00}$ is given by Eq.~\eqref{eq:StressEnergyFourth}. Note that the angular integral in this expression can be evaluated analytically as a sum of 3$j$ symbols by simply also expanding the leading-order waves in spin-weighted spherical harmonics and applying the identity in Eq.~\eqref{SWSHTrippleInt}, which involves three spin-weighted spherical harmonics.\footnote{This could in principle simplify the widely used \textsc{GWMemory} python package for memory computations in GR \cite{Talbot:2018sgr}, where the angular integral is partially performed numerically. We also note that the \textsc{GWMemory} code does not perform any space-time averaging.}


\subsection{Interesting examples}
\label{subsec:examples}
We close this section with a few explicit results for the tensor null memory of interesting subclasses of Horndeski gravity, in order to further exemplify the broad scope of the SVT theory investigated above.

\subsubsection{Brans-Dicke gravity} 
Horndeski gravity reduces to BD theory for the choices
\begin{equation}\label{eq:BDGs}
\begin{split}
    G_2&=\frac{2\omega}{\Phi} X\,\\
    G_4&=\Phi\,,\\
    G_i&=0 \;\;\text{otherwise}\,,
\end{split}
\end{equation}
such that the action given in Eq.~\eqref{ActionHorndeski} becomes
\begin{equation}\label{ActionBD}
    S^{\st{BD}}=\frac{1}{2\kappa_0}\int \dd^4 x\sqrt{-g}\left(\Phi R-\frac{\omega}{\Phi}g^{\mu\nu}\nabla_\mu\Phi\nabla_\nu\Phi\right)\,,
\end{equation}
with physical metric $g_{\mu\nu}$ and where $\omega$ is a coupling constant.
Inserting Eq.~\eqref{eq:BDGs} into Eq.~\eqref{eq:StressEnergyFourth}, the corresponding energy-momentum (pseudo)tensor sourcing the tensor memory reads 
\begin{equation}\label{sourceTensorMemoryBD}
    t_{00}^{\st{BD}}=\frac{1}{2\kappa^{\st{BD}}_{\text{eff}}}\Bigg\langle |\dot{h}|^2+(2\omega+3)\left(\frac{\dot{\varphi}}{\varphi_0}\right)^2\Bigg\rangle\,,
\end{equation}
where $\kappa^{\st{BD}}_{\text{eff}}=\kappa_0/\varphi_0$, such that Eq.~\eqref{NonLinDispMemoryModesSVT} becomes
\begin{equation}\label{NonLinDispMemoryModesBD}
\begin{split}
  \delta h^{lm}_{\st{BD}}=&r\sqrt{\frac{(l-2)!}{(l+2)!}}\int_{S^2}\dd^2\Omega'\,\bar{Y}^{lm}(\Omega')\\
  &\times\int_{-\infty}^u\dd u'\Bigg\langle\,|\dot{h}|^2+(2\omega+3)\left(\frac{\dot\varphi}{\varphi_0}\right)^2\Bigg\rangle\,.
\end{split}
\end{equation}
As we show in Appendix~\ref{MatchToAsymptoticsBD}, the result in Eq.~\eqref{NonLinDispMemoryModesBD} precisely matches the memory extracted from the BMS balance laws in BD theory, which were previously computed in \cite{hou_gravitational_2021,tahura_brans-dicke_2021,hou_conserved_2021,hou_gravitational_2021_2}, and therefore corrects the formula for the tensor memory of BD gravity of \cite{du_gravitational_2016}.

Moreover, note that in this theory 
\begin{equation}
\sigma=\frac{^{\s{(0)}}G_{4\Phi}}{^{\s{(0)}}G_{4}}=\frac{1}{\varphi_0}\neq 0\,,
\end{equation}
such that according to the results in Sec.~\ref{GravPol}, BD gravity has an additional breathing polarization. As discussed above, this fact only minimally modifies the memory formula of Eq.~\eqref{NonLinDispMemoryModesBD}, since $\omega$ is already constrained to be a large number (e.g.~$\omega > 4 \times 10^4$ due to constraints from the tracking of the Cassini spacecraft and the Shapiro time delay~\cite{Bertotti:2003rm}).
However, the existence of such an additional scalar polarization makes it in principle possible to also measure scalar memory, hence memory within the scalar polarization of the detector response. Yet, as already mentioned, we do not find any analogous scalar null-memory component because there is no analogous null source for the scalar mode.
This agrees with the results in \cite{hou_gravitational_2021,tahura_brans-dicke_2021,hou_conserved_2021,hou_gravitational_2021_2} as they do not find a full BMS constraint for the scalar, which indicates, in the terminology of \cite{tahura_brans-dicke_2021}, that the scalar displacement contribution to the Riemann tensor that can potentially arise through other mechanisms\footnote{See for instance \cite{du_gravitational_2016} where a nonvanishing shift in the scalar is shown to arise as a consequence of the no-hair theorem.} is rather a more general persistent gravitational variable, as opposed to a true memory component.

\subsubsection{Scalar Gauss-Bonnet gravity}\label{sec:sGB}
SGB theory can be obtained by choosing \cite{Kobayashi:2011nu,Kobayashi:2019hrl}
\begin{equation}\label{eq:CorrespondencesGBHorndeski}
    \begin{split}
        G_2&=X+8f^{(4)}(\Phi)X^2(3-\ln X)\,,\\ G_3&=4f^{(3)}(\Phi)X(7-3\ln X)\,,\\
        G_4&=1+4f^{(2)}(\Phi)X(2-\ln X)\,,\\
        G_5&=-f^{(1)}(\Phi)\ln X\,,
    \end{split}
\end{equation}
where $f^{(n)}(\Phi)\equiv\partial^n f/\partial\Phi^n$. This leads to a sGB theory action\footnote{Note that while this correspondence is not obvious at the level of the action, the resulting equations of motion are equivalent.}
\begin{equation}\label{ActionsGB}
\begin{split}
    S^{\st{sGB}}=\frac{1}{2\kappa_0}\int \dd^4 x\sqrt{-g}\bigg(& R-\frac{1}{2}g^{\mu\nu}\nabla_\mu\Phi\nabla_\nu\Phi\\
    &+\epsilon^2\,f(\Phi)\,\mathcal{G}\bigg)\,,
\end{split}
\end{equation}
where the Gauss-Bonnet curvature scalar is defined as $\mathcal{G}\equiv -\tilde R\ud{\mu\nu}{\rho\sigma} \tilde R\ud{\rho\sigma}{\mu\nu} =R^{\mu\nu\rho\sigma}R_{\mu\nu\rho\sigma}-4R^{\mu\nu}R_{\mu\nu}+R^2$ and the Hodge dual reads $\tilde R\ud{\mu\nu}{\rho\sigma}\equiv\frac{1}{2}\epsilon^{\mu\nu\alpha\beta}R_{\alpha\beta\rho\sigma}$. To lowest nontrivial order, the theory represents GR with a canonical scalar field, while for constant values of $\Phi$ the entire theory reduces to GR because the Gauss-Bonnet term integrates to a boundary term. Although at first glance the correspondence of sGB to Horndeski theory given by Eq.~\eqref{eq:CorrespondencesGBHorndeski} could suggest that for nontrivial functions $f(\Phi)$ the sGB term could actually contribute nontrivially to the memory, this is not the case. A closer look reveals that
\begin{equation}
\begin{split}
    &3\frac{^{\s{(0)}}G^2_{4\Phi}}{^{\s{(0)}}G^2_4}+\frac{(^{\s{(0)}}G_{2X}-2\,^{\s{(0)}}G_{3\Phi})}{^{\s{(0)}}G_4}\Bigg\lvert_\text{sGB}\\
    &=\lim_{X\rightarrow 0}\;1-4 w(X)\left(f^{(2)}(\varphi_0)+2f^{(4)}(\varphi_0)\right)X+...=1,
\end{split}
\end{equation}
where we defined $w(X)\equiv 2-\ln X$ as well as $f^{(N)}$ for the $N$th derivative of $f$. Hence, the higher-order sGB term does not modify the memory formula and the theory simply contributes through the canonical scalar term within Eq.~\eqref{eq:StressEnergyFourth} as
\begin{equation}\label{eq:sGBenergymomentum}
    t_{00}^{\st{sGB}}
    =\frac{1}{2\kappa}\bigg\langle |\dot{ h}|^2+\,\dot{\varphi}^2\bigg\rangle\,.
\end{equation}
Indeed, in Sec.~\ref{sec:GeneralFormula}, we will show that any term in the action involving more than two derivative operators will not modify the tensor memory in an explicit way.

Furthermore, since 
\begin{equation}
\sigma=\frac{^{\s{(0)}}G_{4\Phi}}{^{\s{(0)}}G_{4}}=0\,,
\end{equation}
sGB gravity does not excite any additional scalar polarizations (breathing or longitudinal) within the physical metric, as opposed to the BD theory considered above. This result is consistent with the analysis of the polarization states of sGB theory in~\cite{Wagle:2019mdq}. This also implies that sGB by definition only features tensor null memory.


\subsubsection{Double-Dual Riemann coupling}
Similarly, choosing \cite{Narikawa:2013pjr,Kobayashi:2019hrl}
\begin{equation}\label{eq:CorrespondenceddR}
    \begin{split}
        G_2&=X\,,\\
        G_5&=X\,,\\
        G_i&=0 \;\;\text{otherwise}\,,
    \end{split}
\end{equation}
the Horndeski action reduces to
\begin{equation}\label{ActionddR}
\begin{split}
    S^{\st{ddR}}=\frac{1}{2\kappa_0}\int \dd^4 x\sqrt{-g}\bigg( R&+\nabla_\mu\Phi \nabla_\nu\Phi\Phi_{\alpha\beta}L^{\mu\nu\alpha\beta}\\
    &-\frac{1}{2}g^{\mu\nu}\nabla_\mu\Phi\nabla_\nu\Phi\bigg)\,,
\end{split}
\end{equation}
where the derivative coupling to the double-dual Riemann tensor defined in Eq.~\eqref{eq:ddR}
therefore also leads to second-order equations of motion. Since this coupling is associated with a nontrivial $G_5$ function only, we immediately know that such an additional term will not alter the tensor null-memory formula. Hence, exactly as for sGB, the memory is sourced by Eq.~\eqref{eq:sGBenergymomentum} and only modified due to the kinetic term of the scalar, and thus by the mere presence of the possibly propagating scalar mode.



\section{\label{sec:GeneralFormula}Tensor Null Memory in Dynamical Metric Theories}

We will now generalize the result of the explicit null-memory computation of massless SVT theory in the previous section by proving that, assuming the scenario outlined in Sec.~\ref{sec:IsaacsonBeyondGR}, the functional form of the low-frequency evolution equation at the basis of the tensor null memory remains unaltered in the limit to null infinity for any dynamical metric theory with trivial potentials (see Definition~\ref{DefMetricTheory}), given a set of minimal assumptions outlined in Sec.~\ref{sec:StatementOfClaim} below. Together with an additional assumption on the structure of the radiative energy-momentum tensor in the limit to null infinity, this further implies directly that the tensor null memory can still be written as [cf. Eq.~\eqref{NonLinDispMemoryModesSVT}]
\begin{equation}\label{NonLinDispMemoryGen3}
     \delta h^{l m}_L=r \sqrt{\frac{(l-2)!}{(l+2)!}}\int \dd^2 \Omega'\,\bar{Y}^{l m}\int\dd u' \,\phantom{}_{\s(2)}t_{00}\,,
\end{equation}
where $\phantom{}_{\s(2)}t_{\mu\nu}$ is a conserved and gauge-invariant generalized energy-momentum (pseudo)tensor of the leading-order radiation in the theory. This (pseudo)tensor contains a purely GR piece in the form of an Isaacson energy-momentum tensor and a correction thereof that consists of a sum of terms, which exclusively depend on high-frequency perturbations of the additional gravitational fields. 

In a second step, we will explore the scope of the theorem in Sec.~\ref{sec:ScopeOfClaim} by outlining which classes of theories actually satisfy the underlying assumptions. 
Before stating and proving the theorem which allows us to draw the above conclusions, however, we need to first introduce the so-called \textit{second-variation} approach \cite{Maccallum:1973gf}, which we present in the next subsection.

\subsection{The second-variation approach}\label{sec:IsaacsonSecondVariation}

In the second-variation framework, the low-frequency metric equation [Eq.~\eqref{eq:EOMIS}],
as well as the propagation equations  [Eqs.~\eqref{eq:EOMIIS} and \eqref{eq:EOMIISPsi}]
can be derived from a perturbed action. To do so, one first separates the metric, as well as any other field, into a background and a perturbation, as in Eq.~\eqref{eq:IsaacsonSplit}, and expands the action $S$ of the theory to second order in perturbation fields
\begin{equation}
    S^{\st{sv}}=\phantom{}_{\s{(0)}}S+\phantom{}_{\s{(1)}}S+\phantom{}_{\s{(2)}}S\,.
\end{equation}
With this in hand, one then promotes the action $S^{\st{sv}}$ to an effective one $S^{\st{sv}}\rightarrow S^{\st{sv}}_\text{eff}$ by treating the background fields and the perturbations as independent fields.

Indeed, one recovers the low-frequency metric equation [Eq.~\eqref{eq:EOMIS}] by varying the effective action
\begin{align}
    S^{\st{sv}}_\text{eff}[\{\bar g^L,\bar\Psi^L\},\{h^H,\Psi^H\}]&=\phantom{}_{\s{(0)}}S_\text{eff}[\{\bar g^L,\bar\Psi^L\}]\label{eq:SecondVariationAction}\\
    &+\phantom{}_{\s{(2)}}S_\text{eff}[\{\bar g^L,\bar\Psi^L\},\{h^H,\Psi^H\}]\,,\nonumber
\end{align}
with respect to the independent background metric $\bar g^L_{\mu\nu}$ and a subsequent restriction to low-frequency quantities. More precisely,
\begin{equation}\label{eq:MemEqAction}
       \left[\frac{\delta S^{\st{sv}}_\text{eff}}{\delta \bar g_L^{\mu\nu}}\right]^L=0
\end{equation}
gives back Eq.~\eqref{eq:EOMIS}, where after the variation, the background metric is not treated as an independent variable anymore. Note that the linear piece $\phantom{}_{\s{(1)}}S_\text{eff}$ in Eq.~\eqref{eq:SecondVariationAction} can always be neglected, because a term linear in perturbation fields generated through a variation with respect to the background metric will vanish upon restriction to low-frequency modes [see property (I) of the average procedure in Sec.~\ref{sec:IsaacsonGR}].

Let us see how this comes about explicitly. The variation of the zeroth-order action $\phantom{}_{\s{(0)}}S_\text{eff}$ simply gives back the background equation operator of the theory
\begin{align}
    \frac{\delta\phantom{}_{\s{(0)}}S_\text{eff}}{\delta \bar g_L^{\mu\nu}}\propto\phantom{}_{\s{(0)}}\mathcal{G}_{\mu\nu}\,,
\end{align}
and thus the left-hand-side in Eq.~\eqref{eq:EOMIS}. A variation of $\phantom{}_{\s{(2)}}S_\text{eff}$, on the other hand, provides a definition for the effective energy-momentum (pseudo)tensor of the perturbation fields via the Hilbert approach, namely
\begin{equation}\label{eq:SecondOrderVariationEMTensor}
    \frac{-2}{\sqrt{-\bar g^L}}\Bigg\langle\frac{\delta \phantom{}_{\s{(2)}}S_\text{eff}}{\delta \bar g_L^{\mu\nu}}\Bigg\rangle\equiv \phantom{}_{\s{(2)}}t_{\mu\nu}
\end{equation}
upon averaging out the small scales. We have already established that the right-hand side of Eq.~\eqref{eq:EOMIS} is proportional to the energy-momentum tensor of the perturbation fields in Eq.~\eqref{eq:RHSToEM}. From that point of view, Eq.~\eqref{eq:EOMIS} can therefore be understood as a dynamical equation of the \textit{a priori} unknown background, sourced by the effective energy-momentum tensor of the perturbation fields.

The propagation equations for the gravitational waves and any other perturbed field can also be obtained through this variational approach. The metric equation [Eq.~\eqref{eq:EOMIIS}], as well as all other high-frequency propagation equations of the additional gravitational fields [Eq.~\eqref{eq:EOMIISPsi}], are recovered by varying $S^{\st{sv}}_\text{eff}$ in Eq.~\eqref{eq:SecondVariationAction} with respect to the corresponding field perturbation. This is because, in general, a variation of a perturbed action with respect to a perturbation field yields the same equation that one obtains by perturbing the total field equations computed from the full action \cite{Maccallum:1973gf,Taub1971}. Moreover, the action in Eq.~\eqref{eq:SecondVariationAction} is enough, because we can again safely neglect any linear piece $\phantom{}_{\s{(1)}}S_\text{eff}$;  a variation of such a piece would not contain any high-frequency components.\footnote{In fact, a variation of $\phantom{}_{\s{(1)}}S_\text{eff}$ with respect to a perturbation field would simply give back the corresponding background equation, and therefore, does not contain any additional information.} Hence, the leading-order propagation equations [Eqs.~\eqref{eq:EOMIIS} and \eqref{eq:EOMIISPsi}] can be obtained from
\begin{equation}\label{eq:PropEqAction}
       \left[\frac{\delta S^{\st{sv}}_\text{eff}}{\delta h^{H\mu\nu}}\right]^H=0\,,\quad  \left[\frac{\delta S^{\st{sv}}_\text{eff}}{\delta \Psi^H}\right]^H=0\,.
\end{equation}

In the following, this approach will enable us to generalize the result for tensor null memory in SVT theories (explicitly computed in the previous section) and will allow us to explore the boundaries of the validity of the associated functional form of tensor null memory. More precisely, we will show that, after further splitting the low-frequency fields into a background and perturbation as in Eq.~\eqref{eq:IsaacsonSplitBackground} and in the limit to null infinity, the low-frequency metric equation [Eq.~\eqref{eq:EOMIS}] that determines the null-memory component very generally takes on the simple form that we derived for SVT gravity. That is, the equation reduces to a massless wave equation for the memory component sourced by the energy-momentum of the high-frequency waves [cf. Eq.~\eqref{eq:SecondOrderEOMMem}]. In other words, the equation remains the same, except for the proportionality coefficient and the energy-momentum tensor, which is replaced by a corresponding generalization that is still conserved. These statements will be made precise through the formulation of an associated Theorem~\ref{Theorem1}.

In particular, the proof of Theorem~\ref{Theorem1} will be based on the key observation that in the limit to null infinity, both the propagation equations [Eqs.~\eqref{eq:EOMIIS} and \eqref{eq:EOMIISPsi}], as well as the low-frequency energy-momentum tensor in Eq.~\eqref{eq:EOMIS} (see Lemma \ref{LemmaP}), actually only depend on the second-order effective action 
    \begin{equation}\label{eq:flatSA}
    \phantom{}_{\s(2)}S^{\s{flat}}_\text{eff}\equiv \phantom{}_{\s(2)}S_\text{eff}[\{\eta,\Psi_0\},\{h^H,\Psi^H\}]
\end{equation}
evaluated on a \textit{flat}, but still independent background $\eta_0\equiv\{\eta_{\mu\nu},\Psi_0\}$. By flat, here we mean a metric with a vanishing curvature, as well as the existence of coordinates, in which the metric and the fields satisfy $\partial_\alpha \eta_{\mu\nu}=\partial_\alpha\Psi_0=0$, so that, in particular, the Christoffel symbols vanish as well.\footnote{For GR, more precisely with the assumption of vanishing torsion and nonmetricity (assumptions that are therefore also satisfied in our case) it can be explicitly proven, that the vanishing of the Riemann tensor implies that for every point on the manifold, there exists a chart, in which the metric can be written as a Minkowski metric. However, we do not prove here the additional requirement $\partial_\alpha\Psi_0=0$, which should therefor be regarded as an additional assumption.}

Moreover, in the subsequent derivation of the tensor null-memory formula of SVT theories in Sec.~\ref{sec:DerTensorMemorySVT} [starting from Eq.~\eqref{eq:SecondOrderEOMMem}], the radiative energy momentum tensor essentially only played the role of a spectator. Based on Theorem~\ref{Theorem1}, this observation will then allow us to conclude that for a very broad class of dynamical metric theories of gravity, the memory formula remains of the same functional form as in Eq.~\eqref{NonLinDispMemoryModesSVT}.

\subsection{A theorem for the evolution of memory in dynamical metric theories}\label{sec:StatementOfClaim}

We want to start with a more formal definition of the dynamical metric theories we will focus on. Based on the framework outlined in Secs.~\ref{sec:Isaacson} and~\ref{sec:IsaacsonSecondVariation} above, we will then first prove Lemma~\ref{LemmaP}, based on which we establish Theorem~\ref{Theorem1} that restricts the functional form of the leading-order, low-frequency metric equation [Eq.~\eqref{eq:EOMIS}] of such theories in source-centered coordinates describing the asymptotic null region.

\subsubsection{Preliminary definitions and lemmas}

\begin{definition} \textbf{Dynamical Metric Theory.}\label{DefMetricTheory}
    Let $\mathscr M$ be a four-dimensional oriented and differentiable pseudo-Riemannian manifold equipped with a Lorentzian metric $g$ of signature (-1,+1,+1,+1) and an associated Levi-Civita connection $\Gamma$. By \textup{dynamical metric theory} we shall mean a local and covariant Lagrangian-metric theory on $\mathscr M$ described by an action
    \begin{equation}\label{eq:ActionMetricTheory}
        S=\frac{1}{2\kappa_0}\int \dd^4x\,\sqrt{-g} \left(L_\text{G}[g,\Psi]+L_\text{m}[g]\right)
    \end{equation}
    consisting of a matter Lagrangian $L_\text{m}$ minimally coupled to the metric $g$ only, and a gravitational Lagrangian $L_\text{G}$ consisting of scalar terms, each involving at least two derivative operators, that are constructed out of curvature invariants as well as a set of additional dynamical gravitational fields $\Psi$ that preserve local Lorentz invariance, and whose dynamical degrees of freedom are massless.\footnote{For fields that are not a scalar, we will also assume a local gauge invariance. However, this does not represent an additional assumption, since for theories with local-Lorentz-invariance, gauge invariance is in general needed to obtain a description of massless fields via components of tensor fields (see e.g. \cite{Schwartz:2014sze}).}
\end{definition}

We will generally distinguish between three classes of dynamical metric theories. The first ones are those that only involve terms in the action with two derivative operators. On a classical level, such theories can be considered to be \textit{exact}, in the sense that they could be, in principle, well posed as an initial value problem for any values of their coupling constants. 

The other two classes encompass theories with more than two derivative operators per term, which typically represent theories admitting higher-order derivative interaction terms. Between those two classes, we further distinguish between theories that still admit second-order equations of motion and theories which lead to higher equations of motion. As discussed in Sec.~\ref{sec:IsaacsonBeyondGR}, such theories require the additional assumption of small couplings $\epsilon$, and they are naturally interpreted as effective field theories. In our context, a small-coupling approximation is satisfied for theories admitting second order equations, if the assumption in Eq.~\eqref{eq:SmallParam4} holds. Moreover, for theories with higher-order equations of motion, we must in principle demand the stronger condition Eq.~\eqref{eq:SmallParam5}. Thus, for theories of the last class, Lemma~\ref{LemmaP} below, as well as the statements in the subsequent Theorem~\ref{Theorem1} need to be understood as statements up to factors of $\mathcal{O}(\epsilon f_H)$. 

\begin{lemma}\label{LemmaP}
    In the limit to null infinity
    \begin{equation}
    \frac{-2}{\sqrt{-\eta}} \Bigg\langle\frac{\delta \phantom{}_{\s{(2)}}S^{\s{flat}}_{\text{eff}}}{\delta \eta^{\mu\nu}}\Bigg\rangle= \phantom{}_{\s{(2)}}t_{\mu\nu}[\eta_0,\{h^H,\Psi^H\}]\,,
     \end{equation}
    and hence, the leading-order energy-momentum tensor [Eq.~\eqref{eq:SecondOrderVariationEMTensor}] only depends on the flat, second-order effective action defined in Eq.~\eqref{eq:flatSA}.
\end{lemma}
\begin{proof}
    The energy-momentum tensor was defined in Eq.~\eqref{eq:SecondOrderVariationEMTensor}, and therefore, the lemma is proven if we can show that in the limit to null infinity 
    \begin{equation}
        \Bigg\langle\frac{\delta \phantom{}_{\s{(2)}}S^{\s{flat}}_\text{eff}}{\delta \eta^{\mu\nu}}\Bigg\rangle = \Bigg\langle\frac{\delta \phantom{}_{\s{(2)}}S_\text{eff}}{\delta \bar g^{L\mu\nu}}\Bigg\rangle\,,
    \end{equation}
    where recall here that the background metric must be treated as an independent field.  
    Since we only consider local and covariant theories with Levi-Civita connection, a variation of the second-order action $\phantom{}_{\s{(2)}}S_\text{eff}$ with respect to the independent background metric $\bar g^L_{\mu\nu}$ will only arise from polynomial contributions of the background metric, the associated covariant derivatives $\bar\nabla_\mu$ through its Christoffel symbols $\bar\Gamma\ud{\mu}{\alpha\beta}$, and contractions of background curvature invariants $_{\s{(0)}}R_{\mu\nu\rho\sigma}[\bar g^L]$, the last two of which vanish for a flat metric by definition. To prove Lemma~\ref{LemmaP}, we therefore need to show that any term in the \textit{nonflat} perturbed action $\phantom{}_{\s{(2)}}S_\text{eff}$ that involves background curvature or connection operators does not contribute to the effective stress-energy tensor $\phantom{}_{\s(2)}t_{\mu\nu}$ at null infinity.
    
    For any term in the perturbed action that contains curvature or connection quantities, a variation of a curvature or connection coefficient with respect to the background metric $\bar{g}_L^{\mu\nu}$ can be written as \cite{Stein:2010pn}
    \begin{equation}
        \delta_{\bar g^L} \left[\phantom{}_{\s{(2)}}S_\text{eff}\right]\supset\int\dd^4 x\sqrt{-\bar g^L} \; \bar \nabla_\sigma P\ud{\sigma}{\mu\nu} \; \delta \bar g_L^{\mu\nu}
    \end{equation}
    for some tensor $P\ud{\sigma}{\mu\nu}$ upon integration by parts. Such a contribution vanishes
    \begin{equation}
        \phantom{}_{\s{(2)}}t_{\mu\nu}\supset-2\Big\langle\bar \nabla_\sigma P\ud{\sigma}{\mu\nu}\Big\rangle=0
    \end{equation}
    due to property (II) of the average (see Sec.~\ref{sec:IsaacsonGR}). See~\cite{Stein:2010pn} for a pedagogical discussion of the above two facts. 
    
    On the other hand, if the variation of such a term is hitting a background metric, the resulting term will still contain either a curvature or a connection operator. Therefore, any such contribution will vanish in the limit to null infinity.
\end{proof}

\subsubsection{Memory evolution theorem beyond GR}

\begin{theorem}\label{Theorem1} Consider a dynamical metric theory, as introduced in Def.~\ref{DefMetricTheory}, for which
    \begin{enumerate}[(i)]
        \item the assumptions of Sec.~\ref{sec:IsaacsonBeyondGR} hold. In particular the Eqs.~\eqref{eq:SmallParam1}, \eqref{eq:SmallParam2} and \eqref{eq:SmallParam4} are satisfied, such that the leading-order, low-frequency metric equation can be written as [Eq.~\eqref{eq:EOMIS}]
        \begin{align}
        \begin{split}\label{eq:EOMISTh}
            \phantom{}_{\s{(0)}}&\mathcal{G}_{\mu\nu}[\{\bar g^L,\bar\Psi^L\}]\\
            &=-\big\langle \phantom{}_{\s{(2)}}\mathcal{G}_{\mu\nu}[\{\bar g^L,\bar\Psi^L\},\{h^H,\Psi^H\}]\big\rangle\,,
        \end{split}
        \end{align}
        while the leading-order, high-frequency propagation equations [Eqs.~\eqref{eq:EOMIIS} and \eqref{eq:EOMIISPsi}] are 
        \begin{align}
           \phantom{}_{\s{(1)}}\mathcal{G}_{\mu\nu}[\{\bar g^L,\bar\Psi^L\},\{h^H,\Psi^H\}]&=0\,,\label{eq:EOMIISTh}\\
           \phantom{}_{\s{(1)}}\mathcal{J}[\{\bar g^L,\bar\Psi^L\},\{h^H,\Psi^H\}]&=0\,.\label{eq:EOMIISPsiTh}
        \end{align}
        Moreover, the low-frequency fields can further be split into a background $\{\bar{\eta}_{\mu\nu},\bar{\Psi}_0\}\equiv \bar{\eta}_0$ that solves the background equations of motion $\phantom{}_{\s{(0)}}\mathcal{G}_{\mu\nu}[\bar{\eta}_0]=0$ and a small perturbation
        \begin{equation}\label{eq:SmallParam3Th}
        \bar{g}^L_{\mu\nu}=\bar{\eta}_{\mu\nu}+\delta h^L_{\mu\nu}\,, \quad \bar{\Psi}^L=\bar{\Psi}_0+\delta \Psi^L\,;
        \end{equation}
        \item the space-time is asymptotically flat, which in particular implies that there exists a chart $\{t,x,y,z\}$, with $r\equiv\sqrt{x^2+y^2+z^2}$, for which $g_{\mu\nu}=\eta_{\mu\nu}+\mathcal{O}(1/r)$ and $\Psi=\Psi_0+\mathcal{O}(1/r)$ in the limit to null infinity $r\rightarrow\infty$ at fixed $u\equiv t-r$. The background $\eta_0=\{\eta_{\mu\nu},\Psi_0\}$ solves the vacuum field equations and preserves local Lorentz invariance, with $\partial_\alpha\Psi_0=0$ for scalar fields and $\Psi_0=0$ for all other tensor fields and where $\eta_{\mu\nu}$ is the Minkowski metric\footnote{Note that asymptotic flatness also requires that any matter energy momentum tensor components $T_{\mu\nu}$ fall off quickly enough, so that $r^{2}T_{\mu\nu}$ has a smooth limit to null infinity.};
        \item there exists a set of leading-order high-frequency field perturbations $\hat{h}_{\mu\nu}^H=W(h^H,\Psi^H)$ and $\hat{\Psi}^H=V(\Psi^H)$, for some functions $W$ and $V$, for which, in the limit to null infinity, the leading-order, high-frequency propagation equations [Eqs.~\eqref{eq:EOMIISTh} and \eqref{eq:EOMIISPsiTh}] reduce to a set of decoupled wave equations for each field
        \begin{equation}\label{eq:ThWaveEq}
            \Box \hat{h}_{\mu\nu}^H=0\,,\quad \Box \hat{\Psi}^H=0\,,
        \end{equation}
        upon imposing the Lorenz gauge as well as tracelessness
        \begin{equation}\label{eq:ThLorenzGauge}
        \partial^\mu\hat{h}_{\mu\nu}^H=0\,,\quad\eta^{\mu\nu}\hat{h}_{\mu\nu}^H=0\,.
        \end{equation}
    \end{enumerate}
    Then, in the limit to null infinity, the leading-order low-frequency metric equation [Eq.~\eqref{eq:EOMISTh}] can be written as
    \begin{equation}
    \label{eq:final-proof}
        \Box \delta \hat h^L_{\mu\nu}=-2 \kappa_\text{eff}\,\phantom{}_{\s{(2)}}t_{\mu\nu}[\eta_0,\{\hat h^H,\hat\Psi^H\}]\,,
    \end{equation}
    where $\delta \hat h^L_{\mu\nu}$ satisfies the Lorentz gauge $\partial^\mu\delta \hat h^L_{\mu\nu}=0$, and where $\kappa_\text{eff}=\kappa_0A(\eta_0)$, with $A(\eta_0)$ a function that only depends on the Minkowski background. Furthermore, $\phantom{}_{\s{(2)}}t_{\mu\nu}$ has the following properties:
    \begin{enumerate}
        \item[(a)] it is conserved: $\partial^\mu\phantom{}_{\s{(2)}}t_{\mu\nu}[\eta_0,\{\hat h^H,\hat\Psi^H\}]=0$.
        \item[(b)] it can be written as a sum of terms:
        \begin{equation}   
        \begin{split}
         \phantom{}_{\s{(2)}}t_{\mu\nu}&[\eta_0,\{\hat h^H,\hat\Psi^H\}]\\
         &=\phantom{}_{\s{(2)}}t^{\st{GR}}_{\mu\nu}[\eta_0,\hat h^H]+\sum_\Psi \phantom{}_{\s{(2)}}t^{\hat{\Psi}}_{\mu\nu}[\eta_0,\hat \Psi^H]\,,
        \end{split}
    \end{equation}
     where $\phantom{}_{\s{(2)}}t^{\st{GR}}_{\mu\nu}\propto\Big\langle \partial_\mu \hat h^H_{\alpha\beta}\partial_\nu \hat h^{H\alpha\beta}\Big\rangle$ and where each $\phantom{}_{\s{(2)}}t^{\hat{\Psi}}_{\mu\nu}$ involves two derivative operators.
     \item[(c)] it is invariant under infinitesimal coordinate transformations $x^\mu\rightarrow x '^\mu= x^\mu+\xi_H^{\mu}$.
    \end{enumerate}
\end{theorem}

\begin{proof}
     Let  $S_\text{eff}=\phantom{}_{\s(0)}S_\text{eff}+\phantom{}_{\s(2)}S_\text{eff}$ be the effective action of a dynamical metric theory (see Definition \ref{DefMetricTheory}), where here we have ignored the linear term because it will not contribute to the equations of motion of the perturbations at leading order [Eqs.~\eqref{eq:EOMIISTh} and \eqref{eq:EOMIISPsiTh}], or to the leading-order memory-evolution equation [Eq.~\eqref{eq:EOMISTh}].
   In the following, every equality or proportionality is to be understood as a statement in the limit to null infinity. Moreover, proportionalities are relations up to scalar functions, which may depend on the Minkowski background $\eta_0=\{\eta_{\mu\nu},\Psi_0\}$ at null infinity, as well as the dimensionfull bare Newton's constant $\kappa_0$. 

    With this in mind, let us separately consider the left-hand side and right-hand side of the leading-order, low-frequency equation [Eq.~\eqref{eq:EOMISTh}], starting with the former. The left-hand side of Eq.~\eqref{eq:EOMISTh} arises from $\delta\phantom{}_{\s(0)}S_\text{eff}/\delta \bar g_L^{\mu\nu}$ and it simply corresponds to the metric field equation operator of the theory in terms of the low-frequency fields. Imposing the split in Eq.~\eqref{eq:SmallParam3Th} into a low-frequency background $\bar{\eta}_0=\{\bar{\eta}_{\mu\nu},\bar \Psi_0\}$ that satisfies the vacuum equations and low-frequency perturbations $\delta h^L_{\mu\nu}$ and $\delta\Psi^L$, the operator splits into a homogeneous piece $_{\s{(0)}}\mathcal{G}_{\mu\nu}[\bar{\eta}_0]=0$ and an inhomogeneous part $_{\s{(1)}}\mathcal{G}_{\mu\nu}[\bar{\eta}_0,\{\delta h^L,\delta\Psi^L\}]$ [see Eq.~\eqref{eq:EOMIS2}]. In the limit to null infinity this becomes $_{\s{(1)}}\mathcal{G}_{\mu\nu}[\eta_0,\{\delta h^L,\delta\Psi^L\}]$, since $\bar{g}^L_{\mu\nu}=\eta_{\mu\nu}+\mathcal{O}(1/r)$ and $\Psi=\Psi_0+\mathcal{O}(1/r)$.
    
    Note that this operator $_{\s{(1)}}\mathcal{G}_{\mu\nu}[\eta_0,\{\delta h^L,\delta\Psi^L\}]$ has the same functional form as the operator $_{\s{(1)}}\mathcal{G}_{\mu\nu}[\eta_0,\{\delta h^H,\delta\Psi^H\}]$, which is the operator of the propagation equation [Eq.~\eqref{eq:EOMIISTh}] in the limit to null infinity. Therefore, assumption (iii) also dictates the functional form of the inhomogeneous piece $_{\s{(1)}}\mathcal{G}_{\mu\nu}[\eta_0,\{\delta h^L,\delta\Psi^L\}]$. 
    
    By assumption (iii), the operator takes the form of a wave operator $_{\s{(1)}}\mathcal{G}_{\mu\nu}[\eta_0,\{\delta h^H,\delta\Psi^H\}]\propto\Box \hat{h}^H_{\mu\nu}$ upon a field redefinition $\hat{h}_{\mu\nu}^H=W(h^H,\Psi^H)$ and $\hat{\Psi}^H=V(\Psi^H)$ and after imposing the conditions in Eq.~\eqref{eq:ThLorenzGauge}. Thus, there exists a redefinition $\delta \hat h^L_{\mu\nu}=\hat{W}(\delta h^L_{\mu\nu},\delta \Psi^L)$, together with a suitable choice of $\xi_L^{\mu}$ in a coordinate transformation $x^\mu\rightarrow x^\mu+\xi_L^{\mu}$, such that the remaining leading-order operator reduces to a decoupled wave operator $\Box\delta \hat h^L_{\mu\nu}$.
    The corresponding gauge choice is the Lorenz gauge on the rescaled, low-frequency metric perturbation, $\partial^\mu\delta \hat h^L_{\mu\nu}=0$. The existence of this gauge follows from the conservation of the energy-momentum tensor $\partial^\mu \phantom{}_{\s{(2)}}t_{\mu\nu}=0$, which in turn can be established from property (II) of the average (see Sec.~\ref{sec:IsaacsonGR}). If the total energy-momentum tensor is moreover traceless, it is possible to impose tracelessness on $\delta \hat h^L_{\mu\nu}$, and, in this case, Eq.~\eqref{eq:ProofLHS} follows directly by choosing $\delta \hat h^L_{\mu\nu}=W(\delta h^L_{\mu\nu},\delta \psi^L)$. If the total energy momentum tensor is not trace-less, we can perform a further redefinition to a trace-reversed variable $\delta \hat{\bar{h}}^L_{\mu\nu}\equiv \delta \hat h^L_{\mu\nu}-\frac{1}{2}\eta_{\mu\nu}\delta \hat h^{Lt}$, which yields Eq.~\eqref{eq:ProofLHS} in terms of this new variable.
    Thus, the left-hand side of Eq.~\eqref{eq:EOMISTh} includes the homogeneous equation
    \begin{equation}\label{eq:ProofLHSHom}
    \phantom{}_{\s{(0)}} \mathcal{G}_{\mu\nu}[\eta_0] =0\,,
    \end{equation}
    and an inhomogeneous contribution given by
    \begin{equation}\label{eq:ProofLHS}
    \phantom{}_{\s{(1)}} \mathcal{G}_{\mu\nu}[\eta_0,\{\delta h^L,\delta\Psi^L\}] \propto \Box\delta \hat h^L_{\mu\nu}\,.
    \end{equation}
     
    The right-hand side of Eq.~\eqref{eq:EOMISTh} only contributes to the inhomogeneous equation and is, by definition, the energy-momentum (pseudo)tensor, which in the limit to null infinity reduces to $\phantom{}_{\s{(2)}} t_{\mu \nu} \to t_{\mu\nu}[\eta_0,\{\hat h^H,\hat\Psi^H\}]$. In this limit, we can then write
     \begin{equation}\label{eq:ProofRHS}
        \big\langle \phantom{}_{\s{(2)}}\mathcal{G}_{\mu\nu}\big\rangle=2\kappa_0\,A(\eta_0)\,\phantom{}_{\s{(2)}}t_{\mu\nu}[\eta_0,\{\hat h^H,\hat\Psi^H\}]\,,
    \end{equation}
    where the proportionality can only depend on a scalar function of the Minkowski background $\eta_0$, while dimensional analysis requires the presence of a factor of $\kappa_0$.
    Combining Eqs.~\eqref{eq:ProofLHS} and \eqref{eq:ProofRHS} establishes that the memory-evolution equation can be written as Eq.~\eqref{eq:final-proof}.

The proof of the theorem is finalized if we can establish  properties (a)--(c) of the flat energy-momentum (pseudo)tensor. Property (a), the conservation of the energy-momentum (pseudo)tensor $\partial^\mu\phantom{}_{\s{(2)}}t_{\mu\nu}=0$,  was already touched upon when describing the left-hand side of Eq.~\eqref{eq:EOMISTh}, but let us be more explicit here. This (pseudo)tensor is conserved because it is defined in terms of the average of the variation of the second-order effective action with respect to the background metric [Eq.~\eqref{eq:SecondOrderVariationEMTensor}]. When we take the divergence of this quantity, the derivative commutes with the average symbol and then hits the variation itself, which is a rank-2 (pseudo)tensor. Property (II) of the average procedure (see Sec.~\ref{sec:IsaacsonGR}), however, guarantees that the average of the divergence of a tensor vanishes up to boundary terms, which are higher order in the ratio of the frequency scales. Therefore, the divergence of the (pseudo)tensor also vanishes up to these boundary terms, which establishes property (a). 

Let us now focus on property (b), which tells us that the energy-momentum (pseudo)tensor near null infinity can be decomposed into two terms: a standard GR term and a term that depends on the perturbations of the additional gravitational fields. In order to establish this property, we must first massage the second-order effective action near null infinity. In the limit to null infinity,
    \begin{align}
       0=\phantom{}_{\s{(1)}}\mathcal{G}_{\mu\nu}[\eta_0,\{h^H,\Psi^H\}]&\propto\left[\frac{\delta \phantom{}_{\s(2)}S^{\s{flat}}_\text{eff}}{\delta h^{H\mu\nu}}\right]^H\,,\\
        0=\phantom{}_{\s{(1)}}\mathcal{J}[\eta_0,\{h^H,\Psi^H\}]&\propto\left[\frac{\delta \phantom{}_{\s(2)}S^{\s{flat}}_\text{eff}}{\delta \Psi^{H}}\right]^H\,,
    \end{align}
    and hence, the leading-order propagation equations [Eqs.~\eqref{eq:EOMIISTh} and \eqref{eq:EOMIISPsiTh}] only depend on the $\phantom{}_{\s(2)}S^{\s{flat}}_\text{eff}$ of Eq.~\eqref{eq:flatSA}. 
    Assumption (iii) of the theorem then implies that there exists a set of redefined perturbation variables $\{\hat{h}_{\mu\nu}^H,\hat{\Psi}^H\}$ for which the effective action $\phantom{}_{\s(2)}S^{\s{flat}}_\text{eff}$ can be written as
     \begin{align}\label{ProofActionGenn2nd}
    \phantom{}_{\s(2)}S^{\s{flat}}_\text{eff}\propto&\int\dd^4x\sqrt{-\eta}\,\bigg[\hat{h}^{H\mu\nu}\mathcal{E}^{\alpha\beta}_{\mu\nu}\hat{h}^H_{\alpha\beta}+\sum_{\Psi}\left(\partial \hat\Psi^H\right)^2\bigg]\,.
    \end{align}
    The first term is the usual Fierz-Pauli operator with indices contracted with the independent background metric $\eta_{\mu\nu}$. The second term corresponds schematically to the sum of kinetic terms of the high-frequency perturbations of the additional gravitational fields. 
    Up to integrations by parts, this is the unique action that will lead to a set of decoupled massless wave equations for the leading-order, high-frequency perturbations upon imposing the adequate gauge choices of Eq.~\eqref{eq:ThLorenzGauge}. 
    
    With this in hand, we can now establish property (b).  Lemma \ref{LemmaP} and the perturbed generalized field equations imply that, in the limit to null infinity,
     \begin{equation}\label{eq:RHSP}
    \Big\langle \phantom{}_{\s{(2)}}\mathcal{G}_{\mu\nu}[\eta_0,\{\hat h^H,\hat\Psi^H\}]\Big\rangle \propto \Bigg\langle\frac{\delta \phantom{}_{\s{(2)}}S^{\s{flat}}_\text{eff}}{\delta \eta^{\mu\nu}}\Bigg\rangle \propto \phantom{}_{\s{(2)}}t_{\mu \nu}\,.
     \end{equation}
    In other words, in this limit, the effective stress-energy (pseudo)tensor of radiation, and therefore, also the right-hand side of the memory equation [Eq.~\eqref{eq:EOMISTh}], only requires knowledge of the flat, second-order effective action $\phantom{}_{\s(2)}S^{\s{flat}}_\text{eff}$ as well. Since $\phantom{}_{\s(2)}S^{\s{flat}}_\text{eff}$ can be written in the form of Eq.~\eqref{ProofActionGenn2nd}, it can be split into a term that only depends on the metric perturbation $h^H_{\mu\nu}$ and another set of terms, each of which only depends on the field perturbation $\Psi^H$ of every additional field $\Psi$ and involves only two derivative operators. Each of these terms can be treated separately when computing the radiative energy-momentum (pseudo)tensor, and hence
     \begin{equation}\label{eq:ProofRHS2}
         \phantom{}_{\s{(2)}}t_{\mu\nu}=\phantom{}_{\s{(2)}}t^{\st{GR}}_{\mu\nu}[\eta_0,\hat h^H]+ \sum_\Psi\phantom{}_{\s{(2)}}t^{\hat{\Psi}}_{\mu\nu}[\eta_0,\hat \Psi^H] \,,
    \end{equation}
    where each term only involves two derivative operators.
   Moreover, since the metric perturbation piece in the flat, effective second-order action in Eq.~\eqref{ProofActionGenn2nd} admits the same form as in GR, the first term is proportional to the radiative energy-momentum tensor of GR. After imposing the gauge of Eq.~\eqref{eq:ThLorenzGauge}, the latter reads
    \begin{equation}
        \phantom{}_{\s{(2)}}t^{\st{GR}}_{\mu\nu}\propto\Big\langle \partial_\mu \hat h^H_{\alpha\beta}\partial_\nu \hat h^{H\alpha\beta}\Big\rangle\,,
    \end{equation}
thus establishing property (b). 

    Finally, let us consider property (c), which states that the energy-momentum (pseudo)tensor is invariant under infinitesimal coordinate transformations. The the assumptions on asymptotic flatness of the peacetime, together with the split in Eq.~\eqref{eq:ProofRHS2} and the assumption that $\hat{\Psi}^H$ does not depend on $h^H_{\mu\nu}$,  directly implies that $\phantom{}_{\s{(2)}}t_{\mu\nu}$ is invariant under infinitesimal diffeomorphisms  $x^\mu\rightarrow x^\mu+\xi_H^{\mu}$. This is because the assumption of a local-Lorentz preserving background with $\Psi_0=0$  for fields that are not a scalar and $\partial_\alpha\Psi_0=0$ for scalars implies that the Lie derivative of the first-order perturbations with respect to a vector field generating such a diffeomorphism will vanish. Therefore, only the metric perturbations $h^H_{\mu\nu}$ transform under the gauge shift. Moreover, since the $\hat\Psi^H$ perturbations are independent of $h^H_{\mu\nu}$, the coordinate gauge transformation only affects the term $\phantom{}_{\s{(2)}}t^{\st{GR}}_{\mu\nu}$ , which is invariant upon integration by parts and after imposing the Lorenz gauge Eq.~\eqref{eq:ThLorenzGauge}.

    Since we have established  Eq.~\eqref{eq:final-proof} together with all the required properties (a)--(c) of the radiative energy-momentum tensor $\phantom{}_{\s{(2)}}t_{\mu\nu}$, we have therefore proven the theorem.
\end{proof}

\subsubsection{Technical remarks on the theorem}\label{sec:ThmRemarks}

    Let us conclude the discussion of the theorem by making several technical remarks. First, let us stress that in assumption (iii) we only require the \textit{first-order propagation equations} to reduce to a set of decoupled second-order wave equations. The two main results of the theorem are therefore that (1) the decoupling between fields remains intact even at $\mathcal{O}(\alpha^2)$ at the low-frequency level; (2) the low-frequency $\mathcal{O}(\alpha^2)$ term also only involves two derivative operators. In particular, the latter directly implies that in any theory satisfying the assumptions of the theorem, the memory equation will only directly depend on terms in the action that involve two derivative operators. This is nicely exemplified in the SVT theory result of the radiative energy-momentum tensor in Eq.~\eqref{eq:StressEnergyFourth}.

    Furthermore, observe that assumption (iii) is quite generic. Indeed, in any theory satisfying second-order equations of motion, the first-order propagation equation will only involve two derivative operators. But as we will explore in Sec.~\ref{sec:ScopeOfClaim}, this is even true for a large class of theories, whose field equations are higher-order in derivatives. Local Lorentz invariance as well as masslessness then ensure that the leading-order propagation equation generically takes the form of a massless wave equation.

    Moreover, also a decoupling of the equations at first order in perturbations is quite generic. First, note that a Minkowski background ensures that the tensor, vector and scalar sectors\footnote{The terms "tensor", "vector" and "scalar" refer here to the polarization type of each mode.} can always be decoupled at leading order in perturbations and each of these sectors can therefore be treated individually. Therefore, potential couplings between perturbations at the level of the leading-order perturbation equations could only arise within each of these sectors. We are, however, not aware of any concrete massless theory that admits such a coupling of first-order perturbations. Indeed, explicit examples of such coupled equations typically only arise in theories that include an explicit mass term, such as, for instance, in massive bigravity models (see e.g. \cite{Comelli:2012db,DeFelice:2013nba}) or massive multi-Proca theories \cite{BeltranJimenez:2016afo}. It would, however, be interesting to explore null memory for such theories with coupled leading-order perturbation equations, a task we leave for future work.
   
    Also, any typical massless theory involving multiple interacting vector or scalar fields at the level of the full action, such as non-Abelian vector fields or typical scalar multifield models (see e.g. \cite{Dimakis:2019qfs}), naturally decouple to leading order in perturbations on a Minkowski background, and thus, still abide by the decoupling assumption in (iii).


\subsection{Relation of the memory evolution theorem to the tensor null memory}\label{sec:RelToNullMemTh}

The form of the memory-evolution equation [Eq.~\eqref{eq:final-proof}] for dynamical metric theories (see Definition~\ref{DefMetricTheory}) established through Theorem~\ref{Theorem1}, precisely corresponds to the memory equation [Eq.~\eqref{eq:SecondOrderEOMMem}] obtained for the specific example of SVT gravity, considered in Sec.~\ref{sec:SVTMem}. The only difference is that the energy-momentum (pseudo)tensor of the leading-order radiation $t^{\st{SVT}}_{\mu\nu}$, is replaced by a more general (pseudo)tensor $\phantom{}_{\s(2)}t_{\mu\nu}$, which is still conserved and gauge invariant. The precise form, of course, depends on the nature of the additional gravitational fields $\Psi$.

To make a statement about the resulting tensor null memory, we will now further assume, that in the asymptotic region far from the source, the radiative energy-momentum tensor can be written in the form (cf. Eq.~\eqref{eq:EMTAs})
\begin{equation}\label{eq:EMTAsP}
    \phantom{}_{\s(2)}t_{\mu\nu}=\phantom{}_{\s(2)}t_{00}\,\ell_\mu \ell_\nu\,.
\end{equation}
where $\ell_\mu=-\delta^0_\mu+\delta^i_\mu n_i$ and the scaling with $r$ of $\phantom{}_{\s(2)}t_{00}$ is $\sim r^{-2}$ \footnote{Recall that in the SVT example, this structure followed from the general asymptotic falloff behavior in Eq.~\eqref{eq:OutgoingPlaneWave} of the solutions to the wave equation and the explicit structure of the energy-momentum tensor.}. This is all that is needed to apply the same steps presented in Secs.~\ref{sec:DerTensorMemorySVT} and \ref{sec:SWSHTensorMemorySVT}, in order to write the resulting asymptotic tensor null memory as (see Eq.~\eqref{NonLinDispMemoryModesSVT})
\begin{equation}\label{NonLinDispMemoryGen3_s}
     \delta h_L^{l m}=\frac{2\kappa_\text{eff}}{r} \sqrt{\frac{(l-2)!}{(l+2)!}}\int \dd^2 \Omega'\,\bar{Y}^{l m}\int\dd u' \,r^2\phantom{}_{\s(2)}t_{00}\,,
\end{equation}
where 
\begin{equation}
    \phantom{}_{\s(2)}t_{00}=\phantom{}_{\s{(2)}}t^{\st{GR}}_{00}+ \sum_\Psi\phantom{}_{\s{(2)}}t^{\hat{\Psi}}_{00}\,.
\end{equation}
Here, $\delta h_L^{l m}$ are the coefficients of a SWSH expansion of the quantity $\delta h^{L}=\delta h^{L\,TT}_{ij}m_im_j$, with $\delta h^{L\,TT}_{ij}=\delta \hat h^{L\,TT}_{ij}$. This result follows immediately, because, in the derivation in Secs.~\ref{sec:DerTensorMemorySVT} and \ref{sec:SWSHTensorMemorySVT}, the details of the radiative energy-momentum tensor are irrelevant, since the radiated energy flux $r^2\phantom{}_{\s(2)}t_{00}$ remains untouched after Eq.~\eqref{eq:EMTAsP} is established.

The question is now the following: How general is the assumed structure of the radiative energy-momentum in Eq.~\eqref{eq:EMTAsP}? We will argue that this assumption is not very restrictive, in the sense that it is expected to be generically satisfied by dynamical metric theories of Definition~\ref{DefMetricTheory}, and thus by most of the viable extensions to GR considered in the literature, whose fields are massless and obey local-Lorentz invariance. 

To analyze this question carefully, we will first assume that any of the modes, and hence any of the gauge-invariant solutions to the leading-order wave equations $h$ and $\psi$, scale as (cf. Eq.~\eqref{eq:OutgoingPlaneWave})
\begin{align}\label{eq:OutgoingPlaneWaveTh}
    \left(\hat{h}\,,\;\hat{\psi}\right) \sim \frac{1}{r} \left[f_h(u,\theta,\phi), f_\psi(u,\theta,\phi) \right]\,,
\end{align}
near null infinity. This implies that the GR piece $\phantom{}_{\s{(2)}}t^{\st{GR}}_{\mu\nu}\propto\big\langle \partial_\mu \hat h^H_{\alpha\beta}\partial_\nu \hat h_H^{\alpha\beta}\big\rangle$ satisfies Eq.~\eqref{eq:EMTAsP} (cf. Sec.~\ref{sec:DerTensorMemorySVT}). Using the considerations in Sec.~\ref{sec:IsaacsonBeyondGR}, the additional terms in $\phantom{}_{\s(2)}t^{\hat \Psi}_{\mu\nu}$ can be schematically written as
\begin{equation}
    \phantom{}_{\s(2)}t^{\hat \Psi}_{\mu\nu}\sim \,\eta_0^n \Big\langle \partial \hat\Psi^H\partial \hat\Psi^{H}\Big\rangle\,,
    \label{eq:t2-partial}
\end{equation}
where we have omitted the index structure to keep the expression general, and where $\eta_0^n$ denotes $n$ background-field contributions with indices contracted onto the derivatives or the tensor structure of the gravitational fields themselves. The additional terms in $\phantom{}_{\s(2)}t^{\hat \Psi}_{\mu\nu}$ take this form because the leading-order term in $\langle\phantom{}{\s{(2)}}\mathcal{G}_{\mu\nu}\rangle$ has two derivatives acting on the two leading-order wave perturbations, which, at most, can be multiplied by background fields $\eta_0$ in the limit to null infinity. Moreover, due to Theorem~\ref{Theorem1}, the assumption of decoupled high-frequency propagation equations ensures that there are no mixed terms between the  metric perturbation $\hat{h}^H$ and the other gravitational fields $\hat\Psi^H$ (e.g.~of the form $\langle \partial \hat\Psi^H \; \partial \hat h\rangle$) or between different types of gravitational fields (e.g.~of the form $\langle \partial \hat V^H \; \partial \hat \phi^H\rangle$ if $\Psi$ contains a vector $V$ and  scalar $\phi$). 

We will now partially restore the tensor structure. Requiring that the background preserves local-Lorentz invariance and upon imposing a Lorenz gauge, we can rewrite Eq.~\eqref{eq:t2-partial} as 
\begin{equation}
    \phantom{}_{\s(2)}t^{\hat \Psi}_{\mu\nu}\sim \,\Big\langle \partial_\mu \hat\Psi^H\partial_\nu \hat\Psi^{H}\Big\rangle\,,
    \label{eq:t2-new}
\end{equation}
where we have explicitly included the indices on the partial derivatives and assumed that any tensor structure in the first $\hat\Psi^H$ that is contracted onto the tensor structure of the second $\hat\Psi^H$ is a trace or a scalar. The free indices $\mu$ and $\nu$ need to be placed on the two derivative operators because otherwise the resulting term vanishes either due to the first-order equations of motion or the Lorenz gauge, upon integration by parts. Moreover, here we have removed the n copies of the $\eta_0$ background fields because by assumption, only scalar fields can have nontrivial background values, and therefore there are no indices to be contracted with background fields. Together with the falloff Eq.~\eqref{eq:OutgoingPlaneWaveTh} this implies that Eq.~\eqref{eq:EMTAsP} holds.

To continue, we want to be more specific and consider a concrete type of possible additional gravitational fields, which covers a large class of theories considered in the literature. Namely, from now on, we will focus on dynamical metric theories whose additional gravitational fields $\Psi$ are $k$-form field potentials with an associated Abelian $U(1)$ gauge symmetry.\footnote{Recall that a differential $k$-form field $\Psi$ is a totally antisymmetric tensor field, which in a coordinate-induced basis can be written as $\Psi=(1/k!)\,\Psi_{\mu_1...\mu_k}\,dx^{\mu_1}\wedge...\wedge dx^{\mu_k}$, with $k<d$ and where $"\wedge"$ denotes the exterior product. Such $k$-form fields naturally generalize $U(1)$ vector field potentials because their field strength $\mathscr{F}\equiv d\Psi$ is invariant under Abelian gauge transformations $\Psi\rightarrow \Psi + d\Lambda$, where $\Lambda$ is an arbitrary $(k-1)$-form and $d$ is the exterior derivative. See for instance \cite{Henneaux:1986ht,zee_quantum_2010} for a review of the topic.} The collection of these additional dynamical $k$-form fields are assumed to describe $N$ additional propagating gravitational degrees of freedom. Thus, the theory admits $N$ independent and propagating solutions to the wave equations, characterized through $N$ modes in the canonically normalized, second-order action, which we will denote as $\hat{\psi}_\lambda$, where $\lambda=1,..,N$.  Note that in four spacetime dimensions, we only consider $k$-forms for $k<4$. A $0$-form field simply corresponds to a scalar field, while a 1-form field naturally describes an Abelian vector field, and, therefore, it carries two propagating degrees of freedom. A 2-form field, on the other hand, again only describes one dynamical mode equivalent to a scalar degree of freedom (see e.g. \cite{Heisenberg:2019akx}). A 3-form field will not contain any propagating modes in four dimensions simply because the components of the associate 4-form field strength are constant (see e.g. \cite{Bandos:2019wgy}).\footnote{However, a non-trivial coupling to the metric of such fields can for example lead to a dynamical contribution to the cosmological constant \cite{Duncan:1989ug} and may thus still have physical implications.}

In Appendix~\ref{app:ExampleNullMemoryKForm} we offer the explicit derivation of the associated tensor null-memory formula for this class of theories. Provided that the assumptions of Theorem~\ref{Theorem1} hold, the end result is 
\begin{equation}\label{NonLinDispMemoryGen2}
    \begin{split}
     &\delta h_H^{l m}(u,r)=\frac{1}{r} \sqrt{\frac{(l-2)!}{(l+2)!}}\int_{S^2} \dd^2 \Omega'\,\bar{Y}^{l m}(\Omega') \\
     &\quad\times\,\int_{-\infty}^{u}\dd u'\,r^2\bigg\langle | \dot{\hat{h}}_+|^2+| \dot{\hat{h}}_\times|^2+\sum_{\lambda=1}^N|\dot{\hat{\psi}}_\lambda|^2\bigg\rangle\,,
    \end{split}
\end{equation}
where $\hat{h}_{+,\times}=h_{+,\times}$. Recall that for certain theories (see Sec.~\ref{sec:SVTFirst}) it is necessary to redefine the tensor perturbation variable $\hat{h}_{\mu\nu}^H=W(h^H,\Psi^H)$ to obtain a perturbation variable that satisfies a first-order wave equation in the appropriate gauge. However, the TT component of this redefined variable, and thus, also the polarization modes $h_{+,\times}$ always correspond to the TT component of the physical metric present in the detector response $\hat{h}^{\text{TT}H}_{\mu\nu}=h^{\text{TT}H}_{\mu\nu}$. The same is true for the memory component. The need of such a change of variables to decouple the leading-order equations is, typically, a sign of the presence of additional gravitational polarizations.

Let us end this subsection by stressing that the tensor null-memory result in Eq.~\eqref{NonLinDispMemoryGen2} was obtained without any knowledge of the precise form of the Lagrangian, and it simply follows from Theorem~\ref{Theorem1} and the resulting solution of the memory-evolution equation. The coupling constants of a specific theory would then enter through a transformation from the canonically normalized modes $\hat{\psi}_\lambda$ to the physical modes of the theory [recall the discussion around Eq.~\eqref{eq:StressEnergyFourth}]. The expression in Eq.~\eqref{NonLinDispMemoryGen2} should be compared to the result in Eqs.~\eqref{NonLinDispMemoryModesSVT} and \eqref{eq:StressEnergyFourthHat} and it represents a generalization of the explicit SVT theory example we considered in Sec.~\ref{sec:SVTMem}.

\subsection{Scope of the memory evolution theorem and its link to the tensor null memory}\label{sec:ScopeOfClaim}

We will now explore the scope of the above statements and, in particular, investigate which theories admit a tensor null memory of the form Eq.~\eqref{NonLinDispMemoryGen2}. Such theories must, of course, satisfy the assumptions of Theorem~\ref{Theorem1}. Let us then restrict ourselves to dynamical metric theories of gravity that admit an arbitrary number of additional $k$-form fields in the gravitational sector (see Appendix~\ref{app:ExampleNullMemoryKForm}). In particular, this restriction implies that we focus on theories with Abelian gauge groups, but similar conclusions should also hold in the non-Abelian case. A restriction to $k$-form fields also implies a limitation to bosonic fields. Fermionic fields do not usually play the role of a massless force carrier in known theories. Nonetheless, in principle, it could still be interesting to enrich metric theories with fermionic fields in the gravitational sector, an investigation we leave for future work. 

Given these restrictions, there is a large class of theories for which the tensor null memory takes the form Eq.~\eqref{NonLinDispMemoryGen2}. The theorem clearly encompasses any covariantized version of massless $k$-form Galileon theories \cite{Deffayet:2010zh,Deffayet:2016von}, restricting the full equations of motion to second order. Such theories represent a natural generalization of the SVT class of theories with second-order equations of motion discussed in Sec.~\ref{sec:SVTMem}. As shown explicitly in Sec.~\ref{subsec:examples}, these SVT theories contain Horndeski-type theories, including BD theory~\cite{Brans:1961sx,Dicke:1961gz} and other similar scalar-tensor theories, sGB gravity~\cite{Zwiebach:1985uq,Gross:1986iv}, and double-dual Riemann gravity~\cite{Charmousis:2011ea,Charmousis:2011bf}. 

The memory formula in Eq.~\eqref{NonLinDispMemoryGen2}, however, is not restricted to theories that satisfy second-order equations of motion. A first interesting concrete example of a theory that does not fall under the class of covariantized $k$-form Galileon theories is dCS gravity \cite{Jackiw:2003pm} (see also \cite{Alexander:2009tp}). This theory is defined through the action
\begin{equation}\label{ActiondCS}
\begin{split}
    S^{\st{dCS}}=\int \dd^4 x\sqrt{-g}\bigg(&\frac{1}{16\pi G} R+\frac{\epsilon^2}{4}\,\Theta\,R\du{\mu\nu\rho\sigma} {}\tilde R^{\nu\mu\rho\sigma}\\
    &-\frac{1}{2}g^{\mu\nu}\nabla_\mu\Theta\nabla_\nu\Theta\bigg)\,,
\end{split}
\end{equation}
where $\epsilon$ is a dimensionfull coupling constant and $\Theta$ is a (pseudo)scalar field.\footnote{A vanishing of the scalar field potential implies that $\Theta$ is massless, and the theory is shift symmetric in $\Theta$. This modified theory reduces to GR smoothly when $\Theta$ is constant or when $\epsilon$ vanishes.} This theory is an example of the effective field theories with higher-than-second-order field equations, discussed in Sec.~\ref{sec:IsaacsonBeyondGR}.  In recent years, dCS gravity has received some attention, mainly because its rotating black holes have a nontrivial (pseudo)scalar profile or ``hair'' \cite{Yunes:2009hc}. This implies that when such hairy black holes accelerate, a scalar wave is sourced, which carries energy-momentum to null infinity. 

This fact is what makes dCS interesting from the point of view of the memory because the scalar wave should modify the tensor null memory. Just as in sGB gravity (see Sec.~\ref{sec:sGB}), the dCS action expanded to second-order in perturbations simply reduces to the GR one with a canonical scalar field near null infinity, because the Pontryagin density is of higher order in this limit. Therefore, the \textit{linear-order}, high-frequency propagation equations are just given by two decoupled, second-order wave equations for the metric and the scalar field perturbations. In turn, this fact implies that dCS gravity naturally falls within the scope of Theorem~\ref{Theorem1} and admits the tensor null-memory formula of Eq.~\eqref{NonLinDispMemoryGen2}, with $N=1$, $\Psi\to\Theta$ and corresponding leading-order wave mode $\hat\psi_\lambda\to\hat\vartheta=\vartheta$. Note, however, that, just as in the sGB case, the coupling $\epsilon$ does not enter explicitly into the memory equation, but rather it enters implicitly through the dependence of the metric and scalar perturbations on the coupling. For dCS gravity, this result is indeed confirmed by the explicit computation of the associated BMS balance laws in \cite{hou_gravitational_2022}, from which one can extract the null memory following the same steps outlined in Appendix~\ref{MatchToAsymptoticsBD} for BD gravity.

The above result can be generalized to a wide class of theories with higher-derivative field equations. Indeed, any GR correction involving higher powers of curvature invariants, including nonminimally coupled, massless scalar fields, satisfies decoupled second-order wave equations at the linear level \cite{Stein:2010pn}. In particular, assuming Eq.~\eqref{eq:SmallParam5}, this is even true when the GR corrections involve terms that are quadratic in the curvature and do not form topological invariants (as in the dCS cases). Unlike other GR corrections, such theories do in fact lead to higher-order equations of motion at the linear level in perturbations, when taken at face value. However, in the small-coupling approximation, more precisely, if we assume the strong condition Eq.~\eqref{eq:SmallParam5}, these contributions are suppressed by factors of $\mathcal{O}(\epsilon f_H)$, where $\epsilon f_H\ll1$, and, therefore, they do not appear at leading order. Indeed, it was explicitly shown in \cite{Stein:2010pn} that iteratively solving these higher-derivative equations through order reduction in fact leads again to second-order wave equations for the metric perturbations. All of this then immediately implies that, even for these theories, Theorem~\ref{Theorem1} holds and the tensor null memory takes the form of Eq.~\eqref{NonLinDispMemoryGen2}.

Note that the arguments and conclusions above also naturally encompass metric theories with higher-derivative terms in the action that do not introduce new gravitational fields, as for example systematically explored in \cite{Endlich:2017tqa}. Moreover, the statements can be generalized to theories with nonminimal couplings to other gauge-invariant, Abelian, $k$-form fields. This is because for $k>0$, gauge invariance forces any non-minimal coupling to occur through the field strength, whose background value vanishes when evaluated at null infinity, such that the derivation in \cite{Endlich:2017tqa} is essentially unaltered. The massless and covariant (including local-Lorentz invariant) conditions on the other hand ensure that the leading-order propagation equations reduce to massless wave equations at null infinity. As discussed in Sec.~\ref{sec:ThmRemarks}, even if multiple vector or scalar modes are present in the theory, these leading-order propagation equations typically decouple and the assumptions of Theorem~\ref{Theorem1} are satisfied. Furthermore, just as for the higher-derivative curvature terms, a small-coupling approximation can in principle also be applied to any higher-derivative terms of the additional $k$-form fields that are not tuned to lead to second-order equations of motion. Therefore, it is expected that very generally, Theorem~\ref{Theorem1} can be applied to dynamical metric theory with additional $k$-form fields, which further implies that the tensor null memory is given by  Eq.~\eqref{NonLinDispMemoryGen2}.

In summary, Theorem \ref{Theorem1} suggests that, very generically, for dynamical metric theories of gravity defined in Definition~\ref{DefMetricTheory} in the small-coupling approximation (ensuring a viable EFT expansion), the tensor null memory is of the form of Eq.~\eqref{NonLinDispMemoryGen2}. In other words, null memory is modified in comparison to GR in two ways: (I) through contributions of the energy fluxes to null infinity of additional, massless, dynamical degrees of freedom in the theory; and (II) through modifications in the generation and propagation of the leading-order tensor perturbations. This simple result could have interesting implications, as we will explore in the discussion below. Clearly then, the theorem does not cover metric theories (a) with additional \textit{massive} degrees of freedom or a massive physical metric, such as in massive gravity~\cite{deRham:2014zqa}, (b) with nonlocal effects, (c) with explicit Lorentz breaking in the gravitational sector, such as in Einstein-\AE{}ther theories~\cite{Jacobson:2004ts}, and (d) with nonvanishing torsion or nonmetricity, such as in torsion bigravity~\cite{Hayashi:1979wj,Hayashi:1980av}.



\section{Discussion and Conclusion}\label{sec:Discussion}

Planned space-based detectors, such as LISA~\cite{LISA}, or next-generation ground-based detectors, such as the Einstein Telescope~\cite{Punturo:2010zz} or Cosmic Explorer~\cite{Reitze:2019iox} are expected to provide the first detections of the tensor memory effect in the near future. These observations of gravitational wave memory may play an important role in establishing a better understanding of gravity and constraining modifications of GR. This is because tensor memory is a very special, nonlinear correction to the gravitational wave response. Indeed, the memory's dominant null component can be understood as being sourced by the leading-order wave front itself, allowing for potential consistency checks between the standard oscillatory part of the response and the memory itself. The results of the work we have presented, including the computation of the null memory in the most general, massless, and gauge-invariant SVT theory with second-order equations of motion, together with its subsequent generalizations, is but one of the first steps toward the possible future exploitation of gravitational wave memory as a test of GR.

One of the main discoveries we have presented in this paper is a simple but very generic conclusion [see Eq.~\eqref{NonLinDispMemoryGen2}]: The functional form of the tensor null memory of massless dynamical metric theories of gravity is merely modified from the GR expectation through the presence of additional null fluxes due to other gravitational fields present in the beyond-GR extension. This result could potentially be exploited to develop a largely model-insensitive test of GR of perhaps the most straightforward manifestation of new physics: the existence of additional scalar, vectorial or tensorial propagating degrees of freedom. Since memory is sensitive to any kind of energy-momentum emitted from the source, such a test would not only complement ongoing searches for additional gravitational polarizations \cite{Chatziioannou:2012rf}, but also extend the sensitivity of such tests to scenarios in which additional gravitational fields do not excite any other polarization modes of the physical metric.

Much work remains to be done to explicitly construct such a model-insensitive null-memory test of GR, but let us suggest some general directions that could be explored. One direction for future research is to study whether the oscillatory part of the gravitational wave signal can be separated from the nonoscillatory part. This could be accomplished by exploiting the fact that, by definition of the memory-evolution equation in the Isaacson picture, the characteristic frequency of the oscillatory component of the signal is parametrically separated from the frequency of the null-memory part. Therefore,
Morlet wavelets \cite{Cornish:2014kda} in a time-frequency analysis may be able to separate these two components. Based on such a model-independent analysis, one could imagine developing a consistency check between the extracted oscillatory part of the gravitational wave signal and the nonoscillatory memory contribution. Such a test would require knowledge of the inclination angle and the extraction of both polarizations, which could, for example, be obtained through an electromagnetic counterpart and the detection of the signal with multiple interferometers. On a very general basis, a potential mismatch between the two signals could then be interpreted as the presence of additional propagating (scalar, vector or tensor) degrees of freedom.  

Another possible approach would be to target as well the merger sensitivity of the memory signal directly, by developing a parameterization for the beyond-GR memory effect, akin to the parameterized post-Einsteinian framework used for the oscillatory signal~\cite{Yunes:2009ke}. Such an approach would require the construction of a parametric post-Einsteinian model for the GR deformation to the null memory through simple transition functions, such as those studied in \cite{Yunes:2006mx}. Based on a set of examples of the null memory in specific beyond-GR theories computed through Eq.~\eqref{NonLinDispMemoryGen2}, such a parametric model could then be chosen to ensure that the various specific beyond-GR modifications to the memory can be recovered by the model. This would in turn allow for a Bayesian parameter estimation study of the full signal to attempt to extract the parameters of the post-Einsteinian tensor memory from the signal.  

On the theoretical side, a very interesting future task will be to generalize our study to ordinary memory \cite{Zeldovich:1974gvh,Braginsky:1987gvh,Garfinkle:2022dnm} defined colloquially as the part of the memory that is sourced by fields that do not reach null infinity. The typical example of such a source are unbound components of the matter stress-energy of the system itself, which, due to compact support, are typically assumed to not reach null infinity. In theories beyond GR, however, the possibility of additional massive gravitational fields in the theory represents an additional source of ordinary memory. In particular, considering massive fields would allow for the study of a much richer structure of Ostrogradsky-stable vector field theories \cite{Heisenberg:2014rta}, as well as a higher-derivative massive 2-form interaction \cite{Heisenberg:2019akx}, which can be related to the topological mass generation of BF theories \cite{Almeida:2020lsn}. 

Moreover, we would like to mention, that a definition of the memory-evolution equation as the leading-order low-frequency equation of motion in the Isaacson picture, potentially allows for an even broader study of memory beyond GR. Indeed, in principle, the Isaacson approach could also be applied to the field equations of theories which we explicitly disregarded in this work, such as theories with nondynamical field content or explicit Lorentz breaking. In such a general setting, one could also think about solving the resulting memory-evolution equation not in the vicinity of null infinity, but in a different appropriate limit. 

This last point could potentially be explored for the study of the tensor memory in metric theories with broken diffeomorphism invariance, such as massive gravity theories \cite{deRham:2014zqa}. Such an analysis may, however, turn out to be less straightforward, as the current understanding of memory is heavily based on the structure of the theory around null infinity. In the case of massive gravity, other calculation techniques, such as point-particle source approximations \cite{garfinkle_memory_2017}, might be more appropriate for the computation of its memory contribution. For example, \cite{kilicarslan_graviton_2019} applied such approximations to massive Fierz-Pauli theory, but the analysis in this work may need to be generalized to avoid an unphysical discrepancy with GR in the small mass limit.

Another avenue for future work would be to investigate a possible generalization of the BMS balance laws to a wide class of theories, based on the results obtained in this paper. In light of the close connection between BMS balance laws and tensor null memory established in Appendix~\ref{MatchToAsymptoticsBD} for BD theory, our results strongly suggest that the approach of \cite{hou_gravitational_2021,tahura_brans-dicke_2021,hou_conserved_2021,hou_gravitational_2021_2} for BD gravity may be readily generalized to the asymptotic structure of any dynamical metric theories. It would be interesting to explore this conjecture in detail, especially with regard to the ordinary memory, which we have excluded in our analysis. As concerns scalar or vector null memory, the connection with BMS balance laws is less clear. In our work, have not found any nonzero scalar or vector memory, which parallels the observation that no analogous BMS constraint can be formulated for scalar memory within BD theory \cite{hou_gravitational_2021,tahura_brans-dicke_2021,hou_conserved_2021,hou_gravitational_2021_2}. Recently, however, \cite{Seraj:2021qja} suggested that scalar memory in BD theory could be associated with asymptotic symmetries of a dual description of the scalar sector instead, the implications of which still remain unexplored.

\begin{acknowledgments}
We would like to express our gratitude to Luc Blanchet for pointing out the relation in Eq.~\eqref{eq:RelBlanchet} and sharing his personal notes of the proof thereof. Moreover, we thank Lydia Bieri, David Garfinkle and Henri Inchausp\'{e} for useful discussions on the memory effect, Justin Ripley for helpful comments on gravitational radiation in sGB theories and recent progress in numerical relativity, as well as Fabio D'Ambrosio, Shaun Fell, Francesco Gozzini, David Maibach and Stefan Zentarra for valuable conversations on asymptotic symmetries and BMS flux blanace laws.
LH is supported by funding from the European Research Council (ERC) under the European Unions Horizon
2020 research and innovation programme grant agreement No 801781 and by the Swiss National Science Foundation
grant 179740. LH further acknowledges support from the Deutsche Forschungsgemeinschaft (DFG, German Research Foundation) under Germany's Excellence Strategy EXC 2181/1 - 390900948 (the Heidelberg STRUCTURES Excellence Cluster).
NY acknowledges support from the Simmons Foundation through Award No. 896696. JZ is supported by an ETH Z\"{u}rich Doc.Mobility Fellowship.

\end{acknowledgments}

\printnoidxglossary



\appendix

\section{Transverse Traceless Multipole Expansion}\label{SWSH}

For convenience, we gather in this appendix a collection of definitions and formulas from the unifying review by Thorne \cite{Thorne:1980ru} for different multipole expansions used in this work. In particular, we will touch upon spin-weighted spherical harmonics, pure-spin transverse-traceless tensor harmonics, as well as STF tensors, all tied to the $SO(3)$ rotation group and the irreducible representations thereof. The notation for STF tensors will, however, be adapted to a more contemporary custom (see for instance \cite{poisson2014gravity,Blanchet:2013haa,maggiore2008gravitational}). For more details and derivations, we refer the reader to \cite{Thorne:1980ru}, while part of the treatment can also be found in \cite{maggiore2008gravitational} and in the Appendix of \cite{Nichols:2017rqr}.

\subsection{Spin-weighted spherical harmonics}\label{Ap:SWSH}

Spin-weighted spherical harmonics can be constructed from ordinary spherical harmonics through applications of angular derivative operators
\begin{equation}
 \phantom{}_{\s{s}}Y^{lm}=
\begin{cases}
			\sqrt{\frac{(l-s)!}{(l+s)!}}\eth^s Y^{lm}\,, & l\geq s\geq 0\\
            (-1)^s \sqrt{\frac{(l+s)!}{(l-s)!}}\bar{\eth}^{-s} Y^{lm}\,, & 0> s\geq -l\\
            0\,, & |s|>l\\
\end{cases}
\end{equation}
 where the angular derivative operator $\eth$ and its complex conjugate $\bar{\eth}$ are defined through their action on functions $f_s$ of spin-weight $s$. The application of the operator $\eth$ on a function $f_s$ defines a quantity with spin-weight $s+1$, while $\bar{\eth}$ lowers the spin-weight by one unit.
\begin{align}\label{eth}
	\eth f_s &\equiv -\sin^s\theta\left(\partial_\theta+i\,\csc\theta\right)(f_s \sin^{-s}\theta)\,,\\
	\bar\eth f_s &\equiv -\sin^{-s}\theta(\partial_\theta-i\,\csc\theta)\left(f_s \sin^{s}\theta\right)\,.
\end{align}
They satisfy
\begin{align}
    (-1)^{s+m}\, \phantom{}_{\s{-s}}\bar{Y}^{l-m}(\theta,\phi)&=\phantom{}_{\s s}Y^{lm}(\theta,\phi)\,,\label{CCSWSH}\\
    \int_{S^2}\dd\Omega \,_{\s s}Y^{lm}(\theta,\phi) _{\s s}\bar{Y}^{l'm'}(\theta,\phi)&=\delta_{ll'}\delta_{mm'}\,,\label{OrthogonalitySWSH}\\
    \sqrt{(l-s)(l+s+1)}\,_{\s{s+1}}Y^{lm}&=\eth\,_{\s s}Y^{lm}\,,\label{eq:ASWSHid1}\\
     -\sqrt{(l+s)(l-s+1)}\,_{\s{s-1}}Y^{lm}&=\bar{\eth}\,_{\s s}Y^{lm}\,.
\end{align}
\begin{equation}\label{SWSHTrippleInt}
\begin{split}
    &\int_{S^2}\dd^2\Omega\,_{\s{s_1}}Y^{l_1m_1}\,_{\s{s_2} }Y^{l_2m_2}\,_{\s{s_3}}Y^{l_3m_3}\\
    &\,=\,\sqrt{\frac{\prod_{i=1}^3(2l_i+1)}{4\pi}} \footnotesize{
\begin{pmatrix}
l_1 & l_2 & l_3\\
m_1 & m_2 & m_3
\end{pmatrix}
\begin{pmatrix}
l_1 & l_2 & l_3\\
-s_1 & -s_2 & -s_3
\end{pmatrix}}\,,
\end{split}
\end{equation}
if $s_1+s_2+s_3=0$.

\subsection{Pure spin transverse-traceless tensor harmonics}\label{Ap:TT}

The electric- and magnetic-parity pure-spin TT harmonics can also be constructed from ordinary spherical harmonics through
\begin{align}
    T^{\scriptscriptstyle{E2}\,\scriptstyle{l m
    }}_{ij}&\equiv r^2\sqrt{\frac{2(l-2)!}{(l+2)!}}\,\perp_{ijab}\partial_a\partial_b\,Y^{lm}\,,\\
    T^{\scriptscriptstyle{B2}\,\scriptstyle{l m
    }}_{ij}&\equiv r\sqrt{\frac{2(l-2)!}{(l+2)!}}\,\perp_{ijab}\epsilon_{cd(a}\partial_{b)}r\,n_d\partial_c\,Y^{lm}\,.
\end{align}
where $n_i$ is a unit radial vector and the TT projector $\perp_{ijab}$ is defined in Eq.~\eqref{eq:Projectors}.
They satisfy
\begin{align}
    (-1)^m\bar{T}^{\scriptscriptstyle{P}\,\scriptstyle{l -m
    }}_{ij}(\theta,\phi)&=T^{\scriptscriptstyle{P}\,\scriptstyle{l m
    }}_{ij}(\theta,\phi)\,,\\
    \int_{S^2}\dd\Omega \,T^{\scriptscriptstyle{P}\,\scriptstyle{l m
    }}_{ij}(\theta,\phi)\,T^{\scriptscriptstyle{P'}\,\scriptstyle{l' m'
    }}_{ij}(\theta,\phi)&=\delta_{PP'}\delta_{ll'}\delta_{mm'}\,.
\end{align}

\subsection{STF tensors}\label{Ap:STF}

Denoting a multi-index of $l$ spatial indices collectively as $L\equiv i_1i_2...i_l$, the projection of a tensor component $A_L$ onto its STF part will be referred to as $A_{\langle L\rangle}$. Note that $B_{L}A_{\langle L\rangle}=B_{\langle L\rangle}A_{\langle L\rangle}$, where a sum over $l$ is implicit. Moreover, $n_L$ stands for the tensor product of $l$ radial vectors $n_i$. An expansion in terms of STF tensors $n_{\langle L\rangle}=n_{\langle L\rangle}(\theta,\phi)$ is equivalent to an expansion in spherical harmonics
\begin{align}\label{eq:AYlmToNL}
    n_{\langle L\rangle}&=\frac{4\pi\,l!}{(2l+1)!!}\sum_{m=-l}^l\,\mathcal{Y}_{\langle L\rangle}^{lm}\, Y^{lm}\,,\\
    Y^{lm}&=\mathcal{Y}_L^{lm}\,n_L=\mathcal{Y}_L^{lm}\,n_{\langle L\rangle}\,,
\end{align}
where $\mathcal{Y}_L^{lm}$ are angle-independent, STF tensors connecting the two basis. They satisfy
\begin{align}
    (-1)^m\bar{\mathcal{Y}}^{l-m}_L&=\mathcal{Y}^{lm}_L\,\\
    \bar{\mathcal{Y}}_L^{lm}\mathcal{Y}_L^{lm'}&=\frac{(2l+1)!!}{4\pi\,l!}\,\delta_{mm'}\,.
\end{align}


\subsection{Relations between harmonics}
The pure-spin TT harmonic tensors are related to spin weight $s=\pm 2$ spin-weighted spherical harmonics through
\begin{align}
    T^{\s{E2}\,\s{l m
    }}_{ij}&=\frac{1}{\sqrt{2}}\left(\phantom{}_{\s{-2}}Y^{lm}\,m_im_j+\phantom{}_{\s{2}}Y^{lm}\,\bar{m}_i\bar{m}_j\right)\,,\\
    T^{\s{B2}\,\s{l m
    }}_{ij}&=-\frac{i}{\sqrt{2}}\left(\phantom{}_{\s{-2}}Y^{lm}\,m_im_j-\phantom{}_{\s{2}}Y^{lm}\,\bar{m}_i\bar{m}_j\right)\,,
\end{align}
while the relation to the STF basis is given by
\begin{align}
    T^{\s{E2}\,\s{l m}}_{ij}&=\perp_{ijab}\sqrt{\frac{2(l-1)l}{(l+1)(l+2)}}\,\mathcal{Y}_{abL-2}^{l m}\,n_{L-2}\,,\\
    T^{\s{B2}\,\s{l m}}_{ij}&=\perp_{ijab}\sqrt{\frac{2(l-1)l}{(l+1)(l+2)}}\,\epsilon_{cd(a}\,\mathcal{Y}_{b)dL-2}^{l m}\,n_{cL-2}\,.
\end{align}


\subsection{Expansion of a rank-2, TT tensor}
Using the definitions above, we will now relate the expansions of an arbitrary, rank-2, TT tensor $H_{ij}^\text{TT}(u,r,\theta,\phi)$, in terms of spin-weighted spherical harmonics, pure-spin TT tensor harmonics, and STF tensors. 


\subsubsection{Spin-weighted spherical harmonic expansion}\label{Ap:SWSHExpansion}
Using the transverse null vector of spin-weight $s=-1$
\begin{equation}
    \bar m^i\equiv\frac{1}{\sqrt{2}}(u^i-iv^i)\,,
\end{equation}
we can define the spinweight $s=-2$ scalar quantity
\begin{equation}\label{eq:ADefHs2}
    H\equiv H_{ij}^\text{TT}\, \bar m^i\, \bar m^j=H_{ij}\, \bar m^i\, \bar m^j\,,
\end{equation}
which we expand in terms of spin-weighted spherical harmonics as
\begin{equation}\label{eq:AExpansionHlm}
    H=\sum_{l=2}^\infty\sum_{m=-l}^l\,H^{lm}\, \phantom{}_{\scriptscriptstyle-2}Y^{lm}\,,
\end{equation}
where
\begin{align}
    H^{lm}(u,r)=\int_{S^2}\dd^2\Omega\,\phantom{}_{\scriptscriptstyle-2}\bar{Y}^{lm}\, H\,.
\end{align}
\subsubsection{TT tensor expansion}
The general expansion in terms of TT tensor harmonics reads
\begin{align}
   H_{ij}^\text{TT}=\sum_{l=2}^\infty\sum_{m=-l}^l \left[ U^{l m}\,T^{\scriptscriptstyle{E2}\,\scriptstyle{l m
    }}_{ij}+V^{l m}\,T^{\scriptscriptstyle{B2}\,\scriptstyle{l m
    }}_{ij}\right]\,,
\end{align}
where 
\begin{align}
    U^{l m}(u,r)&=\int_{S^2}\dd^2\Omega \,T^{\s{E2}\,\scriptstyle{l m
    }}_{ij}\,H^\text{TT}_{ij}\,,\label{eq:AUlmCalc}\\
    V^{l m}(u,r)&=\int_{S^2}\dd^2\Omega \,T^{\s{B2}\,\s{l m
    }}_{ij}\,H^\text{TT}_{ij}\,.\label{eq:AVlmCalc}\\
\end{align}
The expansion coefficients $U^{l m}$ and $V^{l m}$ are the electric- and magnetic-parity multipole moments, respectively, also commonly known as mass and current multipole moments. In the literature, it is common to explicitly factor out the $r$ dependence of mass and current multipole moments, which we will not do here. Since $H^\text{TT}_{ij}$ is real, they satisfy
\begin{align}
    \bar{U}^{l m}=(-1)^m U^{l -m}\,,\quad \bar{V}^{l m}=(-1)^m V^{l -m}\,.
\end{align}
\subsubsection{STF expansion}
The corresponding STF multipole expansion is given by
\begin{align}\label{eq:AExpansionULVL}
\begin{split}
   H_{ij}^\text{TT} =4\perp_{ijab}\sum_{l=2}^\infty &\frac{1}{l !}\Big[\,U_{klL-2}\,n_{L-2}\\
   &+\frac{2l}{l+1}\epsilon_{cd(a}\, V_{b)cL-2}\,n_{dL-2}\Big]\,,
\end{split}
\end{align}
where
\begin{align}
    U_{ijL-2}&=\frac{1}{4\pi}\frac{l(l-1)(2l+1)!!}{2(l+1)(l+2)}\int\dd^2\Omega\, n_{L-2}\,H_{ij}^\text{TT} \,,\\
    V_{ijL-2}&=\frac{1}{4\pi}\frac{(l-1)(2l+1)!!}{4(l+2)}\int\dd^2\Omega\, n_{aL-2}\,\epsilon_{iab}H_{bj}^\text{TT} \,.
\end{align}

\subsubsection{Relations between expansion coefficients}
The multipole moments in a spin-weighted basis are related to the electric- and magnetic-parity moments as
\begin{equation}\label{eq:AHlmToUlmVlm}
    H^{lm}=\frac{1}{\sqrt{2}}\left[U^{lm}-iV^{lm}\right]\,,
\end{equation}
which can be inverted to give
\begin{align}
    U^{lm}&=\frac{1}{\sqrt{2}}\left[H^{lm}+(-1)^m \,\bar{H}^{l-m}\right]\,,\label{eq:AUVlmToHlm}\\
    V^{lm}&=\frac{i}{\sqrt{2}}\left[H^{lm}-(-1)^m \,\bar{H}^{l-m}\right] \,.
\end{align}
On the other hand, the relation between a STF and a TT tensor basis is given by
\begin{align}
  U^{l m}&=\frac{16\pi}{(2l+1)!!}\sqrt{\frac{(l+1)(l+2)}{2(l-1)l}}\,\bar{\mathcal{Y}}_{L}^{l m}\,U_L\,,\\
  V^{l m}&=\frac{-32\pi l}{(l+1)(2l+1)!!}\sqrt{\frac{(l+1)(l+2)}{2(l-1)l}}\,\bar{\mathcal{Y}}_{L}^{l m}\,V_L\,,\label{eq:AUVlmToULVL}\\
    U_L&=\frac{l!}{4}\sqrt{\frac{2(l-1)l}{(l+1)(l+2)}}\sum_{m=-l}^l\, \mathcal{Y}_L^{l m}\,U^{l m}\,,\\
    V_L&=\frac{-(l+1)!}{8l}\sqrt{\frac{2(l-1)l}{(l+1)(l+2)}}\sum_{m=-l}^l\, \mathcal{Y}_L^{l m}\,V^{l m}\,.
\end{align}

\section{Proof of Eq.~\eqref{eq:RelBlanchet}}\label{DerivationEq}

In this appendix, we provide a derivation of the relation
\begin{equation}\label{eq:RelBlanchetA}
   \left[\frac{n'_in'_j}{1-\vec{n}'\cdot\vec{n}}\right]^\text{TT}=\sum_{l=2}^\infty \frac{2(2l+1)!!}{(l+2)!}\left[n_{L-2}n'_{\langle ijL-2\rangle}\right]^\text{TT}\,,
\end{equation}
where a superscript TT implies a contraction of free indices with the TT projector with respect to $\vec{n}$, which we defined in Eq.~\eqref{eq:Projectors}. The proof below is based on private communications with Blanchet and can also be found in \cite{BlanchetPaper}.

\begin{proof} First, note that due to the TT projection, the left-hand side of Eq.~\eqref{eq:RelBlanchetA} vanishes for $\vec{n}'=\vec{n}$, while this is also true for the right-hand side. Thus, in the following, we consider the case $\vec{n}'\neq\vec{n}$, which implies $|\vec{n}'\cdot\vec{n}|<1$ such that we can expand the left-hand side in terms of a geometric series
    \begin{align}
        \frac{n'_in'_j}{1-\vec{n}'\cdot\vec{n}}&=n'_in'_j\sum_{l=2}^\infty(\vec{n}'\cdot\vec{n})^{l-2}=\sum_{l=2}^\infty n_{L-2}n'_{ijL-2}\,,
    \end{align}
where we use the notation introduced in Sec.~\ref{Ap:STF}. The remaining task is to rewrite $n'_{ijL-2}$ in terms of its STF part $n'_{\langle ijL-2\rangle}$, by using the formula [see e.g. Eq.~(A21a) in \cite{Blanchet:1985sp}]
    \begin{equation}
        n_L=\sum_{k=0}^{[l/2]}a_k^l\delta_{\{2K}n'_{\langle L-2K\rangle\}}\,,
    \end{equation}
where 
    \begin{equation}
        a_k^l\equiv\frac{(2l-4k+1)!!}{(2l-2k+1)!!}\,,
    \end{equation}
and where $[l/2]$ selects the integer part of $l/2$. Moreover, the operator $\{...\}$ on tensor indices $A_{\{L\}}$ denotes the sum $\sum_{\sigma\in S} A_{i_{\sigma(1)}...i_{\sigma(l)}}$, where $S$ is the smallest set of permutations of $(1...l)$, which makes $A_{\{L\}}$ fully symmetric in $L$. For instance,  \begin{equation}\label{eq:ExampleExpansion}
    \begin{split}
        \delta_{\{ab}n'_{ij\}}=&\,\delta_{ab}n'_{ij}+\delta_{ai}n'_{bj}+\delta_{aj}n'_{bi}+\delta_{bi}n'_{aj}\\
        &+\delta_{bj}n'_{ai}+\delta_{ij}n'_{ab}\,.
    \end{split}
    \end{equation}

With this in hand,
    \begin{align}
        \frac{n'_in'_j}{1-\vec{n}'\cdot\vec{n}}=\sum_{l=2}^\infty n_{L-2}\sum_{k=0}^{[\frac{l-2}{2}]}a_k^l\delta_{\{2K}n'_{\langle ijL-2-2K\rangle\}}\,.
    \end{align}
Notice, however, that due to the TT projection with respect to $\vec{n}$ on the free indices $ij$ and the contraction of the remaining indices with $n_{L-2}$, any term which involves one of the indices $i$ or $j$ within the Kronecker delta, hence terms containing $\delta_{ij}$, $\delta_{ia}$ or $\delta_{ja}$ for any index $a$, will vanish. For example, for the term $l=4$ and $k=1$, written out in Eq.~\eqref{eq:ExampleExpansion}, only the first term will survive. Thus
    \begin{align}
        &\left[\frac{n'_in'_j}{1-\vec{n}'\cdot\vec{n}}\right]^\text{TT}=\sum_{l=2}^\infty n_{L-2}\sum_{k=0}^{[\frac{l-2}{2}]}\left[\delta_{2K}n'_{\langle ij L-2-2K\rangle}\right]^\text{TT}\notag\\
        &\quad\qquad\qquad\qquad\qquad\qquad\qquad\;\;\times a_k^l\,b_k^l\,,\notag\\
        &=\sum_{l=2}^\infty\sum_{k=0}^{[\frac{l-2}{2}]}a_k^l\,b_k^l\left[n_{L-2-2K}n'_{\langle ij L-2-2K\rangle}\right]^\text{TT}\,,
    \end{align}
where 
\begin{equation}
    b_k^l\equiv\frac{(l-2)!}{2^kk!(l-2-2k)!}\,,
\end{equation}
is the number of terms within the sum $\delta_{\{2K}n'_{L-2-2K\}}$.

We can now rearrange the sum over positive integers $l$ and $k$ as
\begin{align}
    \sum_{l=2}^\infty\sum_{k=0}^{[\frac{l-2}{2}]}&=\sum_{l,k}
    [2\leq l\leq \infty] [0\leq k\leq \frac{l-2}{2}]\notag\\
    &=\sum_{l,k}
    [0\leq k\leq \frac{l-2}{2}\leq \infty]\notag\\
    &=\sum_{p,k}
    [0\leq k\leq \frac{p+2k-2}{2}\leq\infty]\notag\\
    &=\sum_{p,k}[2\leq p\leq \infty][0\leq k\leq \infty]\,,
\end{align}
where we defined $p\equiv l-2k$ and, in this context, the angular brackets $[B]$ of a Boolean expression denote Inverson brackets,\footnote{Inverson Brackets are defined as $[B]=\begin{cases}1\,&\text{if $B$ is true}\\0\,,&\text{otherwise}\end{cases}$.} such that
\begin{align}
       \left[\frac{n'_in'_j}{1-\vec{n}'\cdot\vec{n}}\right]^\text{TT}&=\sum_{p=2}^{\infty}\sum_{k=0}^\infty a_k^{p+2k}\,b_k^{p+2k}\left[n_{P-2}n'_{\langle ij P-2\rangle}\right]^\text{TT}\notag\\
        &=\sum_{l=2}^{\infty}\underset{\equiv S_l}{\underbrace{\sum_{k=0}^\infty a_k^{l+2k}\,b_k^{l+2k}}}\left[n_{L-2}n'_{\langle ij L-2\rangle}\right]^\text{TT}\,,
\end{align}
where in the last step we simply relabeled $p\rightarrow l$. The sum over $k$ can indeed be resummed explicitly to give
\begin{align}
    S_l&=\sum_{k=0}^\infty \frac{(2l+1)!!}{(2l+2k+1)!!}\, \frac{(l+2k-2)!}{2^kk!(l-2)!}\notag\\
    &=\frac{2^{2+l}\Gamma(l+\frac{3}{2})}{\sqrt{\pi}\Gamma(l+3)}\notag\\
    &=\frac{2(2l+1)!!}{(l+2)!}\,,
\end{align}
where in the last step we used the identities of gamma functions $\Gamma(l+3)=(l+2)!$ and $\Gamma(l+\frac{3}{2})=\frac{\sqrt{\pi}(2l+1)!!}{2^{1+l}}$. Comparing to Eq.~\eqref{eq:RelBlanchetA}, this concludes the proof.
\end{proof}

\section{Displacement Memory from BMS Balance Laws in Brans-Dicke Gravity}\label{MatchToAsymptoticsBD}

Starting from the action in Eq.~\eqref{ActionBD}, the authors in  \cite{hou_gravitational_2021,tahura_brans-dicke_2021,hou_conserved_2021,hou_gravitational_2021_2} arrive at the BMS supermomentum flux-balance law in spherical coordinates $\{u,r,x^A\}$, $x^A=\{\theta,\phi\}$ [see for instance Eqs.~(10)--(12) in \cite{hou_gravitational_2021_2}, from which we also adopt the notation]
\begin{equation}\label{eq:BMSBalanceLawI}
\begin{split}
   0=\frac{\varphi_0}{4\pi G}&\int_{S^2}\dd^2\Omega\,\alpha \Bigg\{\Delta M+\int_{-\infty}^\infty\dd u'\Bigg[\frac{1}{2}N_{AB}N^{AB}\\
   &+\mathcal{D}_A\mathcal{D}_BN^{AB}+(2\omega+3)\left(\frac{N}{\varphi_0}\right)^2\Bigg]\Bigg\}\,,
\end{split}
\end{equation}
where $\mathcal D_A$ is the covariant derivative on $S^2$, $\alpha=\alpha(\theta,\phi)$ is an arbitrary function on $S^2$ parametrizing supertranslations, and $M$ denotes the Bondi mass aspect, while $\Delta M\equiv M(u\rightarrow\infty)-M(u\rightarrow-\infty)$. Moreover, as in \cite{hou_gravitational_2021_2} we write
\begin{equation}\label{eq:BondiNews}
     N_{AB}\equiv -\dot{\hat{c}}_{AB}\,, \quad N\equiv \dot\varphi_1\,,
\end{equation}
where $\hat{c}_{AB}$ is the symmetric, traceless and transverse shear tensor and $\varphi_1$ the component of the scalar that falls off as $\sim 1/r$.
We additionally expand the shear as
\begin{equation}\label{eq:ShearExpand}
    \hat{c}_{AB}= \bar{c} \,\bar m_A\bar m_B+ c\, m_A m_B\,,
\end{equation}
where in spherical coordinates
\begin{equation}
    m=\frac{1}{\sqrt{2}}\left(\partial_\theta+i\sin\theta\partial_\phi\right)\,.
\end{equation}
Using Eqs.~\eqref{eq:BondiNews} and \eqref{eq:ShearExpand}, as well as the definition of the angular derivative operator [Eq.~\eqref{eth}], which implies that we have $\mathcal{D}_A\mathcal{D}_B \hat{c}^{AB}=\frac{1}{2}\left(\eth^2 c+\bar{\eth}^2\bar{c}\right)$, the flux-balance law in Eq.~\eqref{eq:BMSBalanceLawI} can be rewritten as
\begin{align}\label{BMSSupermomentumBalanceLawBD}
\begin{split}
  \int_{S^2}\dd^2\Omega\,\alpha\, \Delta M=&\frac{1}{4}\int_{-\infty}^\infty\dd u'\int_{S^2}\dd^2\Omega\,\alpha\Bigg[\,|\dot{c}|^2\\
  &-\Re\eth^2\dot{c}+(2\omega+3)\left(\frac{\dot\varphi_1}{\varphi_0}\right)^2\Bigg]\,.
\end{split}
\end{align}

To single out the tensor null memory from the above relation, we can first set the subdominant, left-hand side to zero. This gives rise to the so-called linear or ordinary memory. We then rewrite the BMS supermomentum balance law in Eq.~\eqref{BMSSupermomentumBalanceLawBD} by moving the second to last term to the left, while carrying out the $u'$ integral to obtain
\begin{align}\label{BMSSupermomentumBalanceLawBD2}
\begin{split}
  \int_{S^2}\dd^2\Omega\,\alpha\Re\eth^2\Delta{c}=\int_{-\infty}^\infty&\dd u'\int_{S^2}\dd^2\Omega\,\alpha\Bigg[\,|\dot{c}|^2\\
  &+(2\omega+3)\left(\frac{\dot\varphi_1}{\varphi_0}\right)^2\Bigg]\,.
\end{split}
\end{align}
Furthermore, expanding $c$ with spin-weight $s=-2$ on the left-hand side as
\begin{equation}
   c(u,\theta,\phi)=\sum_{l=2}^\infty \sum_{m=-l}^l c_{lm}(u)\,_{\scriptscriptstyle-2}Y^{lm}(\theta,\phi)\,,
\end{equation}
and using the relation [Eq.~\eqref{eq:ASWSHid1}] (or equivalently integrating $\eth^2$ by parts),
\begin{equation}
    \eth^2\phantom{}_{\scriptscriptstyle-2}Y^{lm}=\sqrt{\frac{(l+2)!}{(l-2)!}}Y^{lm}
\end{equation}
as well as choosing $\alpha(\theta,\phi)=\bar{Y}^{lm}(\theta,\phi)$, we obtain
\begin{align}\label{BMSSupermomentumBalanceLawBD3}
\begin{split}
  \frac{1}{2}\big(\Delta c^{lm}+(-1)^m&\Delta\bar{c}^{l-m}\big)=\sqrt{\frac{(l-2)!}{(l+2)!}}\int_{S^2}\dd^2\Omega\,\bar{Y}^{lm}\\
  &\times \int_{-\infty}^\infty\dd u'\left[\,|\dot{c}|^2+(2\omega+3)\left(\frac{\dot\varphi_1}{\varphi_0}\right)^2\right]\,.
\end{split}
\end{align}
Note that by separating $c^{lm}$ into its electric- and magnetic-parity moments of a spin-2, TT tensor harmonic expansion of the rank-2 TT tensor $c_{ij}=e_i^Ae_j^B \hat c_{AB}$ \eqref{eq:AHlmToUlmVlm}
\begin{equation}\label{eq:clmToUlmVlm}
    c^{lm}=\frac{1}{\sqrt{2}}\left[U_c^{lm}-iV_c^{lm}\right]\,,
\end{equation}
the left-hand side in Eq.~\eqref{BMSSupermomentumBalanceLawBD3} precisely corresponds to the electric-parity part [Eq.~\eqref{eq:AUVlmToHlm}]. Therefore, finally we arrive at
\begin{align}\label{BMSSupermomentumBalanceLawBD4}
\begin{split}
  \Delta U_c^{lm}=&\sqrt{\frac{2(l-2)!}{(l+2)!}}\int_{S^2}\dd^2\Omega\,\bar{Y}^{lm}\\
  &\times \int_{-\infty}^\infty\dd u'\left[\,|\dot{c}|^2+(2\omega+3)\left(\frac{\dot\varphi_1}{\varphi_0}\right)^2\right]\,.
\end{split}
\end{align}

From the balance laws, we can therefore single out the total tensor displacement memory, and hence, the lasting nonzero component within the electric-parity multipole of the shear $c$, which ultimately induces a lasting offset in the detector strain. Note, however, that the shear and the scalar field which enter the balance laws are the total shear and scalar field at $\mathcal{O}(r^{-1})$ within the full nonlinear theory, and therefore, in particular, they already contain all possible memory contributions.\footnote{However, the generalized Bondi News, which enters the right-hand side of the equation, vanishes as $u\rightarrow \pm\infty$, where by assumption no gravitational waves reach null infinity.} In order to connect this result with the computation in Eq.~\eqref{sourceTensorMemoryBD} within the setup of the present work, we should therefore also distinguish between a high- and low-frequency part of the shear and the scalar
\begin{equation}\label{BMSSupermomentumBalanceLawBD5pre}
    c=c^L+c^H\,,\quad \varphi_1= \varphi_1^L+\varphi_1^H\,,
\end{equation}
and slightly reinterpret the flux-balance law in order to use it as a tool to compute the low-frequency displacement memory characterized by the measurable monotonically increasing and nonoscillatory, time-dependent raise of the memory, which is what gravitational wave detectors are sensitive to. We therefore only gradually integrate over retarded time, while extracting the low-frequency part of the expression by averaging out the small scales. Note that in order to compute the full memory, such an averaging is irrelevant.

After averaging, any cross terms of the form ``$c^Lc^H$'' or ``$\varphi_1^L\varphi_1^H$'' on the right-hand side in Eq.~\eqref{BMSSupermomentumBalanceLawBD4} will vanish. Moreover, we assume that we can neglect any contribution of low-frequency components ``$c^Lc^L$'' or ``$\varphi_1^L\varphi_1^L$'' which can be interpreted as the 
``memory of the memory''. In other words, we assume that the source modes for the memory themselves have a negligible memory component, which is indeed a reasonable assumption \cite{Talbot:2018sgr}. Furthermore imposing $c(u\rightarrow-\infty)=0$, we therefore have
\begin{align}\label{BMSSupermomentumBalanceLawBD5}
  c_L^{lm}(u)=&\sqrt{\frac{(l-2)!}{(l+2)!}}\int_{S^2}\dd^2\Omega\,\bar{Y}^{lm}\\
  &\times\int_{-\infty}^u\dd u'\Bigg\langle\,|\dot{c}^H|^2+(2\omega+3)\left(\frac{\dot\varphi^H_1}{\varphi_0}\right)^2\Bigg\rangle\,,\notag
\end{align}
where $c_L^{lm}(u)$ is the resulting low-frequency correction to the shear, given high-frequency radiation modes $c^H$ and $\varphi_1^H$. Here we have used Eq.~\eqref{eq:clmToUlmVlm} with $\delta V_c^{lm}=0$ to rewrite Eq.~\eqref{BMSSupermomentumBalanceLawBD5pre} in terms of the shear.

The reason why such a reinterpretation of the balance law is useful in practice, is because typical numerical relativity waveforms do not capture any memory component due to various technical reasons (see e.g. \cite{Favata:2008yd}). Only recently were people able to compute memory in numerical relativity based on Cauchy-characteristic extraction \cite{Mitman:2020pbt}.

As a last step before finally being able to compare results, we need to connect the perturbative shear and scalar field defined here with the perturbations used in the main text and ensure that these are indeed the same quantities. In the case of BD theory, the easiest way to establish this correspondence is to compare the corresponding leading $\mathcal{O}\left(r^{-1}\right)$ terms of the electric part of the Riemann tensor. In \cite{hou_gravitational_2021,tahura_brans-dicke_2021}, these terms were computed and found to be [see e.g. Eq.~(2.44) in \cite{hou_gravitational_2021}]\footnote{Note, however, that the authors in \cite{hou_gravitational_2021,tahura_brans-dicke_2021} report the result in an orthonormal tetrad basis, instead of the spherical coordinates employed here.}
\begin{equation}
    R_{uAuB}=-\frac{1}{2r}\left(\ddot{\hat{c}}_{AB}-\gamma_{AB}\frac{\ddot{\varphi}_1}{\varphi_0}\right)+\mathcal{O}\left(\frac{1}{r^2}\right)\,.
\end{equation}
By using the embedding of the unit $S^2$ basis $e^{i}_A=\frac{\partial n^i}{\partial x^A}$, such that $\gamma_{AB}=\delta^{ij}e^{i}_Ae^{j}_B=2m_{(A}\bar{m}_{B)}$, we can convert the leading-order expression to a $\{t,x,y,z\}$ Minkowski basis, which yields
\begin{align}\label{RiemannAsymptotic}
    R_{0i0j}&=e_{i}^Ae_{j}^BR_{uAuB}=-\frac{e_{i}^Ae_{j}^B}{2r}\left(\ddot{\hat{c}}_{AB}-\gamma_{AB}\frac{\ddot{\varphi}_1}{\varphi_0}\right)\notag\\
    &=-\frac{1}{2r}\bigg(\underset{= \frac{1}{2}e^+_{ij}\,(\ddot{c}+\ddot{\bar{c}})+\frac{i}{2}e^+_{ij}\,(\ddot{c}-\ddot{\bar{c}})}{\underbrace{m_im_j\,\ddot{c}+\bar m_i\bar m_j\,\ddot{\bar{c}}}}-(\delta_{ij}-n_in_j)\frac{\ddot{\varphi}_1}{\varphi_0}\bigg)\notag\\
    &=-\frac{1}{2r}\left(e^+_{ij}\,\ddot c_++e^\times_{ij}\,\ddot c_\times-(\delta_{ij}-n_in_j)\frac{\ddot{\varphi}_1}{\varphi_0}\right)\,,\notag
\end{align}
where we used that
$e_{i}^Ae_{j}^B\, \gamma_{AB}=2m_{(i}\bar{m}_{j)}=\delta_{ij}-n_in_j$, $e_{i}^Am_A=m_i$ and we defined $ c_+\equiv \Re c$ and $c_\times \equiv -\Im c $.
Hence, comparing to Eq.~\eqref{ElectricRiemannLin} with $\sigma=1/\varphi_0$ we obtain the correspondence
\begin{align}
    c^H_+(u,\Omega)&=\lim_{r\rightarrow\infty}rh_+(u,r,\Omega)\,,\\
    c^H_\times(u,\Omega)&=\lim_{r\rightarrow\infty}rh_\times(u,r,\Omega)\,,\\
   \varphi^H_1(u,\Omega)&=\lim_{r\rightarrow\infty}r\varphi(u,r,\Omega)\,,
\end{align}
while therefore as well
\begin{equation}
   c_L^{lm}=\lim_{r\rightarrow\infty}r\delta h^{lm}\,,
\end{equation}
such that Eq.~\eqref{BMSSupermomentumBalanceLawBD5} indeed corresponds to the result in Eq.~\eqref{NonLinDispMemoryModesBD}.


\section{Tensor Null Memory in a Gauge-Invariant $k$-Form Metric Theory}\label{app:ExampleNullMemoryKForm}

In order to exemplify how generic the assumption given by Eq.~\eqref{eq:EMTAsP} is, we present in this appendix the general form of the tensor null memory of a particular, but fairly general class of additional gravitational fields. Namely, we consider an arbitrary number of additional dynamical $k$-form connection fields that we schematically denote as $\Psi$, with an associated, Abelian, gauge symmetry for $k>1$ (for $k=0$, the field simply corresponds to a scalar field). The action in Eq.~\eqref{eq:ActionMetricTheory} is thus constructed out of curvature invariants of the metric, as well as field strengths $\mathscr{F}\equiv d\Psi$ that are invariant under Abelian gauge transformations
\begin{equation}
    \Psi\rightarrow \Psi + d\Lambda\,,
\end{equation}
where $\Lambda$ are arbitrary $(k-1)$-forms and $d$ is the exterior derivative. These dynamical $k$-form fields are assumed to describe $N$ additional propagating gravitational degrees of freedom.
Note that in a local chart-induced basis, a $k$-form field reads
\begin{equation}
\begin{split}
    \Psi&=\frac{1}{k!}\Psi_{\mu_1\mu_2...\mu_k}\,dx^{\mu_1}\wedge dx^{\mu_2}\wedge...\wedge dx^{\mu_k}\,,\\
    &\equiv\frac{1}{k!}\Psi_{K}\,\dd X^{K}\,.
\end{split}
\end{equation}

The assumptions of Theorem~\ref{Theorem1} are now natural. In particular, the leading-order propagation equations [Eq.~\eqref{eq:EOMIISPsiTh} of the theory reduce to a set of decoupled wave equations for the leading-order waves $\Psi^H_K$ upon imposing the analog of the Lorentz gauge 
\begin{equation}\label{eq:ThLorenzGaugeKForm}
\partial^{\mu_i}\hat{\Psi}_K^H=0\,,
\end{equation}
on each of the field perturbations. This is because, by assumption, the theory is covariant, locally Lorentz invariant and massless. Moreover, any term in the action that would lead to higher-order derivatives in the equations of motion of the leading-order waves does not contribute to leading order in the small-coupling approximation [Eq.~\eqref{eq:SmallParam4}]. An additional assumption here is that, in the limit to null infinity, there is no coupling between different fields at the level of the leading-order, high-frequency equation.

Following the proof of Theorem~\ref{Theorem1}, the flat, effective, second-order action of Eq.~\eqref{ProofActionGenn2nd} can thus be written as
\begin{align}\label{ProofActionGenn2ndKForm}
\begin{split}
    \phantom{}_{\s(2)}S^{\s{flat}}_\text{eff}=&\frac{-1}{2\kappa_\text{eff}}\int\dd^4x\,\sqrt{-\eta}\bigg[\hat{h}^{\mu\nu}_H\mathcal{E}^{\alpha\beta}_{\mu\nu}\hat{h}^H_{\alpha\beta} \\
    &+\sum_{\Psi}\frac{1}{2q}\eta^{\mu_1\nu_1}...\eta^{\mu_{q}\nu_{q}}\hat{\mathscr{F}}^H_{\mu_1...\mu_{q}}\hat{\mathscr{F}}^H_{\nu_1...\nu_{q}}\bigg]\,,
\end{split}
\end{align}
where $q\equiv k+1$, $\kappa_\text{eff}=\kappa_0A(\eta_0)$ and $\hat{\mathscr{F}}^H=d\hat{\Psi}^H$ are the field strength of the canonically normalized $k$-form perturbations $\hat{\Psi}^H$. The first term in Eq.~\eqref{ProofActionGenn2ndKForm} is again the usual Fierz-Pauli operator with indices contracted with the independent background metric $\eta_{\mu\nu}$, while the second term corresponds to the sum of kinetic terms of all the additional gravitational, gauge-invariant, $k$-form field perturbations.

With this explicit form of the flat effective action at hand, we can now explicitly compute the resulting effective stress-energy tensor for the additional leading-order waves (see also \cite{Navarro:2012hv} for a study of the properties of the energy-momentum tensor of $k$-forms)
    \begin{align}
    \begin{split}
        \phantom{}_{\s(2)}t_{\mu\nu}^{\hat{\Psi}}=\frac{1}{2\kappa_\text{eff}}\sum_{\Psi}\bigg\langle&\hat{\mathscr{F}}^H_{\mu K}\hat{\mathscr{F}}_{\nu}^{HK}-\frac{1}{2q}\eta_{\mu\nu}\hat{\mathscr{F}}^H_{Q}\hat{\mathscr{F}}^{HQ}\bigg\rangle\,.
     \end{split}
    \end{align}
Imposing the equations of motion [Eq.~\eqref{eq:ThWaveEq}], as well as the Lorenz gauge [Eq.~\eqref{eq:ThLorenzGaugeKForm}], while recalling that the averaging allows for integrations by parts, the only surviving terms are
    \begin{align}\label{eq:IndexStP}
        \phantom{}_{\s(2)}t_{\mu\nu}^{\hat{\Psi}}=\frac{1}{2\kappa_\text{eff}}\sum_{\Psi}\bigg\langle\partial_\mu\hat{\Psi}^H_K\partial_\nu\hat{\Psi}^{KK}\bigg\rangle\,.
    \end{align}
Including the Fierz-Pauli result, the total energy-momentum tensor therefore reads
\begin{align}\label{eq:EMTotKForm}
        \phantom{}_{\s(2)}t_{\mu\nu}=\frac{1}{2\kappa_\text{eff}}\bigg\langle&\frac{1}{2}\partial_\mu\hat{h}^H_{\alpha\beta}\partial_\nu\hat{h}^{H\alpha\beta}+\sum_{\Psi}\partial_\mu\hat{\Psi}^H_K\partial_\nu\hat{\Psi}^{HK}\bigg\rangle\,.
    \end{align}
This (pseudo)energy-momentum tensor is indeed conserved (and traceless), as well as gauge invariant (both under linearized diffeomorphisms and gauge transformations of the $k$-form fields, upon imposing a Lorenz gauge).

The above also means that the (pseudo)energy-momentum tensor can be written in terms of the radiative modes of each field, which are the solutions to the wave equations that we can expand in terms of a polarization basis. For the metric perturbation, we have the usual TT tensor modes of the physical metric
    \begin{equation}
        \hat{h}^{H\text{TT}}_{\mu\nu}=h^{H\text{TT}}_{\mu\nu}=\epsilon^+_{\mu\nu} \,\hat h_++\epsilon^\times_{\mu\nu} \,\hat h_\times\,,
    \end{equation}
where
    \begin{equation}
        \epsilon^+_{\mu\nu} \epsilon^{+\mu\nu}=\epsilon^\times_{\mu\nu} \epsilon^{\times\mu\nu}=2\,,\quad\epsilon^+_{\mu\nu} \epsilon^{\times\mu\nu}=0\,.
    \end{equation}
For the additional gravitational fields, $\Psi$, by assumption they describe $N$ additional radiative modes $\psi_\lambda$, where $\lambda=1,..,N$. These can be written in some polarization basis as\footnote{A practical way of completely fixing the gauge is the use of spacetime light-cone coordinates, as exemplified for instance in \cite{Heisenberg:2019akx} for 2-forms. These polarizations should, however, not be confused with the gravitational polarizations of the physical metric defined in Sec.~\ref{GravPol}.} 
    \begin{equation}
        \hat{\Psi}^p_K=\sum_{\text{P}}\epsilon^\text{P}_K \,\hat\psi_\text{P}\,,\quad\text{where}\quad
        \epsilon^\text{P}_K \epsilon^{\text{P}'K}=\delta^{\text{PP}'}\,.
    \end{equation}
    
These solutions to the wave equation in the limit to null infinity take the form (cf. Eq.~\eqref{eq:OutgoingPlaneWave})
    \begin{align}\label{eq:OutgoingPlaneWaveKForm}
    &\left(\hat{h}_+\,,\;\hat{h}_\times\,,\;\hat{\psi}_\lambda\right) \sim \frac{1}{r} \left[f_+(u,\theta,\phi),f_\times(u,\theta,\phi), f_\lambda(u,\theta,\phi) \right]\,.
    \end{align}
    for some real functions $f_{+,\times,\lambda}$. Together with the form of the energy-momentum tensor Eq.~\eqref{eq:EMTotKForm} this behavior near null infinity implies that the radiative energy-momentum tensor can indeed be written as (see Sec.~\ref{sec:DerTensorMemorySVT})
    \begin{equation}\label{eq:StressEnergyGen}
        \phantom{}_{\s(2)}t_{\mu\nu}=\frac{1}{2\kappa_\text{eff}}\bigg\langle | \dot{\hat{h}}_+|^2+| \dot{\hat{h}}_\times|^2+\sum_{\lambda=1}^N|\dot{\hat{\psi}}_\lambda|^2\bigg\rangle \,\ell_\mu \ell_\nu\,,
    \end{equation}
    where $\ell_\mu=-\nabla_\mu t+\nabla_\mu r$ and the falloffs of Eq.~\eqref{eq:OutgoingPlaneWaveKForm} impose the scaling $t_{\mu\nu}\sim r^{-2}$.

    Following Sec.~\ref{sec:RelToNullMemTh}, the tensor null-memory formula therefore takes the form
\begin{equation}\label{NonLinDispMemoryGen2App}
    \begin{split}
     &\delta h_H^{l m}(u,r)=\frac{1}{r} \sqrt{\frac{(l-2)!}{(l+2)!}}\int_{S^2} \dd^2 \Omega'\,\bar{Y}^{l m}(\Omega') \\
     &\quad\times\,\int_{-\infty}^{u}\dd u'\,r^2\bigg\langle | \dot{\hat{h}}_+|^2+| \dot{\hat{h}}_\times|^2+\sum_{\lambda=1}^N|\dot{\hat{\psi}}_\lambda|^2\bigg\rangle\,.
    \end{split}
\end{equation}
Note that the $\hat{\psi}$ here are the canonically normalized modes. The coupling constants of a specific theory will then enter the memory formula by transforming back to the physical modes.

The example presented here could be readily generalized by the inclusion of non-Abelian, 1-form gauge fields with a simple and compact but otherwise arbitrary gauge group and field strength $d\mathrm{F}=d\Psi+\Psi\wedge \Psi$. Such theories are a natural generalization of the SVT theory considered in this work, but would not change the result significantly, as the background solution requires a vanishing vector field in this case. Moreover, just as for massless Abelian fields, no higher-order self-interaction terms exist in $d=4$, which would still lead to second-order equations of motion \cite{Deffayet:2010zh}. Similarly, in the same way as Abelian 1-forms can be generalized to non-Abelian gauge groups, $k$-forms have non-Abelian generalizations, which typically require the use of gerbes (see e.g. \cite{Strobl:2016aph,Breen:2001ie}). We conjecture, that also in such a general framework an adapted version of the memory equation [Eq.~\eqref{NonLinDispMemoryGen2}] should hold.



\newpage
\bibliographystyle{utcaps}
\bibliography{references}

\clearpage

\end{document}